\documentclass[a4paper,twocolumn,11pt,accepted=2025-10-16]{quantumarticle}
\pdfoutput=1
\usepackage{tikz,tikz-3dplot}
\usetikzlibrary{arrows,decorations.pathreplacing, calligraphy}

\usepackage{graphicx}
\usepackage{dcolumn}
\usepackage{bm}
\usepackage{framed} 
\usepackage{wrapfig}
\usepackage{boxedminipage}
\usepackage{hyperref}
\usepackage[numbers,sort&compress]{natbib}

\usepackage{amsmath}
\usepackage{amsfonts}
\usepackage{amsthm}
\usepackage[capitalise]{cleveref}
\usepackage{algorithm}
\usepackage{algorithmic}
\usepackage{subcaption}
\usepackage{boxedminipage}
\usetikzlibrary{shapes.misc,calc}

\usepackage{xcolor}

\newtheorem{theorem}{Theorem}[section]

\newtheorem{definition}{Definition}

\newtheorem{lemma}[theorem]{Lemma}
\newtheorem{proposition}[theorem]{Proposition}

\definecolor{quantum}{rgb}{0.2,0.2,1}
\definecolor{laser}{rgb}{0.5,0.5,0}
\definecolor{microwave}{rgb}{1,0,0}

\newcommand{\ket}[1]{\left\vert#1\right\rangle}
\newcommand{\bra}[1]{\left\langle#1\right\vert}

\newcommand\blfootnote[1]{%
    \begingroup
    \renewcommand\thefootnote{}\footnote{#1}%
    \addtocounter{footnote}{-1}%
    \endgroup
}

\newcommand{\braket}[2]{\left\langle\left. #1 \right\vert #2 \right\rangle}
\newcommand{\ketbra}[2]{\left\vert \left. #1 \right\rangle \kern-0.2em \left\langle #2 \right .\right\vert}

\newcommand{\ceil}[1]{\left\lceil #1\right\rceil}
\newcommand{\N}{\mathbb{N}}

\newcommand{\Hil}{\mathcal{H}}
\newcommand{\floor}[1]{\left\lfloor #1\right\rfloor}
\newcommand{\todo}[1]{}

\newcommand{\op}[2]{|#1\rangle\langle #2 |}
\newcommand{\lnorm}[1]{\left\lVert #1 \right\rVert}

\usepackage{bm}
\newtheorem{problem}{Problem}
\crefname{problem}{Problem}{Problems}
\Crefname{Problem}{Problem}{Problem}

\renewcommand{\vec}[1]{\bm{#1}}

\renewcommand{\selectlanguage}[1]{}

\usepackage{xcolor}

\begin{document}

\title{Quantum random access memory: a survey and critique}%

\author{Samuel Jaques}
 \email{sejaques@uwaterloo.ca}
 \affiliation{Department of Combinatorics and Optimization, University of Waterloo, 200 University Ave W, Waterloo, Canada}
 \affiliation{Institute for Quantum computing, University of Waterloo, 200 University Ave W, Waterloo Canada}%

 \author{Arthur G. Rattew}
 \email{arthur.rattew@materials.ox.ac.uk}
 \affiliation{Department of Materials, University of Oxford, Parks Road, Oxford OX1 3PH, United Kingdom}

\begin{abstract}
Quantum random-access memory (QRAM) is a mechanism to access data (quantum or classical) based on addresses which are themselves a quantum state. QRAM has a long and controversial history, and here we survey and expand arguments and constructions for and against. 

We use two primary categories of QRAM from the literature: (1) active, which requires external intervention and control for each QRAM query (e.g. the error-corrected circuit model), and (2) passive, which requires no external input or energy once the query is initiated. In the active model, there is a powerful opportunity cost argument: in many applications, one could repurpose the control hardware for the qubits in the QRAM (or the qubits themselves) to run an extremely parallel classical algorithm to achieve the same results just as fast. We apply these arguments in detail to quantum linear algebra and prove that most asymptotic quantum advantage disappears with active QRAM systems, with some nuance related to the architectural assumptions. 

Escaping the constraints of active QRAM requires ballistic computation with passive memory, which creates an array of dubious physical assumptions, which we examine in detail. Considering these details, in everything we could find, all non-circuit QRAM proposals fall short in one aspect or another. 

In summary, we conclude that \emph{cheap, asymptotically scalable} passive QRAM is unlikely with existing proposals, due to fundamental obstacles that we highlight. These obstacles are deeply rooted in the requirements of QRAM, but are not provably inevitable; we hope that our results will help guide research into QRAM technologies that circumvent or mitigate these obstacles. Finally, circuit-based QRAM still helps in many applications, and so we additionally provide a survey of state-of-the-art techniques as a resource for algorithm designers using QRAM.
\end{abstract}

\maketitle

\section{Introduction}
Computers are machines to process data, and so access to data is a critical function. Quantum computers are no exception. Specialized memory-access devices (e.g., RAM) are now commonplace in classical computers because they are so useful. It thus seems intuitive that future quantum computers will have similar devices. 

Unfortunately for quantum computers, it's not clear that the analogy holds. Memory access, classical or quantum, requires a number of gates that grows proportional to the memory size. For a classical computer, manufacturing gates is generally a fixed cost and we care more about the runtime (though this assumption starts to break down with large scale, high-performance computing). In contrast, gates in most quantum technologies are active interventions, requiring energy or computational power to enact. 

The foundational work on QRAM, \cite{PRL:GioLLoMac08}, attempts to break from this model and imagines passive components to enact the memory, such that after a signal is sent into the device, it propagates without external control (i.e. ballistically) to complete the memory access. This would be more efficient, but might not be realistic.

QRAM has sparked significant controversy and yet numerous applications presume its existence. Because of this, results are scattered. Circuit-based QRAM, where one accepts a high gate cost, can be still be useful, so optimized techniques appear as intermediate results in subject-specific works. Here we attempt to survey and collect all such results. 

Similarly, there is a broad literature of attempts to construct QRAM (in theory or experiment) and many published criticisms of it. This paper aims to unify the existing criticisms and provide new ones, and consider all existing QRAM proposals in light of these criticisms. 

\subsection{Summary}\label{sec:summary}
To frame this paper, we first summarize existing arguments for QRAM and try to explain, conceptually, how to think about the nature of a QRAM device. 

\begin{figure*}
\resizebox{\textwidth}{!}{
\includegraphics{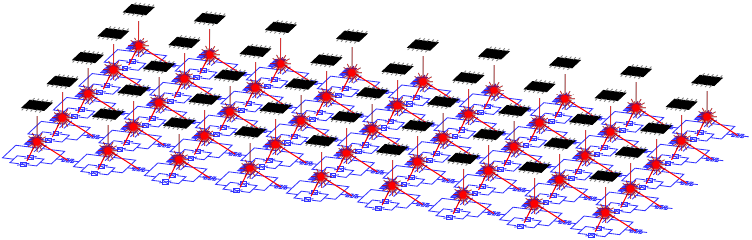}
}
\caption{An illustration of a quantum computer in the memory peripheral model~\cite{CRYPTO:JaqSch19}. Classical controllers are the chips on the top, which send signals (in red) to qubits (in blue).}\label{fig:qc-diagram}
\end{figure*}

Classically we are used to ``gates'' being static physical components that data propagates through, just as depicted in circuit diagrams. Applying that intuition to quantum computing is unjustified. In contrast, the memory peripheral framework of \cite{CRYPTO:JaqSch19} models quantum computing like \Cref{fig:qc-diagram}, where data (in the form of qubits) are static physical components, and gates are operations that a controller applies to the data. In this framework a quantum computer is a physical object (e.g., some Hilbert space) called a \emph{memory peripheral} that evolves independently under some Hamiltonian. There is also a \emph{memory controller} that can choose to intervene on the object, either by applying a quantum channel from some defined set, or modifying the Hamiltonian. This more accurately reflects today's superconducting and trapped-ion quantum computers, as well as surface code architectures of the future where classical co-processors will be necessary to handle error syndrome data, possibly in addition to scheduling and applying gates. 

While future devices may drastically reduce the overhead of each operation, there will still be some level of computational power necessary to manage each intervention in the device. As this paper will focus primarily on asymptotics, this distinction will matter less; e.g., a laser's microcontroller is not much different than any other application-specific integrated circuit.

Overall, this framework highlights a view of a quantum computer as a peripheral of a classical device. We imagine a quantum computer as a collection of quantum devices with many associated classical computers built into the architecture whose purpose is to manage the quantum devices.

The crux of QRAM is that if a query state is a superposition over all addresses, then for the device or circuit to respond appropriately, it must perform a memory access over \emph{all} addresses simultaneously. While ``it performs all possible computations at once'' is a classic misunderstanding of quantum computing, it applies more directly to QRAM: imagining the memory laid out in space, a QRAM access must transfer some information to each of the bits in memory if it hopes to correctly perform a superposition of memory accesses. This sets up an immediate contrast to classical memory, which can instead adaptively direct a signal through different parts of a memory circuit.

In the memory peripheral framework, we need to apply gates for every possible pattern of memory access, and since each gate requires some action from the controller -- be it time, computation, or energy -- the relevant cost is proportional to the total size of the memory. The problems with high time or energy costs are obvious, but \Cref{fig:qc-diagram} highlights that a computational cost is really an \emph{opportunity} cost. If we have a dedicated co-processor for every $O(1)$ qubits, could they do something besides run the quantum computer?

This captures arguments against particular uses of QRAM. If we are allowed to reprogram the control circuitry, then:
\begin{itemize}
\item
two-to-one collision finding is just as fast for the quantum computer or its controller~\cite{SHARCS:Bernstein09}
\item
claw finding is \emph{faster} for the controller than the quantum computer~\cite{CRYPTO:JaqSch19}
\item
the quantum advantage for matrix inversion drops from exponential to (at best) polynomial according to \cite{NP:Aaronson2015, ProcRoySoc:CHIP+18} (see \Cref{sec:quantum-linear-algebra} for more nuance)
\end{itemize} 

In short, a massive quantum computer implies a massive classical co-processor to control it, and that co-processor can quickly solve many problems on its own. 

This holds for what \cite{ProcRoySoc:CHIP+18} calls ``active'' QRAM. Proponents of QRAM can then argue that this is the wrong way to conceptualize QRAM. In classical computers, RAM is a separate device, manufactured differently than the CPU, designed specifically to perform memory accesses. From the CPU's perspective, it sends a signal into the RAM and waits for the RAM to process that signal on its own. An analogous quantum device -- what \cite{NP:Aaronson2015} denotes as ``passive'' QRAM -- would not incur the costs of external control, but would need incredible engineering to function correctly without that control. 

The memory peripheral framework also allows more complicated Hamiltonians and longer evolution to capture this situation. This reflects ``ballistic'' computation, where some environment is prepared and then allowed to propagate on its own. Today's boson sampling experiments are more like this: the classical control (e.g. motors or humans) must establish the optical layout and then inject a photon, but after this setup, computation will proceed without any further intervention.

The main technical contributions of this paper are a careful analysis of the requirements and assumptions of this other perspective. In short, we argue that for passive QRAM:
\begin{itemize}
\item
If the QRAM requires no intervention, it is ``ballistic'', and must be described by a single time-independent Hamiltonian;
\item 
No intervention means no active error correction, and hence each physical component must intrinsically have near-zero noise levels;
\item
Unless the device is carefully designed to limit error propagation, the physical component error rates must be problematically low;
\item
To be compatible with the rest of a fault tolerant computation, we need some way to apply the QRAM, which (in the passive case) acts only on physical qubits, to the logical qubits in an error corrected code;
\item
In particular, a QRAM gate cannot be teleported, so we must resort to other techniques which risk decohering the state.
\end{itemize}

In many ways the comparison to classical RAM is inaccurate, and passive QRAM faces daunting engineering challenges. In our opinion, these challenges are probably insurmountable, in which case quantum memory would always need to be active, and thus access to $N$ bits of QRAM would always have a total (parallellizable) cost proportional to $N$ or higher.

We stress that QRAM can still be a useful tool: in numerous applications in chemistry and quantum arithmetic, such as those in \Cref{sec:applications}, QRAM can offload expensive computations to a classical co-processor, providing a benefit that exceeds the overall access cost.

\subsection{Outline}
We first define our models of quantum computation and the definitions of QRAM in \Cref{sec:definitions}. There we summarize the different names and notions of this device from the literature, to unify the discussion.

\Cref{sec:applications} discusses the main uses of QRAM, and the requirements that each application places on the QRAM device, and the scale they require.

\Cref{sec:quantum-linear-algebra} analyses the implications of various QRAM models for the application of quantum linear algebra. We prove that in many regimes, an active QRAM cannot have quantum advantage over a parallel classical algorithm for a certain general linear algebra tasks.

Then in \Cref{sec:circuit-qram} we give state-of-the-art examples of \emph{circuit} QRAM from the literature. These are circuits to perform QRAM that are intended to work in the fault-tolerant layer by building the QRAM gate out of a more fundamental gate set. As expected, the cost is proportional to the number of bits of memory, and we show a tight lower bound for this.

\Cref{sec:gate-qram} then explores the concept of passive hardware-QRAM, i.e., a device tailor-made to perform QRAM accesses without intervention, with \Cref{sec:errors} focusing specifically on errors. These describe the precise errors listed in the introduction.

While some of these issues are quantum-specific, others seem to also apply to classical RAM, which we know is a feasible technology. In \Cref{sec:classical-ram}, we apply the same arguments to classical RAM. For each argument, either it does not apply because classical memory and its access is an easier task, or the analysis and real-world use of memory in large-scale devices already accounts for the costs we describe.

We survey approaches to the bucket-brigade in \Cref{sec:bucket-brigade}, noting how each different proposal fails in one or more of the criteria above. \Cref{sec:other-qram} describes 3 other architectures and one error correction method, all of which also fall short.

While \Cref{sec:definitions} and \Cref{sec:applications} will be useful for all readers, those interested in \emph{using} QRAM in their algorithms should focus on \Cref{sec:circuit-qram}. \Cref{sec:quantum-linear-algebra} is intended for readers focusing on quantum linear algebra. Those interested in the arguments against QRAM should read \Cref{sec:gate-qram}, \Cref{sec:errors}, and \Cref{sec:classical-ram}. Readers already familiar with this debate might focus on the original results, such as \Cref{thm:no-qracm-distillation-informal}, the critiques of specific schemes in \Cref{sec:bucket-brigade} and \Cref{sec:other-qram}, and the analysis of quantum linear algebra in \Cref{sec:quantum-linear-algebra}.

\tableofcontents

\section{Definitions and Notation}\label{sec:definitions}
We will frequently use asymptotic notation. Writing $f(n)\in O(g(n))$ means that $g(n)$ is an asymptotic upper bound for $f$, $f(n)\in \Omega(g(n))$ means $g(n)$ is an asymptotic \emph{lower} bound, and $f(n)\in \Theta(g(n))$ means both. $f(n)\in o(g(n))$ means that $\lim_{n\rightarrow\infty}\frac{f(n)}{g(n)}=0$. A tilde over any of these means we are ignoring logarithmic factors, e.g., $\tilde{O}(g(n))$ means $O(g(n)\log^k(n))$ for some constant $k$. We may write $f(n)=O(g(n))$ to mean the same as $f(n)\in O(g(n))$.

\subsection{QRAM Definitions}

We will follow the notation of Kuperberg~\cite{TQC:Kuperberg2013}, who proposes a taxonomy of four types of memory, based on the product of two criteria:
\begin{itemize}
\item
Does the device store classical or quantum data? (CM for classical memory or QM for quantum memory)
\item
Can the device access this data in superposition or not? (CRA for classical access or QRA for quantum access)
\end{itemize}
The ``R'' stands for ``Random'', though we will ignore this because a memory device does not need to gracefully handle random access patterns to fit the definitions.

This results in four types of memory:

\paragraph{CRACM (Classical access classical memory):} This is classical memory that can only be accessed classically, i.e., the memory we are all familiar with today. Kuperberg does not actually distinguish the \emph{speed} of memory access. Comparing quantum and classical, we do not really care. In this sense DRAM, flash memory, hard disk drives, and even magnetic tapes can all be considered CRACM. We discuss more properties of CRACM in \Cref{sec:classical-ram}.

\paragraph{CRAQM (Classical access quantum memory):} Here we have quantum memory, but which can only be accessed classically. The notion of ``access'' is somewhat fuzzy: obviously a classical controller cannot read quantum memory without destructively measuring it. Our notion of access will be more akin to ``addressable'', meaning the classical computer can find the required bits of quantum memory. This means essentially all qubits are CRAQM, since any quantum circuit will require a classical controller to find the necessary qubits to apply the gates for that circuit. 

Classically addressing qubits is a challenging and extremely important problem (e.g., \cite{QI:JPCL+19}), but it is fundamentally different and mostly unrelated to what this paper addresses. 

\paragraph{QRACM (Quantum access classical memory):} We formally define ``QRACM'' as follows:

\begin{definition}
Quantum random-access classical memory (QRACM) is a collection of unitarities $U_{QRACM}(T)$, where $T\in\{0,1\}^N$ is a table of data, such that for all states $\ket{i}$ in the computational basis, $0\leq i\leq N-1$,
\begin{equation}
U_{QRACM}(T)\ket{i}\ket{0} = \ket{i}\ket{T_i}.
\end{equation}
\end{definition}

Here $T$ is the classical memory and $\ket{i}$ is an address input. We call the register with $\ket{i}$ the \emph{address register} and the bit which transforms from $\ket{0}$ to $\ket{T_i}$ the \emph{output register} (recall $T_i\in\{0,1\}$). Since $U_{QRACM}(T)$ is unitary, superposition access follows by linearity. Our definition specifies a unitary action, but we do not require the device to be unitary: a quantum channel whose action is identical to (or approximates) conjugation by $U_{QRACM}(T)$ will still be considered QRACM (for example, the method in \cite{QUANTUM:BGMMB2019} uses measurement-based uncomputation).

In this definition, the table $T$ parameterizes the unitaries, since QRACM must take some classical data as input. This means that there is no single QRACM ``gate'': there is a \emph{family} of gates, parameterized by $T$. In fact, since $T$ has $N$ entries of $\{0,1\}$, there are $2^N$ possible memories. 

In some literature the table is written as a function $f$ so that $T_i=f(i)$. In this case, the QRAM might be written as an oracle $O_f$ implementing the mapping $O_f\ket{x}\ket{0} = \ket{x}\ket{f(x)}$. For some QRACM methods, the table does not need to be in classical memory if $f$ is efficiently computable.

We take $N$ as the number of elements in the table that the QRACM must access, which we will sometimes call \emph{words} (for cases where they are larger than single bits). This means the address register has $n:=\ceil{\lg N}$ qubits. Other papers choose a notation focusing on $n$, the number of qubits in the address, so that the number of words in the table is $2^n$ (and the number of possible tables is $2^{2^n}$). Using $n$ and $2^n$ emphasizes that, for certain algorithms, the size of the table will be exponential compared to other aspects of the input and running time. We will mainly use $N$ and $\ceil{\lg N}$ since the classical table must be input to the algorithm, and therefore anything $O(N)$ is linear in the input size, but we use $n$ throughout as shorthand for $\ceil{\lg N}$.

Of course most applications require the elements of $T$ to be more than single bits, but concatenating single-qubit QRACM gives multi-bit QRACM.

\paragraph{QRAQM (Quantum access quantum memory)}: We formally define this as:

\begin{definition}
Quantum random-access quantum memory (QRAQM) is a collection of unitaries $U_{QRAQM}(N)$, where $N\in\N$, such that for all states $\ket{i}$ with $0\leq i\leq N-1$ and all $N$-qubit states $\ket{T_0,T_1,\dots, T_{N-1}}$ (referred to as the \emph{data register}, where $T_0,\dots,T_{N-1}$ is a bitstring),
\begin{equation}
U_{QRAQM}\ket{i}\ket{0}\ket{T_0,\dots,T_{N-1}} = \ket{i}\ket{T_i}\ket{T_0, \dots, T_{N-1}}
\end{equation}
\end{definition}
The difference now is that the memory itself is a quantum state, meaning we could have a superposition of different tables. Here the gate is implicitly parameterized by the table size, and needs to suffice for all possible states of that size.

This definition only shows a unitary that can \emph{read} the quantum memory, but one can show that unitarity forces 
\begin{equation}
U_{QRAQM}\ket{i}\ket{1}\ket{T_0,\dots,T_{N-1}} = \ket{i}\ket{1\oplus T_i}\ket{T_0, \dots, T_{N-1}}
\end{equation}
and from this, one can conjugate $U_{QRAQM}(N)$ with Hadamard gates on all qubits in the output and data registers to construct a write operation. 

The literature still has some ambiguity about this definition. For example, we could define a SWAP-QRAQM such that 
\begin{align}
U_{S-QRAQRM}\ket{i}\ket{\psi}&\ket{\phi_0}\dots\ket{\phi_{N-1}}\nonumber\\
=\ket{i}\ket{\phi_i}&\ket{\phi_0}\dots\ket{\phi_{i-1}}\ket{\psi}\ket{\phi_{i+1}}\ket{\phi_{N-1}}.
\end{align}
SWAP-QRAQM is equivalent to the previous read-and-write QRAQM, if some ancilla qubits and local gates are used.

\subsection{Routing and Readout}
Imprecisely, QRACM and QRAQM have two tasks: the first is routing and the second is readout. 

By analogy, consider retrieving data classically from a network database. A request for the data will contain some address, and this will need to be appropriately routed to a physical location, such as a specific sector on a specific drive in a specific building. Once the signal reaches that location, the data must be copied out before it can be returned. For example, this could involve reading the magnetization of the hard drive. 

Difficulties can arise from either routing or readout, but we stress that both are necessary. \cite{TQC:Kuperberg2013} states ``Our own suggestion for a QRACM architecture is to express classical data with a 2-dimensional grid of pixels that rotate the polarization of light. (A liquid crystal display has a layer that does exactly that.) When a photon passes through such a grid, its polarization qubit reads the pixel grid in superposition.'' This seems to satisfy the requirements of the readout, but does not address routing: how do we generate the photon in such a way that its direction is precisely entangled with the address register, so that for state $\ket{i}$, the photon travels to the $i$th pixel in this grid?

Generally, the routing problem seems more difficult. If we have a device that routes data such that the physical location of some aspect of the device (such as the photon described above) becomes entangled with the input state given as address, then readout becomes relatively simple. A readout with the polarization mechanism described above could give QRACM, or a series of parallel controlled SWAP gates could give QRAQM. Thus, we mainly focus on routing in this paper. 

For that reason, we also will use the term ``QRAM'' to refer to a routing technology that would work for either QRACM or QRAQM with only modest changes to the readout mechanism (the bucket-brigade routing from~\cite{PRL:HarHasLloy09} is a perfect example).

\subsection{Alternate Notation}
``QRAM'' is used inconsistently throughout the literature. It is sometimes used for QRAQM or QRACM (e.g. \cite{ARXIV:PhaChaGho2023} uses ``QRAM'' for both). It is also sometimes meant to specifically refer to an operation that implements QRACM, but via a single gate. 

``QROM'' has been used to refer to QRACM implemented as a circuit (e.g., from Clifford+T gates) (e.g., \cite{PRX:BGB+2018}). QRACM must be read-only, so using QRAM to refer to QRAQM and QROM to refer to QRACM would be a consistent alternative notation; however, since the literature uses QRAM to refer to both kinds, we opt against this. 

In classical computing, devices such as hard drives or tapes \emph{can} implement the functionality of random-access memory, though the access times are longer and more variable than devices such as DRAM or flash memory. Analogously, different proposals achieve the same required function of QRACM -- to access memory in superposition -- but with radically different costs and access times. 

To be more precise than this, we will use \emph{circuit}-QRAM to refer to proposals to construct superposition access using a standard gate set, especially Clifford+T. 

We will use \emph{hardware}-QRAM to refer to proposals for QRAM consisting of a specialized device. Following \cite{NP:Aaronson2015,ProcRoySoc:CHIP+18}, we further categorize into \emph{active} hardware-QRAM, which requires $\Omega(N)$ active interventions from the classical controller for each access, and \emph{passive} hardware-QRAM, which requires $o(N)$ active interventions. 

If we imagine memory access by routing an address through a binary tree, the tree must have $N-1$ nodes (if $N$ is a power of 2), and each node must do some computation on the incoming signal. The difference between active and passive hardware-QRAM is whether that computation requires some external process to make it happen, or whether it proceeds on its own with no external intervention. 

\paragraph{Related Work.} The concurrent survey in \cite{ARXIV:PhaChaGho2023} explains some approaches to QRAQM and QRACM, though they do not consider the physical costs in the same way as our work. Specifically, they claim classical RAM requires $O(2^n)$ ``gate activations'' and quantum RAM requires $O(n)$, without defining a notion of ``gate activations''. Arguably the real situation is reversed: when each gate in a QRAM circuit requires external intervention then it costs $O(2^n)$, and classical RAM only needs to dissipate energy from $O(n)$ gates (see \Cref{sec:classical-ram} for a more nuanced picture).

\emph{State preparation} is a slight generalization of QRAM, with a circuit that performs
\begin{equation}
    \ket{x}\ket{0}\mapsto \ket{x}\ket{\psi_x}
\end{equation}
for some collection of states $\ket{\psi_x}$ for $x\in \{0,\dots, N-1\}$. Typically we assume a unitary $U_x$ for each $x$, such that $U_x\ket{0}=\ket{\psi_x}$. QRACM is the special case where $\ket{\psi_x}=\ket{f(x)}$ (e.g. a specific computational basis state).

The techniques for state preparation and QRAM are generally the same, e.g. the techniques of \cite{PRL:ZhaLiYua2022} are similar to bucket-brigade QRAM (\Cref{sec:circuit-bb}) and \cite{TCIC:STY+2023} is like a select-swap QRAM (\Cref{sec:select-swap}) using phase rotations.

\section{Applications}\label{sec:applications}
Mainly as motivation for later discussions, we note some leading applications of QRAM. A key point for all of these applications is that the competing classical algorithms parallelize almost perfectly, and hence the active QRAM opportunity cost arguments apply fully. 

\paragraph{Optimizing calculations.} Many quantum algorithms perform complicated classical functions in superposition, such as modular exponentiation in Shor's algorithm or inverse square roots in chemistry problems. Recent techniques~\cite{PRX:BGB+2018,TCHES:BBHL2020,ARXIV:Gidney19,Q:GidEke21,PQC:Haner20,QUANTUM:BGMMB2019} choose instead to classically pre-compute the result for a large portion of inputs, and then use QRACM to look up the result. 

The scale of QRACM here is around $2^{15}$ to $2^{25}$ bits of memory. Further, the overall error rate can afford to be fairly high. For example, using windowed elliptic curve point addition in Shor's algorithm requires only about 17 applications of the QRACM gate, so even if each one has much higher error than other fundamental gates, the overall algorithm will still succeed. 

Systematic QRACM errors are also less important in this context, as the inputs are generally large superpositions, so mistakes in a small proportion of inputs will only affect a small fraction of the superposition.

The referenced works already assign an $O(N)$ cost for QRACM look-ups, meaning that the arguments against cheap QRACM in this paper have no effect on their conclusions.

Cheap QRACM would also not drastically reduce the costs. First, for factoring, each table is distinct. Since the classical controller requires at least $N$ operations to construct a table of $N$ words, asymptotically the cost of using QRACM in this context is already proportional to $N$ just to \emph{construct} each QRACM table. While this classical pre-processing is likely much cheaper than a T-gate, it limits how much computation can be offset.

In contrast, quantum chemistry that uses QRACM to look up the value of a function of Coulumb potential~\cite{QUANTUM:BGMMB2019} will only have to access one table. Thus, cheap QRACM would reduce the cost of these algorithms, but it would only be a polynomial factor.

\paragraph{Dihedral Hidden Subgroup.} Kuperberg's algorithm for solving the dihedral hidden subgroup problem relies on QRACM~\cite{TQC:Kuperberg2013}. In this application, the QRACM table is only read once by the quantum computer. Hence, it is asymptotically irrelevant whether access costs $\Omega(N)$ for the quantum computer or less, since the classical controller has an $\Omega(N)$ cost to construct the table.

For cryptographic applications, the size of each QRACM access ranges from about $2^{18}$ bits to $2^{51}$, depending on the parameters~\cite{EC:BonSch2020,JCE:CCJR21,EC:Peikert2020}. The robustness to noise is unknown.

\paragraph{Exponential Quantum Cryptanalysis.} Algorithms with exponential run-time or space requirements are used in the security analysis for cryptography intended to be quantum-safe. Specifically, QRACM is needed in algorithms in collision finding~\cite{SIGACT:BraHoyTap1997} and lattice sieving~\cite{EPRINT:AlbShe2022,EPRINT:Heiser2021,SAC:Laarhoven2017,DCC:LaaMosPol2015}, and QRAQM is necessary for attacks based on quantum random walks, such as claw-finding~\cite{MFoCS:Tani2007}, information set decoding~\cite{ARXIV:KacTil2017,PQC:Kirshanova2018}, subset-sum~\cite{PQC:BJLM2013}, and other lattice sieving methods~\cite{AC:ChaLoy2021,EC:BCSS2022}. As NIST selected lattice cryptography for standardization in 2022~\cite{NIST2022}, lattice sieving is critical.

For the asymptotically ``fastest'' quantum lattice sieving against the smallest proposed parameters of the post-quantum schemes to be standardized~\cite{EC:BCSS2022}, the QRACM needs $2^{49}$ bits. Against mid-scale parameters one needs $2^{78}$ bits. 
For collision search against the ubiquitous hash function SHA-256, one needs $2^{93}$ bits of QRACM.

In addition to the hefty memory requirements, the accuracy of the memory must be extraordinarily high. \cite{ALP:RegSch08} and \cite{NJP:AGJMS15} show that Grover's algorithm needs exponentially small QRAQM error, but one can also see that most of these cryptography attacks require finding a \emph{single} element of the memory. If that element is destroyed because of noise, the algorithm fails. 

This is the most extreme application of QRAM with the most stringent requirements on the device. We will refer to these algorithms somewhat imprecisely as \emph{database QAA}, to indicate that we are using QAA (quantum amplitude amplification)~\cite{AMS:BHMT2002}, or a quantum random walk, to search over some classical database.

\paragraph{Context}\label{sec:context}
There may be a disconnect in the discourse around the feasibility of QRAM. Proponents and hardware designers may only be considering small-scale regimes of up to $2^{30}$ bits of QRAM, in which case the arguments presented here and elsewhere \emph{may} be surmountable: classical computational components in 2004 had error rates of $2^{-30}$ to $2^{-40}$~\cite{REPORT:Tezzaron04}, so perhaps quantum components could reach the same level. 

However, at the scale of e.g. $2^{20}$ bits, our conclusions are also less relevant; we know (\Cref{sec:circuit-qram}) how to perform $N$-bit QRAM with a circuit of $O(N)$ gates, so any algorithm using $2^{n}$ bits of QRAM will face only a $2^n$ overhead if physical QRAM assumptions are wrong. For something like $2^{20}$ bits, this does not make a large difference, especially considering the extreme uncertainties in future quantum computing overheads and architectures. 

In contrast, at the scale of $2^{50}$ bits, cheap passive QRAM causes a drastic change in algorithm costs, which has incentivized the research into algorithms that use QRAM. Yet, even minute physical costs become relevent if they scale with memory size. For perspective, if we lose a single visible photon's worth of energy for every bit in memory for each access, accessing $2^{50}$ bits costs about 1 kJ. 

As we move into cryptographic scales, QRAM becomes clearly absurd. A single photon's worth of energy per bit means 1 \emph{giga}-Joule per access if we have $2^{80}$ bits. Even more extreme: in a Fabry-Perot interforometer, the momentum change from a single microwave photon on a 10 g mirror would shift the mirror's position by a portion of $2^{-98}$ of the cavity's length. If the memory needs $2^{-128}$ precision (as one proposed QRAM technology, see \Cref{sec:phase-gate}), then this shift is a billion times larger than the precision needed.

Most of our conclusions in this paper suggest that an $N$-bit QRAM access will need energy/computation/some other cost to scale as $\Omega(N)$ (i.e., be active), and/or each physical component will need error rates of $O(1/N)$. These are asymptotic claims, and it's beyond what we can say today to claim that at at small-scale memory sizes, these terms will dominate other more relevant costs. However, at machine learning scales and beyond, we must consider these factors.

\section{Case Study: Quantum Linear Algebra}\label{sec:quantum-linear-algebra}

Certain statistical and/or machine learning problems can be framed as linear algebra, where we transform some vectors by matrices. Quantum linear algebra aims to accelerate these transformations; however, if the input data are large databases, this requires QRAM. The following references provide an overview of some of the techniques used in, and applications of, quantum machine learning, in some cases discussing their assumptions regarding the necessity of QRAM~\cite{biamonte2017quantum, gilyen2019quantum, pistoia2021quantum, liu2023towards}. Moreover, the following references provide an overview of, and detail important techniques in, quantum linear algebra with some discussing the creation of the QRAM data-structures relevant for linear algebra~\cite{PRL:HarHasLloy09, rebentrost2014quantum, kerenidis2016quantum, low2019hamiltonian, gilyen2019quantum, martyn2021grand, lin2022lecture}.

We now detail our input model and formally define the general linear algebra task we consider. We then discuss the necessity of QRAM (i.e. for ``unstructured'' input data), and evaluate the opportunity cost of QRAM by giving algorithms for the classical co-processors to solve the defined linear algebra problem. We conclude that without passive hardware-QRAM, quantum linear algebra algorithms provide no asymptotic advantage, unless we are at a scale large enough that the classical co-processors cannot share access to the same memory, but small enough that signal latency is negligible. 

Ref~\cite{NP:Aaronson2015} made these opportunity cost arguments for matrix inversion, \cite{TALK:Steiger16} applied them to an approach to ``Quantum PageRank'', and \cite{ProcRoySoc:CHIP+18} briefly notes that these arguments generalize. This section expands and generalizes these ideas, incorporating questions of physical layout.

\subsection{Input Model and Matrix Functions}

Formally, we assume we are given some matrix $H\in \mathbb{C}^{N\times N}$ such that $H=H^{\dagger}$, via query access to the elements of $H$, and some vector $\bm{v}\in \mathbb{C}^{N}$ again via query access to its elements. The query access is provided via oracles which return the output associated with any given input. For example, the oracle for matrix $H$ may be given by
\begin{align}
    O_H\ket{x}\ket{y}\ket{0} = \ket{x}\ket{y}\ket{H_{xy}}.
\end{align}
In the case where $H$ is sparse (i.e. $H$ has at most $d$ non-zero elements in any given row or column), we also assume we are given another oracle specifying the locations of the non-zero elements. Moreover, we restrict our attention to the case where the matrix under consideration is both square and Hermitian, noting that non-square and non-Hermitian matrices $\tilde{H}$ can be embedded into a Hermitian matrix $H := \begin{pmatrix} 0 & \tilde{H}\\ \tilde{H}^{\dagger} & 0\end{pmatrix}$, as is commonly done in the literature~\cite{PRL:HarHasLloy09}. In this embedding, the eigenvalues  of $H$ are directly related to the $\pm$ singular values of $\tilde{H}$, and its eigenvectors are a simple function of the right and left singular vectors of $H$. Finally, we define a function $f$ of a Hermitian matrix $H$ as a function of its eigenvalues, i.e. when $H\ket{\lambda_i} = \lambda_i\ket{\lambda_i}$, $f(H) := \sum_j f(\lambda_j)\op{\lambda_j}{\lambda_j}$. This differs slightly from some definitions, which define the function on the singular values.
We restrict our attention to general set of linear algebra problems which are equivalent to computing the vector $f(H)\bm{v}$.  We further restrict $f$ to be a degree-$k$ polynomial, with small loss of generality since polynomials can approximate most functions of interest on the domain $[-1, 1]$. We give the formal problem statement in \cref{problem:polynomial_eigenvalue_transform}.

\begin{problem}[Polynomial Eigenvalue Transform Problem]\label{problem:polynomial_eigenvalue_transform}
Given a polynomial function $f : [-1, 1] \mapsto \mathbb{C}$ of degree at most $k$, such that $\forall x, |f(x)| \le 1$, an $N\times 1$ initial vector $\bm{v}$ and a $d$-sparse, $N\times N$, Hermitian matrix $H$, specified by an oracle $O_H$ (and an additional location-oracle in the case that $d \in o(N)$), such that $\lnorm{H}_2\le 1$, compute $f(H)\bm{v}/\lnorm{f(H)\bm{v}}_2$. 
\end{problem}

Many quantum linear algebra tasks can be cast in the framework of \cref{problem:polynomial_eigenvalue_transform}, such as matrix inversion, Gibbs sampling, power iteration, and  Hamiltonian simulation~\cite{martyn2021grand}. In essence we are describing the Quantum Singular Value Transform (QSVT) framework~\cite{gilyen2019quantum}, which can be seen as a unifying framework for nearly all quantum algorithms as summarized in~\cite{martyn2021grand}. The QSVT is defined on the singular values rather than the eigenvalues, but in the case of Hermitian matrices, the singular values are just the absolute eigenvalues, and so the problems are practically equivalent (with some additional tedium). 

Of course, QRAM is not always necessary to offer query access to $H$ or $\bm{v}$. In some cases, there exist circuits with $\text{polylog}(N)$ depth which can implement the mappings, such as in~\cite{vedral1996quantum, bhaskar2015quantum, zhang2022quantum, rattew2022preparing}. With entries that are not efficiently computable, an oracle to $H$ needs QRAM access to the $Nd$ non-zero entries of $H$. We assume a ``wide'' QRAM such as a bucket-brigade or fanout-and-swap, which need $Nd$ quantum registers but provide $O(\log(Nd))$ circuit-depth per access. Our analysis holds if the QRAM is implemented with fewer quantum registers but with longer access times, as that benefits the classical co-processors in our comparison.

\subsection{Dequantization}

Tang~\cite{tang2019quantum} first introduced dequantization proving that the quantum algorithm for recommendation systems of Kerenidis and Prakash~\cite{kerenidis2016quantum} does not have exponential quantum advantage over a randomized classical algorithm making a similar input assumption. The argument starts by noting that the $Nd$ non-zero elements of the input matrix must be stored somewhere (e.g., by preparing a QRAM gate or data structure), and this creates a $\Omega(dN)$ set-up cost. For example, the nodes of a bucket-brigade QRAM must be initialized with the appropriate data.

A similar amount of classical precomputation can create a data structure that allows efficient $\ell_2$-norm sampling from all columns of the input matrix. Even if data is added in an online setting to the quantum memory (amortizing the cost), with similar cost a classical data structure can also be maintained with online updates. Combined with a large body of classical randomized numerical linear algebra techniques (e.g. \cite{kannan2017randomized, chepurko2022quantum}), one can obtain classical algorithms for \cref{problem:polynomial_eigenvalue_transform} with polylogarithmic dependence on the dimensions of the input.

Dequantization has some limitations. Most approaches require either the matrix to be low-rank~\cite{tang2019quantum, gilyen2018quantum, bakshi2023improved}, or with some caveats, for the matrix to be sparse and the polynomial function being applied to have at most a constant degree~\cite{bakshi2023improved}. 
If any classical algorithm could implement matrix inversion on sparse matrices (which is a special case of QSVT) with a polynomial dependence on the condition number (i.e., without a meaningful rank constraint), then since matrix inversion was shown to be BQP-complete~\cite{PRL:HarHasLloy09}, such a classical algorithm could efficiently simulate any polynomial time quantum algorithm. This would obviously be an extremely surprising result, and so we cannot expect a classical dequantized QSVT-styled algorithm to have rank-independence, in general. 
Moreover, dequantization also leaves room for some polynomial quantum advantage, with QSVT requiring $O(k)$ QRAM queries for a degree-$k$ polynomial, compared to an $O(k^9)$ cost for leading classical techniques~\cite{bakshi2023improved}. In general, see Figure 1 of \cite{chia2022sampling} to compare complexities of dequantization. New classical techniques could close these gaps, but currently there are many algorithms with at least a quartic speed-up over their dequantized counterparts, so if cheap QRAM existed, even early fault-tolerant devices might see an advantage~\cite{babbush2021focus}.

Critically, dequantization arguments allow cheap \textit{access} to the QRAM, and only consider the opportunity cost of the quantum algorithm's input assumptions (e.g. the cost to construct the QRAM data structure). Instead we focus on the access cost. 
If the QRAM requires classical control or co-processing (e.g. laser pulses enacting the gates, or classical resources dedicated to error-correction), these polynomial quantum advantages over their dequantized counterparts vanish (and indeed the quantum algorithms would potentially even be exponentially slower than their dequantized counterparts) -- such arguments even apply for some algorithms that cannot be dequantized.

\subsection{QRAM versus Parallel Classical Computation}

In this section we discuss the possibility of a general asymptotic quantum speedup in~\cref{problem:polynomial_eigenvalue_transform}, by considering the opportunity cost of the classical control for the QRAM. We assumed the QRAM has $O(Nd)$ logical qubits. With active error correction for each qubit, syndrome measurements and corrections will require substantial classical co-processing, with $O(1)$ classical processors per logical qubit. We assume each of these processors has constant-sized local memory with fast and efficient access. It seems unlikely that any \textit{error-corrected} QRAM architecture could use \textit{asymptotically} fewer classical resources.

Even without error correction, only passive hardware-QRAM avoids costs proportional to $Nd$. Advances in quantum computing could drive the constant of proportionality to be smaller, but asymptotically the arguments in this section will still hold, short of a significant breakthrough in the manufacturing of passive QRAM (see \Cref{sec:gate-qram} for a comparison of active vs. passive QRAM).

We then further consider some frequently neglected constraints on large scale parallel computing, such as signal latency, connectivity, and wire length, which gives us three regimes of scale. The regime that considers total wire length but neglects the speed of light is probably the least realistic regime that we consider, and we suspect that the number of connections between processors/memory is more relevant than the total length of wires. Nevertheless, we include this regime as it is the only one we found which led to potential quantum advantage.

To start the comparison, we give a lower-bound on the number of QRAM accesses that any quantum algorithm will need to solve \cref{problem:polynomial_eigenvalue_transform}, then show a classical algorithm that solves the problem with the same number of matrix-vector multiplications. We can then compare parallel classical matrix-vector multiplication to a single QRAM access.

\begin{theorem}[Quantum Polynomial Eigenvalue Transform]\label{theorem:quantum_polynomial_eigenvalue_transform}
Given a degree $k$ polynomial on the interval $x \in [-1, 1]$, with $f(x) = \sum_{j=0}^k a_j x^j$, a $d$-sparse Hermitian matrix $H \in \mathbb{C}^{N\times N}$ and an $\ell_2$-normalized vector $\ket{\psi}\in\mathbb{C}^N$, no \textit{general} quantum algorithm can prepare the state $f(H)\ket{\psi}/\lnorm{f(H)\ket{\psi}}_2$ with asymptotically fewer than $\Omega(k)$ QRAM queries to $H$, each query with cost $T$.
\end{theorem}
This follows immediately from \Cref{lemma:quantum_polynomial_eigenvalue_transformation_lower_bound}. The result of \cite{gilyen2019quantum} almost immediately implies this bound, but the following derivation is included to enhance the clarity of our discussion. 

The main idea is that we can efficiently construct a block-encoding (a method to encode non-unitary matrices as the top-left block of a unitary) of a rank one projector from a uniform superposition to the marked state of an unstructured search problem (specified via an oracle) with a singular value equal to $\frac{1}{\sqrt{N}}$ (where $N$ is the dimension of the space). We can do this with only $O(1)$ queries to the unstructured search oracle; however, immediately applying the block encoding to an initial state has a low probability of success. We thus embed this in a rank-2 Hermitian matrix and apply a polynomial approximation to the sign function on the eigenvalues of this Hermitian matrix. The degree of this polynomial is $\tilde{O}(\sqrt{N})$, and it gives a constant probability of success in measuring a solution to the unstructured search problem. Consequently, if we could implement this polynomial with $o(\sqrt{N})$ queries, we could solve the unstructured search with less than $\Omega(\sqrt{N})$ complexity. Since unstructured search has an $\Omega(\sqrt{N})$ lower bound, this creates a lower bound of $\Omega(k)$ for any general quantum algorithm implementing an eigenvalue transform of a degree-$k$ polynomial.

\begin{table*}
\begin{tabular}{l l  l  l  l}
\hline
\hline
     Scale & \multicolumn{2}{ c }{Assumptions} &\multicolumn{2}{c }{Quantum advantage}\\
     \cline{2-5}
     & Free wires\hspace{1em} & Instant communication\hspace{1em} & Sparse Matrices\hspace{1em} & Dense Matrices \\
     \hline
     Small & Yes & Yes & None & None\\
     Medium & No & Yes & $\tilde{O}((Nd)^{1/2})$ & None \\ 
     Large & No & No & None & None\\
     \hline
\end{tabular}
\caption{Comparing classical versus quantum algorithms for \cref{problem:polynomial_eigenvalue_transform} for $N\times N$ matrices, under different architectural assumptions with active QRAM. It may appear strange that the potential quantum speedup in the Medium-Scale regime increases as $d$ increases, and then suddenly vanishes when $d\in \Omega(N)$. We suspect that a classical algorithm for the medium-scale regime for sparse matrices with $d \gg 1$ could be improved.}\label{tab:QLA-comparison}
\end{table*}

\begin{theorem}[Classical Parallel Polynomial Eigenvalue Transform]\label{lemma:classical_polynomial_eigenvalue_transform}
Given a degree $k$ polynomial on the interval $x \in [-1, 1]$, with $f(x) = \sum_{j=0}^k a_j x^j$, a $d$-sparse Hermitian matrix $H \in \mathbb{C}^{N\times N}$, and an $\ell_2$-normalized vector $\bm{v} \in\mathbb{C}^N$, a parallel classical algorithm can exactly compute the vector $f(H)\bm{v} / \lnorm{f(H)\bm{v}}_2$ with $O(k)$ matrix-vector multiplications and $k$ vector-vector additions.
\end{theorem}
\begin{proof}
By definition, 
\begin{equation}
    f(H) = \sum_{j=0}^{N-1} \sum_{l=0}^{k}a_l \lambda_j^l \op{\lambda_j}{\lambda_j}.
\end{equation}
Applying the completeness of the eigenvectors of $H$, we immediately get $f(H) = \sum_{j=0}^k a_j H^j$. This means we must simply compute $f(H)\bm{v} = \sum_{j=0}^k a_j H^j\bm{v}$. This can be done with $k$ matrix-vector multiplications by computing the powers $H^1\bm{v}$, $H^2\bm{v}$,$\dots$,$H^k\bm{v}$, and adding $a_jH^j\bm{v}$ to a running total vector after each multiplication. We do not need to store all the powers $H^j\bm{v}$, since we can delete each one once we add it to the total and compute the next power of $H$. Once the unnormalised vector $f(H)\bm{v}$ has been constructed, we can normalize it by computing the $\ell_2$ norm (which requires $O(N)$ parallelizable operations, and $O(\log N)$ joining steps), and then by diving each element in the output vector by this norm, which requires a further $O(N/P)$ steps. We assume that the cost of matrix-vector multiplications dominates the cost of this normalization.
\end{proof}

From now on we ignore the cost of vector-vector additions (which parallelize perfectly), since they are likely negligible compared to matrix-vector multiplications. Thus, our comparison comes down to to the expense of matrix-vector multiplication. In our computational model, the QRAM with $Nd$ registers implies a proportional number of classical processors, so we now consider highly parallel matrix-vector multiplication at different scales. \Cref{tab:QLA-comparison} summarizes our conclusions.

\paragraph{Small-Scale Regime:} At this scale we ignore all latency and connectivity concerns. The QRAM access can then be done in $\Theta(\log(Nd))$ time (i.e., circuit depth). We assume the classical processors can access a shared memory of size $Nd$, also with $O(\log(Nd))$ access time.

A standard parallel matrix multiplication algorithm (see \cref{lemma:classical_matrix_vector_multiplication}) allows the the $P = \Theta(dN)$ processors to multiply the matrix and vector in $\tilde{O}(1)$ time, matching the assumed time of the QRAM access. Together with \cref{theorem:quantum_polynomial_eigenvalue_transform} and \cref{lemma:classical_polynomial_eigenvalue_transform}, the classical algorithms achieves $\tilde{O}(k)$ scaling, while the quantum algorithm achieves a $\tilde{\Omega}(k)$ lower-bound for the general linear algebra task, so \textit{there can be no asymptotic quantum advantage} at this scale.

\paragraph{Medium-Scale Regime:} With a large matrix, shared memory access is unrealistic. One problem is that it requires $\Omega((Nd)^2)$ wires to connect each processor to each memory element. The QRAM does not need so many wires; a bucket-brigade QRAM needs only $O(Nd)$ wires, for example.

Thus, in this regime we only allow each processor to have a small number of wires (and count the total length of wires). Connecting the processors to each other and restricting access to local memory requires different algorithms. \cref{lemma:matrix_multiplication_via_hypercube_sort} shows how to multiply a sparse matrix and vector using $O(1)$ sorts. With $O(\log(Nd))$ connections per processor (and thus $\tilde{O}(Nd)$ wires), the processors can form a hypercube and sort in polylogarithmic time. In this case, the time for QRAM access and classical matrix-vector multiplication match, again negating any quantum advantage. 

However, we should also consider the \emph{length} of the wires. By ``wire'' we mean any communication infrastructure, under the assumptions that (a) two connections cost twice as much as one connection, and (b) a connection which is twice as long will cost twice as much. This means the architecture needs more resources to have longer connections. A real architecture must embed into the 3-dimensional Euclidean space we live in, and we further assume it embeds into 2-dimensional space to account for heat dissipation. Each classical processor has a finite size, so the total width of this machine must be $\Omega((Nd)^{1/2})$. This means the length of wires in this hypercube scales to $\tilde{\Omega}((Nd)^{3/2})$, which is also infeasible. 

The total width of the QRAM must also be $\Omega((Nd)^{1/2})$ by the same reasoning, but the QRAM need only have $O(Nd)$ total wire length if it is well-designed (e.g., an H-tree).

To use $O(Nd)$ wires classically, we will assume the classical processors are in a two-dimensional mesh with nearest-neighbour connectivity. In this case a sort takes time $O((Nd)^{1/2})$, so the time for matrix-vector multiplication is $\tilde{O}((Nd)^{1/2})$, compared to $\tilde{O}(1)$ for QRAM access. Thus, we have a square-root time advantage for the QRAM in this case.

This only applies to sparse matrices, however. If the matrix is dense, the QRAM implies $\Omega(N^2)$ classical co-processors, which can multiply the matrix by a vector in time $O(\log(N))$, as we show in \cref{lemma:dense_matrix_vector_multiplication_with_local_memory}. This means the quantum advantage can only potentially apply to sparse matrices.

Strangely the quantum advantage increases with $d$ until vanishing when $d=\Omega(N)$. Likely this is because the classical hardware budget increases with the QRAM size, which increases with $d$, and we have not considered optimized classical algorithms for large $d$. The sparse case is harder for the classical machine because the memory access for small $d$ could be random, creating a large routing problem to send data from arbitrary memory cells to each processor.

\paragraph{Large-Scale Regime:} At large enough scales, we must account for the time for signals to propagate through the computers, even at lightspeed. This does not change the asymptotic time for the 2-dimensional mesh, which only has short local connections. The QRAM must have a width of $\Omega((Nd)^{1/2})$, so each QRAM access now takes time $\Omega((Nd)^{1/2})$. Thus, the classical and quantum algorithms for \cref{problem:polynomial_eigenvalue_transform} both take time $\tilde{\Theta}(k(Nd)^{1/2})$, and the asymptotic quantum advantage disappears.

{\vspace{1em}}

Which regime is realistic? The scale of QRAM necessary for these algorithms would depend on the application, but for context, Google's ``IMAGEN'', a text-to-image diffusion model, used $2^{43}$ bits of data for its training set~\cite{ARXIV:CWSL+22}.

The small-scale regime might hold up to $Nd\approx 2^{29}$, as \cite{IJHPC:BWYF19} ignore memory connectivity and compute sparse matrix-vector multiplication at that size in 0.03 s with a single CPU and GPU. Larger problems will certainly face bandwidth issues, moving us to at least the medium-scale regime. The large-scale regime might seem far-fetched, where lightspeed signal propagation is the limiting factor, but consider that light would take 173 clock cycles to travel across the Frontier supercomputer, which has ``only'' $2^{56}$ bits of memory~\cite{IEEESPEC:Choi22}. Frontier also has enough wire that any two nodes are at most 3 hops away~\cite{IEEESPEC:Choi22}.

Further, we might be somewhere in between. We can divide a matrix into a block matrix, and use small-scale techniques for the blocks themselves and a larger-scale algorithm for the full matrix. 

In practice, Google's TPUs are 2-dimensional grids of processors with only nearest-neighbour connectivity, optimized for dense matrix multiplication~\cite{WEB:Google22}. Instead of the techniques we describe, they use a ``systolic'' method to multiply an $N\times N$ matrix that takes time proportional to $N$ with $N^2$ processors. Their architecture is optimized for \emph{throughput}, as it can simultaneously multiply the same matrix by many vectors. The throughput asymptotically approaches $O(1)$ cycles per matrix-vector multiplication, the best one can do with $N^2$ processors.

Overall, this becomes a battle of constant factors. The quantum advantage for sparse matrix linear algebra in the medium-scale regime depends on the constant factor overheads of quantum computing being less important than the constant factor of signal latency.

\subsection{Quantum Linear Algebra with Noisy QRAM}\label{subsec:linear_algebra_with_noisy_qram}
The opportunity cost arguments made in the last section apply for any active QRAM system, but they are strongest for error-corrected architectures (which require significant classical co-processing). In contrast, passive QRAM systems (potentially such as certain implementations of bucket-brigade) suffer no such opportunity cost, as once the QRAM query is initiated, the system requires no further energy input or intervention (although they will necessarily be noisy). As such, if sufficient progress can be made in the construction of a passive QRAM system (despite the challenges we identify in \Cref{sec:gate-qram}), then for certain linear algebra algorithms that can tolerate some noise in their QRAM queries, it is possible that practically relevant quantum advantage could be realized.

For example, with the bucket brigade architecture, the overall probability of an error per query to a QRAM with $N$ memory cells scales as $\sim p \log(N)^2$, where $p$ is the physical probability of an error (see \cref{sec:bb-errors},\cite{PRX:HLGJ21}). This immediately rules out an asymptotic advantage for any quantum algorithm asymptotically improving error-rate dependence (for arbitrarily small errors), but may be practically useful in cases where either the overall algorithm is relatively insensitive to QRAM access errors, or if the physical error rate can be made to approach the necessary precision without error correction or external intervention.

In some linear algebra tasks, a small constant error per QRAM access \textit{may} be fine. For instance, in some machine learning applications the underlying data is already noisy, so additional noise may be acceptable. In other cases, such as the training and deployment of large language models (LLMs), it has been demonstrated that models can still function with reasonable accuracy when the weights are quantized down to $16$-bits~\cite{kalamkar2019study}, $8$-bits~\cite{dettmers2022llm}, or \textit{even down to $2$-bits}~\cite{frantar2022gptq}. Clearly a two-bit quantization introduces substantial noise in the evaluation of a model, and if LLMs can still function in this regime, it is possible they could be similarly resilient to noise in a QRAM query to their parameters.

Moreover, in reinforcement learning, additional noise could assist an agent to explore its environment, rather than just exploit its existing policy~\cite{sutton2018reinforcement}. 

For matrix inversion with a structured matrix (i.e. a matrix not stored in QRAM) on an input vector constructed from QRAM, for well-conditioned matrices the output of the matrix inversion is relatively insensitive to small perturbations in the input vector. For matrices in QRAM, for some classes of matrices the inversion could be insensitive to perturbations in the structure of the matrix itself~\cite{el2002inversion}.

We recommend further exploration into error-corrected quantum linear algebra algorithms using noisy QRAM systems. We do not rule out the possibility of practical quantum advantage in such cases. Of course, at some scale the uncorrected noise will overwhelm nearly any algorithm, but with sufficient noise resilience  (in both e.g. the bucket-brigade QRAM and in the algorithm itself), at realistic scales such quantum algorithms could still be of significant commercial and scientific value.

\section{Circuit QRAM}\label{sec:circuit-qram}
Here we highlight four approaches to circuit QRAM and give brief descriptions of each. For the full details, follow the references. All of these approaches have overall costs of $\tilde{\Theta}(N)$ gates to access $N$ bits of memory, even with logarithmic circuit depth. This may feel unreasonably high; however, in a bounded fan-in circuit model, this is essentially optimal.

As \Cref{sec:applications} shows, there are many applications where QRAM is a useful tool, even with this high cost. Circuit QRAM is a valuable tool in quantum algorithm analysis.

\subsection{Lower Bounds}
The key point of these lower bounds is that with $N$ bits of memory in a QRACM, there are $2^N$ possible QRACM gates, depending on the values in those bits. However, the number of different quantum circuits is roughly exponential in the product of the number of qubits and the circuit depth, which one can see because we have a finite set of choices of gate for each qubit at each time step. Thus, we need $N\approx $depth$\times$width. This is analogous to a classic result of Shannon about classical circuits~\cite{BELL:Shannon49}.

More precisely, we prove the following Lemma in the appendix:
\begin{lemma}\label{lem:num-operations}
The set of all distinct quantum circuits on $W$ qubits of depth $D$ using $G$ gates from a fixed set of $g$ possible gates of fanin at most $k$,  has size at most $2^{\gamma G}$, where $\gamma = k \cdot \lg (\tfrac{DWg}{G\sqrt{k}})$.
\end{lemma}
Intuitively, one can obtain a simple non-rigorous version of the above result with the following argument. Suppose we have a circuit with $G$ gates, and that the gate set contains $M=2^m$ possibilities, including the identity. Then, there are $M^{G}= 2^{m G}$ possible circuits. The above argument makes this notion rigorous and general, e.g., covering cases where some gates may be multi-qubit controlled gates.

We will prove the lower bound even for approximations of QRACM, which we define as follows: Given a table $T$ with $N$ rows, let $U_{QRACM}(T)$ be the QRACM circuit implementing $T$. A circuit $C$ $\epsilon$-approximately implements a $U_{QRACM}(T)$ if $\lnorm{C - U_{QRACM}(T)}_2 \le \epsilon$.

\begin{theorem}\label{thm:circuit-qram-lower-bound}
Let $0 \le \epsilon \le 1/\sqrt 2$. Let $0 < \tau \le 1$ be the fraction of possible QRACM tables that circuits with at most $G$ gates from the gate set in~\cref{lem:num-operations} is able to $\epsilon$-approximately implement. Moreover, define $\gamma$ as per~\cref{lem:num-operations}. Then,
\begin{enumerate}
    \item A given circuit $C$ can $\epsilon$-approximate at most 1 QRACM table.
    \item The number of gates $G$ to $\epsilon$-approximately implement a fraction $\tau$ of all possible QRACM tables is lower-bounded by $G \ge \frac{\log \tau + N}{\gamma}$. 
    That is, $G \in \tilde{\Omega}\left(N + \log\tau \right)$.
    \item To approximately implement all possible QRACM tables with circuits using at most $G$ gates, $N\in\tilde{O}(G)$, and equivalently, $G \in \tilde \Omega (N)$. 
\end{enumerate}
\end{theorem}
\begin{proof}
We are given two $N$-row QRAM tables $T$ and $T'$, enacted by the unitaries $U(T)$ and $U(T')$ respectively. Assume that the tables differ at index $j$, i.e., $T[j] \neq T'[j]$. 
Then,
\begin{align}
    \lnorm{(U(T) - U(T'))\ket{j}_n\ket{0}_1} &= \sqrt 2.
\end{align}
Consequently, $\lnorm{U(T) - U(T')}_2 \ge \sqrt{2}$.

Now suppose that a given unitary circuit $C$ $\epsilon$-approximates both $U(T)$ and $U(T')$, i.e., $\lnorm{U(T) - C}_2 \le \epsilon$ and $\lnorm{U(T') - C}_2 \le \epsilon$. Then, 
\begin{align}
    \sqrt{2} &\le 
    \lnorm{U(T) - C + C  - U(T')}_2 \le 2 \epsilon. 
\end{align}
Thus, $1/\sqrt{2} \le \epsilon$. Thus, it is not possible for one circuit to approximate more than one QRAM table, unless $\epsilon > 1/\sqrt{2}$. Importantly, \textbf{allowing tables to be approximated to a constant error does not provide any meaningful asymptotic advantage}.

Using the gate set defined in~\cref{lem:num-operations}, \cref{lem:num-operations} shows that there are $2^{k G (\lg(\frac{DWg}{G\sqrt{k}}))}$ possible circuits using $G$ gates. Assuming that $\epsilon \le 1/\sqrt{2}$, each distinct circuit can only approximate at most a single QRAM table. Let $\gamma := k (\lg(\frac{DWg}{G\sqrt{k}}))$. Consequently, a set of $2^{\gamma G }$ distinct circuits can approximate at most $2^{\gamma G}$ distinct QRAM tables. 

Noting that there are $2^N$ possible QRAM tables with $N$ rows, the fraction of $\epsilon$-approximable tables with at most $G$ gates is bounded by $\frac{2^{\gamma G}}{2^N} = 2^{\gamma G - N}$. Let $0 < \tau \le 1$ represent the fraction of QRAM tables that circuits with $G$ gates are able to approximate. Then,
\begin{align}
    \tau &\le 2^{\gamma G - N} \implies  \frac{\log \tau + N}{\gamma} \le G. 
\end{align}
Thus, for circuits with $G$ gates to be able to implement a $\tau$ fraction of all possible $N$ row QRAMs, $G \in \tilde\Omega(N + \log\tau)$.  
\end{proof} 
We will now briefly outline some special cases of this result, to help build intuition. If we wish to:
\begin{itemize}
    \item Cover all possible tables, then $\tau = 1$, so $G \in \tilde\Omega(N)$.
    \item Cover a constant fraction of possible tables then $\tau \in \Theta(1)$, so $G \in \tilde\Omega(N)$. E.g., if we wish to cover exactly half of all possible QRAM tables, $\tau = 1/2$, so $\frac{N-1}{\gamma} \le G$. If we wish to cover a quarter of all possible tables, $\frac{N-2}{\gamma} \le G$, etc. 
    \item Cover an exponentially small fraction of tables, e.g., $\tau \in \Theta(2^{-N})$, so $G \in \tilde\Omega(1)$. Let $c > 0$ be some constant such that $\tau = c 2^{-N}$. Then, we find $\log_2(c)/\gamma \le G$. 
\end{itemize}
This demonstrates how this result generally applies, even in cases where QRACM circuits can be efficiently constructed, as such efficient QRACM circuits belong to the exponentially small fraction of circuits where the gate-complexity lower-bound is efficient.

\begin{figure*}
\centering
\begin{tikzpicture}[scale=1.1]
\node at (0,0.8) {$\ket{\text{addr}}$};
\node at (0,-0.5) {$\ket{\psi_{in}}$};
\node at (0,0) {$\ket{0}$};
\draw (0.5,0) -- (10.5,0);
\draw (0.5,0.8) -- (10.5,0.8);
\draw (0.5,-0.5) -- (10.5,-0.5);
\node[right] at (10.5,-0.5) {$\ket{\psi_{out}}$};
\node[right] at (10.5,0.8) {$\ket{\text{addr}}$};
\node[right] at (10.5,0) {$\ket{0}$};

\foreach \i in {0,...,4}{
	\node at (1+2*\i,0) {$\oplus$};
	\draw (1+2*\i,0) -- (1+2*\i,0.8);
	\node[fill=white,draw=black] at (1+2*\i,0.8) {$?\atop = \i$};
	
	\node at (2+2*\i,0) {$\oplus$};
	\draw (2+2*\i,0) -- (2+2*\i,0.8);
	\node[fill=white,draw=black] at (2+2*\i,0.8) {$?\atop = \i$};
	
	\draw[fill=black]  (1.5+2*\i,0) circle (0.05);
	\draw (1.5+2*\i,0) -- (1.5+2*\i,-0.5);
	\node[rectangle,fill=white,draw=black] at (1.5+2*\i,-0.5) {$\oplus T[\i]$};
}
\end{tikzpicture}
\caption{Circuit diagram for naive QRACM. Here $? \atop =j$ means a circuit that checks if the input equals $j$, and flips the target output if so.}\label{fig:unary-circuit}
\end{figure*}
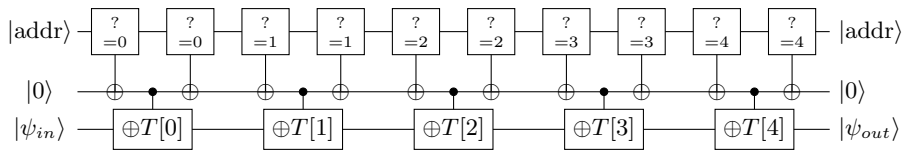

\begin{algorithm}[H]
\begin{algorithmic}[1]
\STATE \textbf{Input:} A qubit address register, an output qubit, and a classical table $T$ of $N$ bits
\FOR {$i=0$ to $N$}
	\IF {$T_i = 0$}
		\STATE Do nothing
	\ELSE
		\STATE Apply a circuit to check if the address register equals $i$ and write the result to an ancilla $\ket{b}$
		\STATE Apply a CNOT from $\ket{b}$ to the output qubit
		\STATE Uncompute $\ket{b}$
	\ENDIF
\ENDFOR
\end{algorithmic}
\caption{Simple implementation of a naive QRACM access.}\label{alg:unary-circuit}
\end{algorithm}

Notably, the lower bound applies to \emph{circuits}, meaning measurements are forbidden. With feedback from measurements, our simple counting arguments fall apart quickly: consider the set of operations, each defined by a table of size $N$, where one measures the entire input, classically finds that address in the table, then applies an $X$ if the bit at that address is 1. It's easy to see that there are $2^N$ distinct ``operations'' of this type, but the number of quantum gates to implement each one is only $\lg(N)+1$. However, these are not QRACM. To be useful, the measurements need to be mostly independent of the input state, but we leave it to future work to incorporate this into the counting arguments. If the feedback from measurements is simple enough to replace it with quantum control without drastically increasing the resource requirements, the bounds will hold. \cite{ARXIV:YuaZha2022} provide a similar lower bound through different techniques which are also restricted to measurement-free circuits.

This lower bound implies that variational approaches to QRAM (\cite{PRL:NZB+2022,ARXIV:PhaLiGho2022}) cannot scale better than other circuits. We can discretize the continuous parameters to get a finite gate set, which increases the parameter $g$ in \Cref{lem:num-operations}, but does not appreciably change the asymptotics until the precision of the parameters becomes exponentiation in $N$, the table size. Such precision seems absurd: if the duration of a laser pulse were specified to such a high precision, it would take $\approx N$ bits just to transmit the duration of the pulse to the laser!

\subsection{Unary Encoding}
In a QRACM access, the memory is entirely classical, so the control program can access all of it. The control program can dynamically construct a circuit to represent a particular memory. A naive strategy is to ``check-and-write'' for each address, as in \Cref{alg:unary-circuit}/\Cref{fig:unary-circuit}:

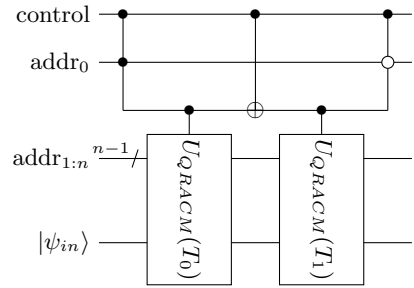
\begin{figure}
    \begin{tikzpicture}[scale=1.0]
\draw (0,0) -- (5.2,0);
\draw (0,0.8) -- (5.2,0.8);
\draw (0.4,-0.8) -- (4.8,-0.8);
\draw (0,-1.6) -- (5.2,-1.6);
\draw (0,-3.0) -- (5.2,-3.0);
\node[left] at (0,0.8) {control};
\node[left] at (0,0) {addr$_0$};
\node[left] at (0,-1.6) {addr$_{1:n}$};
\node[left] at (0,-3.0) {$\ket{\psi_{in}}$};
\draw (0.55,-1.75) -- (0.65,-1.45);
\node[left] at (0.65,-1.45) {$\scriptstyle{n-1}$};
\draw (0.4,0.8) -- (0.4,-0.8);
\node at (0.4,0) {$\bullet$};
\node at (0.4,0.8) {$\bullet$};
\node at (1.5, -0.8) {$\bullet$};
\draw (1.5,-0.8) -- (1.5,-1.2);
\draw[fill=white] (0.8,-1.2) rectangle (2.2,-3.7) node[pos=0.5,rotate=270] {$U_{QRACM}(T_0)$};
\node at (2.6,0.8) {$\bullet$};
\node at (2.6,-0.8) {$\oplus$};
\draw (2.6,0.8) -- (2.6,-0.8);
\draw (3.7,-0.8) -- (3.7,-1.2);
\draw[fill=white] (3.0,-1.2) rectangle (4.4,-3.7) node[pos=0.5,rotate=270] {$U_{QRACM}(T_1)$};
\node at (3.7,-0.8) {$\bullet$};
\draw (4.8,0.8) -- (4.8,-0.8);
\node at (4.8,0.8) {$\bullet$};
\draw[fill=white] (4.8,0.0) circle (0.1);
\end{tikzpicture}
\caption{Recursive controlled unary QRACM from \cite{PRX:BGB+2018}.}\label{fig:recursive-qracm}
\end{figure}

The cost is $O(N\lg N)$ gates and the depth can be as low as $O(N\lg \lg N)$ if we use auxiliary qubits for the address comparison. This is already nearly optimal. 

 \cite{NSR:ParPetRhe19} present approximately this version, while the optimized version of~\cite{PRX:BGB+2018} use an efficient recursive structure of \emph{controlled} QRACM. The base case, a one-element controlled look-up, is just a CNOT. For the recursive case, see \Cref{fig:recursive-qracm}, where $T_b$ is a table of size $N/2$ defined from a table $T$ by taking all elements whose address starts with the bit $b\in\{0,1\}$. This uses $\lg N$ ancilla qubits and a low-T AND gate to give a T-count of $4N-8$, with a total gate cost of $O(N)$.

Attractive features of this circuit are its low qubit count and its simplicity. Even better, if the table has multi-bit words, then the base case -- a controlled look-up in a one-element table -- is just a multi-target CNOT gate. This has constant depth with only Clifford gates and measurements, so the overall costs for an $m$-bit table are $O(Nm)$ gates in depth $O(N)$.

\subsection{Bucket-Brigade}\label{sec:circuit-bb}
Bucket-brigade QRAM~\cite{PRL:GioLLoMac08} was a groundbreaking approach to QRAM that laid the basis for almost all hardware-QRAM since. It is not often considered in the circuit model (though see~\cite{PRA:PalOumBas20,PRX:HLGJ21,TQE:MatGheMos2020}).

The core idea is a binary tree, and we imagine ``sending in'' address bits. At each node in the tree, the first bit to reach that node sets the state of the node to either ``left'' or ``right'', depending on the whether the bit was $0$ or $1$. Any subsequent bit reaching that node is directed either left or right, accordingly.

In the circuit model, we copy the description of~\cite{PRX:HLGJ21}, shown in \Cref{fig:bucket-brigade-circuit}. Each node of the routing tree has two qubits: a routing qubit $\ket{r_i}$ and a control qubit $\ket{c_i}$, all of which are initialized to $\ket{0}$. 

To perform a memory access, we start by swapping the top (most significant) bit of the address into the control qubit at the root node of the tree. We then swap the second-highest address bit into the \emph{routing} qubit at the root node, then apply the routing circuit shown in \Cref{fig:bucket-brigade-route}. This swaps the routing qubit (which now contains the state of the second-highest address qubit) into the control qubit of either the left or right node in the next level, depending on the state of the control qubit of the root node. 

We then repeat this process: we swap the $i$th address bit into the routing qubit at the root, then apply the routing circuit to all nodes at the $0$th level, then the $1$st level, up level $i-1$. Then, at the $i$th level, we apply the control circuit in Figure~\ref{fig:bucket-brigade-control}. This puts the address qubit into the control, so it can properly route the next address qubit.

At this point, there are many possible approaches to read the actual data. For QRACM, one could use $X$ gates to the $i$th leaf of the table if $T[i]=1$. We then apply the routing circuit to all nodes in each layer, but iterating from the bottom layer to the top. This will swap the state of the qubit at the desired address back to the root of the tree, where we can copy it out with a final CNOT. Then the entire circuit must be uncomputed.

To create QRAQM, if we treat the table of memory as the bottom layer of the tree, the same SWAP technique as for QRACM will enact a swap QRAQM.

\def\coffx{-12}
\def\coffy{7.5}
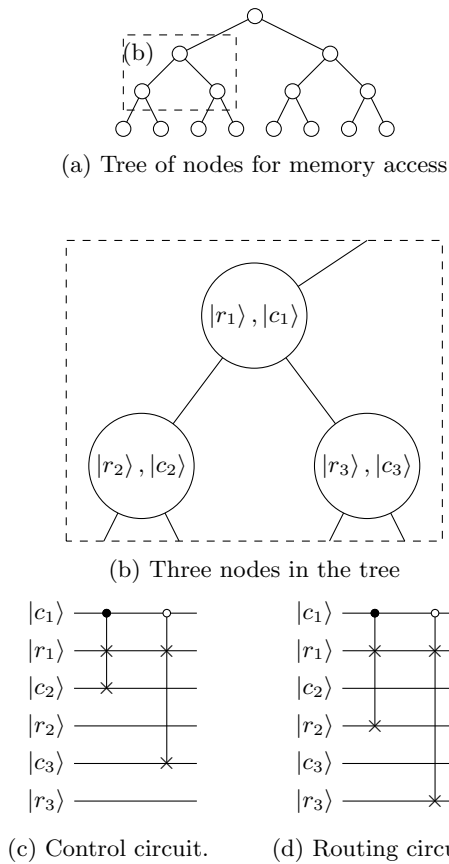
\begin{figure}
\begin{subfigure}{0.5\textwidth}
\begin{subfigure}{\textwidth}
\begin{tikzpicture}[scale=1.2]

\draw (5,5) -- (4,4.5);
\draw (5,5) -- (6,4.5);
\draw (4,4.5) -- (3.5,4);
\draw (4,4.5) -- (4.5,4);
\draw (6,4.5) -- (5.5,4);
\draw (6,4.5) -- (6.5,4);
\draw (3.5,4) -- (3.25,3.5);
\draw (3.5,4) -- (3.75,3.5);
\draw (4.5,4) -- (4.25,3.5);
\draw (4.5,4) -- (4.75,3.5);
\draw (5.5,4) -- (5.25,3.5);
\draw (5.5,4) -- (5.75,3.5);
\draw (6.5,4) -- (6.25,3.5);
\draw (6.5,4) -- (6.75,3.5);

\draw[fill=white] (5,5) circle (0.1);
\draw[fill=white] (4,4.5) circle (0.1);
\draw[fill=white] (6,4.5) circle (0.1);
\draw[fill=white] (3.5,4) circle (0.1);
\draw[fill=white] (4.5,4) circle (0.1);
\draw[fill=white] (5.5,4) circle (0.1);
\draw[fill=white] (6.5,4) circle (0.1);
\draw[fill=white] (3.25,3.5) circle (0.1);
\draw[fill=white] (3.75,3.5) circle (0.1);
\draw[fill=white] (4.25,3.5) circle (0.1);
\draw[fill=white] (4.75,3.5) circle (0.1);
\draw[fill=white] (5.25,3.5) circle (0.1);
\draw[fill=white] (5.75,3.5) circle (0.1);
\draw[fill=white] (6.25,3.5) circle (0.1);
\draw[fill=white] (6.75,3.5) circle (0.1);
\draw[dashed] (3.25,3.75) rectangle (4.75,4.75);
\node at (3.45,4.5) {(b)};
\end{tikzpicture}
\caption{Tree of nodes for memory access}\label{fig:bucket-brigade-tree}
\end{subfigure}
\vspace{2em}

\begin{subfigure}{\textwidth}
\begin{tikzpicture}[scale=1.2]
\draw[dashed] (2.5,5.5) rectangle (7.5,9.5);
\draw (5,8.5) -- (3.5,6.5);
\draw (5,8.5) -- (6.5,6.5);
\draw (5,8.5) -- (6.5,9.5);
\draw (3.5,6.5) -- (3,5.5);
\draw (6.5,6.5) -- (7,5.5);
\draw (3.5,6.5) -- (4,5.5);
\draw (6.5,6.5) -- (6,5.5);
\draw[fill=white] (5,8.5) circle (0.7) node (n1) {$\ket{r_1},\ket{c_1}$};
\draw[fill=white] (3.5,6.5) circle (0.7) node(n2) {$\ket{r_2},\ket{c_2}$};
\draw[fill=white] (6.5,6.5) circle (0.7) node(n3) {$\ket{r_3},\ket{c_3}$};
\end{tikzpicture}
\caption{Three nodes in the tree}\label{fig:bucket-brigade-subtree}
\end{subfigure}
\end{subfigure}

\begin{subfigure}{\columnwidth}
\begin{subfigure}{0.4\textwidth}
\begin{tikzpicture}
\node at (\coffx,\coffy) (c1) {$\ket{c_1}$};
\node at (\coffx,\coffy-0.5) (r1) {$\ket{r_1}$};
\node at (\coffx,\coffy-1.0) (c2) {$\ket{c_2}$};
\node at (\coffx,\coffy-1.5) (r2) {$\ket{r_2}$};
\node at (\coffx,\coffy-2.0) (c3) {$\ket{c_3}$};
\node at (\coffx,\coffy-2.5) (r3) {$\ket{r_3}$};
\draw (c1) -- ($(c1) + (2,0)$);
\draw (r1) -- ($(r1) + (2,0)$);
\draw (c2) -- ($(c2) + (2,0)$);
\draw (r2) -- ($(r2) + (2,0)$);
\draw (c3) -- ($(c3) + (2,0)$);
\draw (r3) -- ($(r3) + (2,0)$);
\draw (\coffx+0.8,\coffy) -- (\coffx+0.8,\coffy-1.0);
\draw (\coffx+1.6,\coffy) -- (\coffx+1.6,\coffy-2.0);

\draw[fill=black] (\coffx+0.8,\coffy) circle (0.05);
\draw[fill=white] (\coffx+1.6,\coffy) circle (0.05);
\draw[fill=black] (\coffx+0.8,\coffy) circle (0.05);
\node at (\coffx+0.8,\coffy-0.5) {$\times$};
\node at (\coffx+0.8,\coffy-1.0) {$\times$};
\node at (\coffx+1.6,\coffy-0.5) {$\times$};
\node at (\coffx+1.6,\coffy-2.0) {$\times$};
\end{tikzpicture}
\caption{Control circuit.}\label{fig:bucket-brigade-control}
\end{subfigure}
\begin{subfigure}{0.4\textwidth}
\begin{tikzpicture}
\node at (\coffx,\coffy) (c1) {$\ket{c_1}$};
\node at (\coffx,\coffy-0.5) (r1) {$\ket{r_1}$};
\node at (\coffx,\coffy-1.0) (c2) {$\ket{c_2}$};
\node at (\coffx,\coffy-1.5) (r2) {$\ket{r_2}$};
\node at (\coffx,\coffy-2.0) (c3) {$\ket{c_3}$};
\node at (\coffx,\coffy-2.5) (r3) {$\ket{r_3}$};
\draw (c1) -- ($(c1) + (2,0)$);
\draw (r1) -- ($(r1) + (2,0)$);
\draw (c2) -- ($(c2) + (2,0)$);
\draw (r2) -- ($(r2) + (2,0)$);
\draw (c3) -- ($(c3) + (2,0)$);
\draw (r3) -- ($(r3) + (2,0)$);
\draw (\coffx+0.8,\coffy) -- (\coffx+0.8,\coffy-1.5);
\draw (\coffx+1.6,\coffy) -- (\coffx+1.6,\coffy-2.5);

\draw[fill=black] (\coffx+0.8,\coffy) circle (0.05);
\draw[fill=white] (\coffx+1.6,\coffy) circle (0.05);
\draw[fill=black] (\coffx+0.8,\coffy) circle (0.05);
\node at (\coffx+0.8,\coffy-0.5) {$\times$};
\node at (\coffx+0.8,\coffy-1.5) {$\times$};
\node at (\coffx+1.6,\coffy-0.5) {$\times$};
\node at (\coffx+1.6,\coffy-2.5) {$\times$};
\end{tikzpicture}
\caption{Routing circuit.}\label{fig:bucket-brigade-route}
\end{subfigure}
\end{subfigure}
\caption{Schematic of the overall tree structure of circuit bucket-brigade QRACM. In the tree structure in (a), the lines between nodes are only to help exposition and do not represent any physical device (unlike in hardware bucket-brigade QRACM; see Section~\ref{sec:bucket-brigade}). (b) shows how each node consists of 2 qubits, and labels them to show where the circuits in (c) and (d) apply their gates.}\label{fig:bucket-brigade-circuit}
\end{figure}

\paragraph{Cost.} Both the control and routing circuits require 2 controlled-SWAPs at each node. For the $i$th address bit, we must apply one of these circuits to \emph{every} node from the root down to the $i$th layer. The layered structure of gate applications gives a total count of $2(n\cdot 2^0+(n-1)2^1+(n-2)2^2+\dots +2^n)=O(N\log N)$, where $N=2^n$. 

Naively the depth is $O(\log^2 N)$, since each address bit requires looping over all previous layers. However, as soon the $i$th address bit has been routed into the 2nd level of the tree, address bit $i+1$ can be swapped into the root and start that layer. Thus, the depth is $O(\log N)$.

This circuit requires $2(N-1)$ ancilla qubits just for the routing tree, and may require more for the data itself, depending on how the the final read is designed.

Overall, the costs are log-linear, but worse than unary QRACM in constant factors. While the depth is only logarithmic in $N$ (the best known), there are many approaches to a log-depth circuit QRAM \cite{JCE:CCJR21,ARXIV:CarThei18,ARXIV:PBKNKB2022}.

\paragraph{Error scaling.}
An advantage of the bucket-brigade approach is the favourable error scaling, proven in \cite{PRX:HLGJ21} and summarized in \Cref{sec:bb-errors}. Each node only needs errors below $O(1/\text{polylog}(N))$, rather than $O(1/N)$.

In a surface code architecture, this favourable error scaling means we could reduce the code distance for the qubits in the bucket-brigade QRACM. This might noticeably decrease overall depth, physical qubits, and perhaps allow easier magic state distillation. Of course, it would need to exceed the losses from the less efficient circuit design compared to (e.g.) a unary QRACM. Further, it would only benefit applications where the depth of the QRACM access is more important the qubit count, since the bucket-brigade and its error scaling depend on a fast, wide circuit. Depth-limited cryptographic attacks are an appealing candidate application. Quantifying where this advantage lies is a fascinating question for future work.

\subsection{Select-Swap}\label{sec:select-swap}
\cite{ARXIV:LowKliSch2018} describe a technique that reduces T-count even further, at the expense of more ancilla qubits (which can be in an arbitrary state, however). The technique appears also in \cite{ARXIV:CDSSBZ2022} and \cite{QUANTUM:BGMMB2019}. This method is analogous to ``paging'' in classical memory. In most classical memory devices, it is just as efficient to retrieve larger chunks of memory (called ``pages'') as it is to retrieve single bits. Hence, the memory device retrieves a page and a second device will extract the specific word within that page. Analogously, this QRACM circuit uses a QRACM circuit on a larger page of the classical table, then uses a QRAQM access on that page. 

Recall that the T-gate cost of unary QRACM (and depth) is independent of the length of each word in the memory table. From any table $T$ with $N$ entries in $\{0,1\}$, we can construct a table $T'$ with $N/2^\ell$ entries in $\{0,1\}^\ell$, simply by grouping $\ell$ consecutive entries -- a page -- as one word in $T'$. Thus, to begin access to $T$, we start with a unary QRACM access to $T'$, writing the output in an $2^\ell$-bit auxiliary register. We do not need the last $\ell$ bits of the address to do this, since the page we retrieve has every element that the last $\ell$ bits could address.

Once this is done, we want to extract the precise bit that the last $\ell$ bits address. Since the required page is already in the auxiliary register, we use a QRAQM access to do this. The address in the page for this QRAQM access is the last $\ell$ bits of the original address, and we write the output to the original output register.

After this, we uncompute the QRAQM access and uncompute the original unary QRACM access.

To see why this works, notice that after the first QRACM access, the $2^\ell$ auxiliary bits are now in a superposition, with the contents of these bits entangled with the high bits of the address register. Thus, a \emph{single} QRAQM access extracts the correct output for each of these states in superposition. 

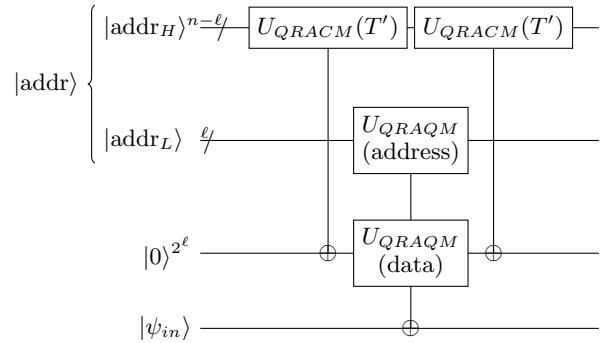
\begin{figure}[H]
\begin{tikzpicture}[scale=1.2]
\node[left] at (2.5,0) {$\ket{\psi_{in}}$};
\node[left] at (2.5,1) {$\ket{0}^{2^\ell}$};
\node[right] at (1.1,2.5) {$\ket{\text{addr}_L}$};
\node[right] at (1.1,4) {$\ket{\text{addr}_H}$};
\draw [decorate, decoration = {brace}] (1.1,2.2) --  (1.1,4.3);
\node[left] at (1,3.25) {$\ket{\text{addr}}$};

\node at (2.6,2.5) {$^\ell\hspace{-0.3em}/$};
\node at (2.6,4) {$^{n-\ell}\hspace{-0.3em}/$};

\draw (2.5,0) -- (6.8,0);
\draw (2.5,1) -- (6.8,1);
\draw (2.5,2.5) -- (6.8,2.5);
\draw (2.5,4) -- (6.8,4);

\draw (3.2,4) -- (3.2,1);
\node at (3.2,1) {$\oplus$};
\node[fill=white,draw=black,rectangle,align=center,rotate=270] at (3.2,4) {$U_{QRACM}(T')$};

\draw (4.3,2.5) -- (4.3,0);
\node at (4.3,0) {$\oplus$};
\node[fill=white,draw=black,rectangle,align=center] at (4.3,1) {$U_{QRAQM}$\\ (data)};
\node[fill=white,draw=black,rectangle,align=center] at (4.3,2.5) {$U_{QRAQM}$\\ (address)};

\draw (5.4,4) -- (5.4,1);
\node at (5.4,1) {$\oplus$};
\node[fill=white,draw=black,rectangle,align=center,rotate=270] at (5.4,4) {$U_{QRACM}(T')$};

\end{tikzpicture}
\caption{Select-swap QRACM access. $T'$ is the table formed from $T$ by combining each page of $2^\ell$ consecutive bits of $T$ into one word in $T'$.}\label{fig:partial-write}
\end{figure}

\begin{figure*}
\centering
\begin{tikzpicture}[scale=1.0]

\draw [decorate, decoration = {brace}] (0,6.2) -- (3.5,6.2) ;
\draw [decorate, decoration = {brace}] (3.6,6.2) -- (4.4,6.2);
\draw [decorate, decoration = {brace}] (4.5,6.2) -- (8,6.2);
\node[above] at (1.75,6.25) {(A)};
\node[above] at (4,6.25) {(B)};
\node[above] at (6.25,6.25) {(C)};

\draw[decorate,decoration = {brace}] (-1,-2.3) -- (-1,3.7);
\draw[decorate,decoration = {brace}] (-1,3.8) -- (-1,5.5);
\draw[decorate,decoration = {brace}] (-1,5.6) -- (-1,6);
\node[above,rotate=90] at (-1.05,0.7) {memory};
\node[above,rotate=90] at (-1.05,4.75) {address};
\node[left,rotate=0] at (-1.05,5.8) {output};

\foreach \i in {0,...,7} {
	\node[left] at (0,3.5-0.8*\i) {$\ket{x_\i}$};
	\node[right] at (8,3.5-0.8*\i) {$\ket{x_\i}$};
	\draw (0,3.5-0.8*\i) -- (8,3.5-0.8*\i);
	\node at (1,3.5-0.8*\i) {$\times$};
	\node at (7,3.5-0.8*\i) {$\times$};
}
\node[left] at (0,4.3) {$\ket{i_0}$};
\node[left] at (0,4.8) {$\ket{i_1}$};
\node[left] at (0,5.3) {$\ket{i_2}$};
\node[left] at (0,5.8) {$\ket{\psi_{in}}$};

\foreach \i in {0,...,3} {
	\draw (0,4.3+\i*0.5) -- (7.7,4.3+\i*0.5);
}

\foreach \j in {0,1}{
	\foreach \i in {0,...,3} {
		\node at (0.6+6.8*\j,3.8-1.6*\i) {$\oplus$};
		\draw[fill=black] (1+6*\j,3.8-1.6*\i) circle (0.05);
		\draw (1+6*\j,3.8-1.6*\i) -- (1+6*\j,2.7-1.6*\i);
		\node at (1.4+5.2*\j,3.8-1.6*\i) {$\oplus$};
		\node at (2.2+3.6*\j,3.5-1.6*\i) {$\times$};
		\node[left] at (0.3+7.9*\j,3.8-1.6*\i) {$\scriptscriptstyle{\ket{0}}$};
		\draw (0.3,3.8-1.6*\i) -- (7.7,3.8-1.6*\i);
	}
	\draw[fill=black] (0.6+6.8*\j,4.3) circle (0.05);
	\draw (0.6+6.8*\j,4.3) -- (0.6+6.8*\j,3.8-1.6*3);
	
	\draw[fill=black] (1.4+5.2*\j,4.3) circle (0.05);
	\draw (1.4+5.2*\j,4.3) -- (1.4+5.2*\j,3.8-1.6*3);
}
\foreach \j in {0,1}{
	\foreach \i in {0,...,1} {
		\node at (1.8+4.4*\j,3.8-3.2*\i) {$\oplus$};
		\draw[fill=black] (2.2+3.6*\j,3.8-3.2*\i) circle (0.05);
		\draw (2.2+3.6*\j,3.8-3.2*\i) -- (2.2+3.6*\j,1.9-3.2*\i);
		\node at (2.6+2.8*\j,3.8-3.2*\i) {$\oplus$};
		\node at (3.4+1.2*\j,3.5-3.2*\i) {$\times$};
	}
	\draw[fill=black] (1.8+4.4*\j,4.8) circle (0.05);
	\draw (1.8+4.4*\j,4.8) -- (1.8+4.4*\j,3.8-1.6*2);
	
	\draw[fill=black] (2.6+2.8*\j,4.8) circle (0.05);
	\draw (2.6+2.8*\j,4.8) -- (2.6+2.8*\j,3.8-1.6*2);
	
	\draw[fill=black] (3.4+1.2*\j,5.3) circle (0.05);
	\draw (3.4+1.2*\j,5.3) -- (3.4+1.2*\j,0.3);
}

\draw[fill=black] (4,3.5) circle (0.05);
\draw (4,3.5) -- (4,5.8);
\node at (4,5.8) {$\oplus$};

\end{tikzpicture}
\caption{Fanout-and-swap QRAQM for $N=8$. (A) uses a binary tree of swaps to move $\ket{x_i}$ into the register for $\ket{x_0}$, where (B) copies it out (replacing this CNOT with a SWAP would produce a SWAP-QRAQM). Then (C) restores the memory register to its original state.}\label{fig:fanout-and-swap}
\end{figure*}
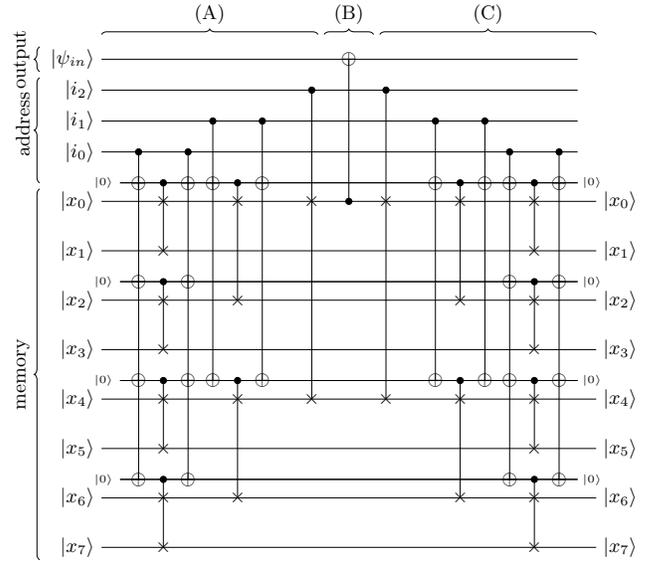

\paragraph{Cost.} The cost of the first unary QRACM access is $O(N/2^\ell)$ T-gates, in depth $O(N/2^\ell)$ generally. For the second QRAQM acess, it needs $\tilde{O}(2^\ell)$ gates and has depth $O(\ell)$. To minimize T-gates, a page size of $2^\ell\approx \sqrt{N}$ gives a total T-cost and depth of $O(\sqrt{N})$. 

However, notice that the total gate cost is still proportional to $N$, since writing each word of $T'$ requires $2^\ell$ CNOT gates. Thus, the majority of the gates in this circuit are used for multi-fanout CNOTs. If these are cheap, then this circuit is cheap overall.

The qubit requirements for the T-gate-optimal parameters are also proportional to $\sqrt{N}$. Using other parameters, this approach provides a nice extrapolation between bucket-brigade and unary QRACM, able to reach almost any combination of depth and width whose product is $\Theta(N)$. Even better, this method can use dirty qubits (i.e., qubits which may hold data necessary for other parts of the algorithm) for all but one word in each page.

The circuit in the appendix of \cite{QUANTUM:BGMMB2019} achieves this asymptotic cost, but they use a measurement-based uncomputation in both the QRAQM and QRACM access. This means they measure each ancilla qubit in the $\{\ket{+},\ket{-}\}$ basis, giving a binary string $\vec{v}$. If the ancilla qubits were in state $\ket{b}$ for $b\in\{0,1\}^m$, then that state acquires a phase of $-1$ if and only if $\vec{v}\cdot \vec{b}$ has odd parity (and extend this linearly). They clean up this phase with a second look-up. Since this phase is a single bit, this is cheaper than a forward computation when the words of $T$ are larger than one bit. For single-bit words they also reduce uncomputation costs by $\frac{1}{2}$.

This optimizes the non-asymptotic gate counts. They use a QRAQM circuit based on controlled swaps, rather than the bucket-brigade; see \Cref{fig:fanout-and-swap} for an analogous technique that keeps depth and T-count small by using high-fanout CNOT gates. The circuits in \cite{QUANTUM:BGMMB2019} achieve the lowest T-count of any QRACM circuit to date. 

\cite{ARXIV:PBKNKB2022} provide a T-depth optimized method to convert binary to unary encodings which could likely be combined with \cite{QUANTUM:BGMMB2019}. \cite{ARXIV:CDSSBZ2022} decompose the SWAP gates to apply T-gates in parallel to achieve a similarly low T-depth.

\subsection{Parallel QRAM}
Taking a step back to the algorithm calling QRAM, many algorithms call the gate repeatedly and thus could benefit from parallelization. We define such a gate as follows:
\begin{definition}
A parallel quantum random access classical memory is a set of unitaries $U_{PQRACM}(T)$, parameterized by a table $T\in\{0,1\}^N$, such that for any $k$ address registes $\ket{i_1}\dots\ket{i_k}$, with $1\leq i_j\leq N$ for all $1\leq j\leq k$, 
\begin{align}
U_{PQRACM}&\ket{i_1}\ket{0}\ket{i_2}\ket{0}\dots\ket{i_k}\ket{0}\nonumber\\
=&\ket{i_1}\ket{T_{i_1}}\ket{i_2}\ket{T_{i_2}}\dots\ket{i_k}\ket{T_{i_k}}.
\end{align}
\end{definition}
We can define parallel QRAQM similarly, with a shared external register $\ket{T_0}\dots \ket{T_{N-1}}$.

Clearly one can enact this by independently applying QRAM to each address register, for a total gate cost of $O(kN)$. However, since the calls share the same table, the cost can be reduced. \cite{ROYSOCA:BBGH+13} solve this with sorting networks.

To briefly explain: they augment the table $T$ to a table $T'$ by adding the address of each element to that element, so each entry is $(i, T_i, 1)$. Then they augment the address registers to $(i_j,0,0)$ and write this entire table into quantum memory. After this, they sort all of the memory registers \emph{and} the address registers together, using an ordering based on the first address register, breaking ties with the final flag register (which is $1$ for memory elements and $0$ for query elements).

Sorting with this ordering ensures that the quantum memory is laid out as
\begin{equation}
\dots\ket{i,T_i,1}\ket{i+1,T_{i+1},1}\ket{i+2,T_{i+2},1}\ket{i+2,0,0}\dots
\end{equation}
where one can see that a query with address $i+2$ has been placed adjacent to the memory element at address $i+2$, which the query needs to access. From there, a careful series of controlled CNOTs copies each memory element into any adjacent queries. This takes some care: at this point the control program does not know whether a given register of qubits contains a memory element or a query, so it must use a control made from the last flag bit of both words as well as whether they have the same address. Moreover, multiple queries may access the same memory element, so the copying must carefully cascade memory elements forwards whenever adjacent registers have the same adders.

Once this is all done, they unsort the network. Now the address registers have their requested data. A quantum sort must keep a set of ancilla qubits recording the outcome of the comparisons of the sort, for reversibility, so these are used and uncomputed to reverse the sort. In some sense, the newly copied data just ``follows'' the rest of the sort.

The cost of this scheme is primarily the cost of the sorting network, which in turn is governed by the connectivity of the quantum computer architecture. \cite{ROYSOCA:BBGH+13} emphasize that with a connectivity that they call the ``hypercube'' (defined to mean that if we address each register by a bitstring, each is connected to registers whose address differs in only one bit), sorting $N$ memory elements and $k$ queries requires only $O((N+k)\log(N+k))$ gates and depth $O(\text{polylog}(N+k))$. However, with a two-dimensional nearest-neighbour architecture (e.g., a surface code with no flying qubits), it requires $O((N+k)^{3/2})$ gates and $O((N+k)^{1/2})$ depth. It remains debatable whether hypercube connectivity is feasible; we discuss this some in \Cref{sec:latency}.

\subsection{Data Structures}
Classical computers use RAM to store and access data in larger data structures such as linked lists, hash tables, etc. Analogous quantum data structures would inherently use QRAQM, such as the quantum radix trees in \cite{PQC:BJLM2013}. However, with a QRAQM access cost proportional to $N$ in the circuit model, bespoke data structure circuits likely exist. 

As an example, \cite[Section 4]{CRYPTO:JaqSch19} performs insertion into a sorted list by making $N$ copies of an input with a fanout gate, then simultaneously comparing \emph{all} elements of the list to a copy of the input, then shifting all elements of the array. This requires $N$ comparisons, but no explicit QRAQM access. \cite{ARXIV:Gidney22} takes a similar approach to build a quantum dictionary.

We encourage algorithms requiring quantum data structures to use similar custom circuits to operate on the data, rather than building them from QRAQM accesses.

\section{Hardware QRAM}\label{sec:gate-qram}
Here we consider a specialized device for QRAM, in the model of \emph{hardware-QRAM}, that is not built from general-purpose qubits. The gate lower bounds for circuit QRAM are unpalatable for certain applications, but there is no reason that all parts of a future quantum computer are built as circuits in a small gate set. Like GPUs or RAM in classical computers today, it's easy to imagine specialized devices for specific tasks. The goal of hardware-QRAM is something much more efficient than implementing QRAM access as a circuit. 

This section asks whether such a specialized device can be built and run efficiently. Not all tasks can have specialized hardware; it's unreasonable to expect a sub-module which solves 3-SAT in quantum polynomial time. Our question is whether QRAM access is a simpler problem. 

To give a quick sense that QRACM is generally difficult, consider that if the table is $T=[0,0,0,1]$, then QRACM access to this table is equivalent to a Toffoli gate. Hence, Clifford+QRAM is a universal gate set.

In this section we lay out several issues that any proposed QRAM device must address to be competitive with circuit-based QRAM.

\subsection{Active vs. Passive QRAM}\label{sec:active-vs-passive}
Recall that in the memory peripheral model, we cost quantum operations by the number of interventions by the classical controller. Here we extend that slightly to include the energy cost of the external controller. For example, we would count the cost to run a laser.

An \emph{active} QRAM gate is one with cost in this sense that is at least proportional to $N$ for each QRAM access. Active hardware-QRAM will still fail asymptotically in the same way as active circuit QRAM. Recalling the scale of memory in \Cref{sec:applications}, even if the constant in this asymptotic term is tiny, it quickly adds up when the memory sizes are $2^{50}$ bits or larger. 

There may be a constant-factor advantage to active hardware-QRAM, but \Cref{sec:circuit-qram} shows that the constants for circuit QRAM are already low. Moreover, when comparing a quantum algorithm using QRAM to a classical algorithm (such as with machine learning), the quantum algorithm must overcome a great deal of overhead to have advantage (e.g., Google's quantum team estimates that a quadratic advantage is insufficient in the near term~\cite{PRXQ:BMJ+2021}; error correction nearly eliminates the advantage of quantum lattice sieving~\cite{AC:AGPS2020}). If QRAM has even a small linear cost, algorithms like database HHL and database QAA will only have -- at best -- a constant-factor advantage over classical.

For these reasons, we discount any active hardware-QRAM, though we cannot rule out a constant-factor advantage from such a device.

Notice that passive QRAM (with $o(N)$ cost for each access) needs both the readout and routing stages to be passive. Even if our readout is quite low-energy -- say, if we can read directly from a non-volatile classical memory -- if the routing stage has cost $\Omega(N)$, the entire QRAM access would have cost $\Omega(N)$.

\subsection{Path Independence}
Accessing memory in superposition means that we must not destroy any of the states in superposition, which means that no external system, including the controller enacting the memory, can learn \emph{anything} about the memory access itself. A major obstacle here is the path in which the memory access travels.

In contrast, consider one aspect of classical memory access that QRAM cannot emulate. If there is a tree of routing nodes as in the bucket-brigade QRAM, a classical computer can keep all of the nodes in a low- or zero-power ``idle'' mode, where a node activates only when it receives a signal from its parent node. Such a device would consume energy only in proportion to $O(\log N)$, not $O(N)$. QRAM cannot do this because in general it needs to be able to act on an arbitrary superposition of addresses, and we cannot know the support of the given superposition non-destructively. If we learn that a path was taken, this collapses the state.
Thus, if a QRAM gate requires any active intervention (such as a controlled swap like the circuit QRAM), then the controller \emph{must} apply that intervention to \emph{all} nodes, immediately bringing the (parallelizable) cost to $\Omega(N)$.

In their description of bucket-brigade QRAM, \cite{PRL:GioLLoMac08} describe an ``active'' path through the tree for each address. If a node requires some external intervention (such as energy) to perform its role in the device, then \emph{all} nodes must receive that energy or intervention. That is, the ``activation'' of a node is all-or-nothing: if we cannot build a device where all nodes can route data passively (such as \cite{THESIS:Cadellans15} attempts), then all nodes must be active during a memory access.

\subsection{Hamiltonian Construction}\label{sec:hamiltonian-construction}

If a QRAM gate is \emph{completely} passive, i.e., requires no external intervention, then this implies that the Hamiltonian describing the entire QRAM device must be constant throughout the operation of the gate, since any change to the Hamiltonian requires external intervention. This means there must be some Hamiltonian $H_{QRACM}(T)$ and time $t$ such that (let $\hbar = 1$)
\begin{equation}\label{eq:time-indie-ham}
U_{QRACM}(T) = e^{i H_{QRACM}(T)t}.
\end{equation}

We now ask how difficult it will be to establish such a Hamiltonian.
If we need to reconstruct the Hamiltonian every time we access the QRAM, then the number of terms (and hence the number of interventions) gives a lower bound on the cost.

If changing the system's Hamiltonian needs $\Omega(N)$ interventions (e.g., if we need photon pulses on all nodes of a routing tree), then our QRAM is active. We could instead try to design a system where we can evolve to the desired Hamiltonian with fewer interventions, but we will now sketch the limitations of this approach. 

Firstly, we show that a Hamiltonian implementing QRACM access must be fairly complex, by recalling previous results on simulating Hamiltonians with circuits, and using the previous circuit lower bound from \Cref{thm:circuit-qram-lower-bound}.

\begin{lemma}\label{lem:hamiltonian-cost}
    Let $H=\sum_{j=0}^{n-1} a_jh_j$ be a Hamiltonian where each $h_j$ is an $m$-qubit Pauli string, with $n$ such terms, acting on $\mathsf{W}$ qubits. Let $\vec a := \begin{pmatrix}a_0 & ... a_{n-1} \end{pmatrix}^T$. If $e^{-i H t}$ enacts a QRACM unitary for a memory with $N$ rows (i.e., with $\log_2 N$ address bits) then to be able to implement any possible QRACM table, $N\in \tilde{O}(\mathsf W n \lnorm{\vec a}_1 t)$. 
\end{lemma}
\begin{proof}
    We will first show how existing results allow for any such Hamiltonian to be simulated in the circuit model. Each Pauli string is unitary, so it is a trivial $(1,0,0)$-block-encoding of itself. We can straightforwardly construct a state preparation unitary $U$ with $U\ket{0}=\sum_{j=0}^{n-1}\sqrt{a_j}\ket{j} / \lnorm{\sum_{j=0}^{n-1}\sqrt{a_j}\ket{j}}_2$, using $O(n)$ circuit depth. Let $b := \ceil{\log_2(n)}$. Then, using \cite[Lemma 52]{gilyen2018quantum} we can obtain a $(\lnorm{\vec a}_1,b,0)$-block-encoding of $H$ with $O(n)$ circuit of depth, and then \cite[Theorem 58]{gilyen2018quantum} gives us a $(1,b+2,\epsilon)$-block-encoding of $e^{iHt}$ with a total circuit depth of $O(n\log_2(n)(\lnorm{\vec a}_1 t + \ln(1/\epsilon))$ using $O(\log_2(n))$ auxiliary qubits. 

    Clearly, this block-encoding is an $\epsilon$-approximate implementation of the QRACM gate $e^{i H t}$ (as per~\cref{thm:circuit-qram-lower-bound}). Moreover, this circuit has $\tilde{O}(\mathsf W n (\lnorm{\vec a}_1 t + \ln(1/\epsilon))$ total gates. Picking a fixed constant accuracy $\epsilon \le 1/2$, we get the total gate count of $\tilde{O}(\mathsf W n \lnorm{\vec a}_1 t)$. For a fixed precision,~\cref{thm:circuit-qram-lower-bound} gives the lower-bound of $G \in \tilde \Omega (N)$, and so we can conclude that circuits constructed in this way could only approximate all possible QRACM tables if the table size is at most $N\in\tilde{O}(\mathsf{W}n\lnorm{\vec a}_1 t)$, meaning the Hamiltonians constructed in this way can also implement only these QRACM tables.

\end{proof}

\Cref{lem:hamiltonian-cost} shows that if our computational device pays a cost for each term it adds to the Hamiltonian, and the device must construct the Hamiltonian for each QRAM access, the cost scales just as badly as active QRAM.

One way to escape this limitation would be to construct a device where we only need to construct the Hamiltonian \emph{once}, so the number of interventions ($n$) does not govern the cost to \emph{use} the QRAM. One approach would be a Hamiltonian such that there is an ``idle'' eigenstate and we only need to perturb the system to start it. For example, one of the proposals for bucket-brigade QRAM imagines excitable atoms at each node, excited by input photons. Adding a photon would start the system on the desired evolution.

Feynman's original proposal for a quantum computer~\cite{OPTICA:Feynman1985} (with Kitaev's refinement~\cite{BOOK:KitSheVya2002}) was a large Hamiltonian that evolved on its own -- i.e., ballistic computation. Practical difficulties of this scheme mean it has not been experimentally explored (to the best of our knowledge). Some interesting questions that bear on QRAM include:
\begin{enumerate}
\item
Can we add more than one 2-local term to the Hamiltonian with one intervention from the controller and energy cost $O(1)$? If not, then constructing a Hamiltonian for QRA\textbf{Q}M will cost $\Omega(N)$, as the Hamiltonian will need at least $\Omega(N)$ terms just to interact with all bits of memory.
\item
What is a physically realistic error model for the Hamiltonian terms themselves? Unlike discussions about the error in the input state to a ballistic computation, here we consider the errors in the Hamiltonian itself, i.e., how even a perfect input might stray from the desired computational trajectory. Certain proposed QRAM devices partially address this; see \Cref{sec:bb-errors}.
\end{enumerate}

Overall, we highlight the notion of a single Hamiltonian for the QRAM to exemplify the extreme engineering challenges of such a device, which are far above those of even a large-scale fault-tolerant quantum computer.

\section{Errors}\label{sec:errors}
Defining the exact error rate of QRAM is tricky, since there are several definitions of the error rate of a channel (e.g., average fidelity vs. minimum fidelity). For applications like quantum look-ups, average fidelity is likely more relevant because the look-up accesses many addresses simultaneously, whereas for database QAA, we need something like minimum fidelity.

Generally, we will use these metrics to measure the quality of the QRAM gates and their components, though somewhat imprecisely. ``Error rate'' will refer to the infidelity, $1-F(\cdot,\cdot)$. For the overall QRAM, we want to compare the output of our real QRAM device (modeled as a quantum channel) to the ideal behaviour: a unitary channel defined by $U_{QRACM}(T)$ or $U_{QRAQM}(T)$.

\paragraph{Error rates. }
It's well-known that a quantum algorithm with $G$ gates can succeed if each gate has an error rate of $O(1/G)$, and this is often a necessary condition as well (with even more stringent conditions if idle qubits are also error-prone). The error rates in curent quantum devices are too high for circuits with more than a few hundred gates, and experts seem to generally expect the future to centre around fault-tolerant error-correcting codes, not spectacular jumps in physical error rates~\cite{GRI:MosPia2021}. 

For this reason, if a QRAM technology requires error rates per \emph{component} to scale inversely with the number of bits of memory we do not consider this an effective approach to passive hardware-QRAM, regarding such components as just as unlikely as quantum computing without error correction. 

Where this might not matter is applications with only a small number of QRAM accesses. As an example, if an algorithm needs $10^{6}$ Hadamard gates but only $10^2$ QRAM accesses, then physical error rates of $10^{-3}$ are too high for Hadamard gates but just fine for the QRAM. In this case we would not need error correction, and in fact we could decode a logical state to a physical state, apply the QRAM, then re-encode to the logical state. This is a risky process for a quantum state, but if we make few total accesses to the QRAM then we also do few total decodings. This might be helpful, but we stress again that this limits the potential total size of the QRAM.

Bucket-brigade QRAM specifically only needs per-component error rates to scale inversely with the \emph{log} of the number of bits of memory~\cite{PRX:HLGJ21}. The overall QRAM error cannot be less than the physical error rates, so this is not fault-tolerant at arbitrarily large scales, but we discuss how this might reach practical scales in \Cref{sec:bb-errors}.

For the large scales of QRAM which must be error-corrected, we consider two approaches: error correction within the QRAM, or QRAM within the error correction.

\subsection{Error-corrected components}
Suppose that we attempt to suppress errors in each component of a QRAM device using some form of error correction, such as~\cite{CLEO:CDELE21} propose for their method: ``scaling up the qRAM necessitates further exploration in converting each tree node to a logical qubit and adapting quantum error correction.''

If this error correction requires active intervention, like a surface code, we have essentially recreated circuit QRAM.  Even if the QRAM device uses specialized hardware, active error correction implies intervention from the controller and energy dissipation for each component, pushing us back into the regime of active hardware-QRAM with a $\Omega(N)$ cost of access.

Passive error correction, where the system naturally removes errors by evolving to lower energy states, is dubious. There are no known methods in 2 or 3 physical dimensions and there are several broad impossibility results (see the survey in \cite{RMP:BLPSW2016}). Methods for passive error correction would be a massive breakthrough, but it would also not be enough for passive QRAM: passive QRAM also needs passive gates on the error-corrected components (see \Cref{sec:active-vs-passive}).

Because of these difficulties, it seems infeasible to build error-correction within the QRAM that does not lead to access costs at least proportional to $N$. Thus, to avoid strict error requirements, we want a QRAM gate such that the total error of each access relative to the error rate of each component is only $O(\text{polylog}(N))$. Using ``error'' imprecisely to give a sense of scale, if $N=2^{50}(\approx 10^{15})$ and the error scales as $\log^2 N$ as in \cite{PRX:HLGJ21}, we ``only'' need error rates of $4\times 10^{-7}$ in each component to achieve a total error rate of $10^{-3}$. 

\subsection{Error Correction}\label{sec:error-correction}
Instead of correcting the errors within the components of a QRAM device, we could instead attempt to correct the error-affected output of a noisy device. Because the set of possible QRACM gates includes the CNOT and Toffoli gates, the Gottesman-Knill theorem means we cannot expect to apply QRACM transversally to any error-correcting code for which Clifford gates are also transversal (such as a surface code). Here we imagine a few different non-exhaustive approaches:

\begin{enumerate}
\item
apply QRACM gates independently of the code and attempt to inject the results into another fault-tolerant quantum system;
\item
construct an error-correcting code for which QRACM is transversal, and inject \emph{other} gates into this code as needed;
\item
code switch from a more practical code to a QRACM-transversal code;
\item
find a non-transversal method to fault-tolerantly apply QRACM gates.
\end{enumerate}

We are pessimistic about QRAM error correction, for reasons we discuss below, but this area needs more research. Finding constructions for satisfactory codes, or proving more impossibility results, would shed more light on the feasibility of QRAM. 

\textbf{Transversal codes.}  There are a few issues with a QRACM-transversal code. 

The transversal gates for a code form a group. Once we have a transversal QRACM of at least $2$ bits of input, if the memory does not encode a linear function, we can fix certain inputs and pre-compose with $X$ gates to construct a Toffoli gate from the QRACM. Once we have transversal Toffoli gates, we can compose them to construct all possible QRACM, or more generally all classical reversible circuits. Thus, transversal QRACM is an ``all-or-nothing'' situation for a code. 

Having all possible QRACM gates transversal for a code is so powerful that it is likely impossible. Almost all codes are proven to fail at this \cite{QIC:NewShi2018}. This would also contain every level of the Clifford hierarchy, since the $n$-multi-controlled NOT is in the $n$th level and can be realized with QRACM (consider access to a table $T$ where only the $N-1$ entry is $1$, the rest are $0$). Using \cite[Theorem 2]{PRA:PasBen2015}, this means QRACM cannot be transversal for any family of subsystem codes that can correct any probability of qubit loss.

If such a code were to exist, it would also be unwieldy for computations, as it will require another procedure like state injection to perform $H$ gates (since Toffoli+$H$ is universal). Switching from a more practical code to a QRACM-specific code would need a series of intermediate codes to preserve the logical state, lest it decohere during the transition. 

To escape these no-go theorems, we could use a different code for the output qubit, so that we cannot compose the transveral gates. Using this strategy allows quantum Reed-Solomon codes to have transversal Toffoli gates, though extending the method to $n$-bit QRACM implies the output qubit needs to be encoded in $2^n$ physical qubits.

\subsection{Gate Teleportation}\label{sec:gate-teleportation}
Universal quantum computation in a surface code is possible via magic state distillation and gate teleportation. Several auxiliary qubits are prepared in a fixed state, the controller applies a physical T gate to each one, then distills them together with measurement to create a much higher fidelity state. Using this one state, the T gate can be teleported onto any qubit in the surface code.

For an analogous process with QRACM, we define the following gate teleportation channel, shown in \Cref{fig:qracm-distill}. Let $\Hil_Q$ be the address and output space of the QRACM, and let $\mathcal{Q}(T):B(\Hil_Q)\rightarrow B(\Hil_Q)$ be a channel representing the physical QRACM gate. We assume (using $\circ$ to mean function composition) $\mathcal{Q}(T) = \mathcal{N}_1\circ \mathcal{U}_{QRACM}(T)\circ \mathcal{N}_0$ for all $T$, meaning that our QRACM process is modelled as a perfect QRACM access sandwiched between two noise channels. We model in this way so that the noise channels are not dependent on the table itself, which aids in the proofs.

 We have $d$ distillation channels $\Phi_1,\dots, \Phi_d$, with $\Phi_i:B(\Hil_{i-1}\otimes \Hil_Q)\rightarrow B(\Hil_i\otimes \Hil_Q)$, and the entire state distillation process produces a state $\rho_{distill}$ as follows: 
\begin{align}
\rho_{distill}(T):=&\Phi_d\circ (I_{d-1}\otimes \mathcal{Q}(T))\circ\nonumber\\
&\Phi_{d-1}\circ(I_{d-2}\otimes \mathcal{Q}(T))\circ\nonumber\\
&\dots \nonumber\\
  &\Phi_2 \circ (I_1\otimes \mathcal{Q}(T))\circ \Phi_1(\ketbra{0}{0})\label{eq:qracm-distill}
\end{align}
Letting $\Hil_{QL}$ represent the space of a logically encoded input and output space, we have a final teleportation channel $\Psi_T:B(\Hil_d\otimes \Hil_Q\otimes \Hil_{QL})\rightarrow B(\Hil_{QL})$. We create a logical QRACM channel $\mathcal{Q}_L(T)$ by 
\begin{equation}
\mathcal{Q}_L(T)(\rho) = \Psi(\rho_{distill}(T)\otimes \rho).
\end{equation}
The goal is for $\mathcal{Q}_L(T)(\rho)$ to approximate a QRACM access on the logical state.

\begin{figure*}
\centering
\begin{tikzpicture}
\node[left] at (0,3) {$\ket{0}$};
\draw (0,3) -- (5,3);
\draw (1,2) -- (5,2);
\draw (0,1) -- (10.2,1);
\draw[fill=white] (0.5,1.5) rectangle (1.2,3.5) node[pos=0.5] {$\Phi_1$};
\draw[fill=white] (1.5,1.5) rectangle (2.5,2.5) node[pos=0.5] {$\mathcal{Q}(T)$};
\draw[fill=white] (2.8,1.5) rectangle (3.5,3.5) node[pos=0.5] {$\Phi_1$};
\draw[fill=white] (3.8,1.5) rectangle (4.8,2.5) node[pos=0.5] {$\mathcal{Q}(T)$};
\node at (5.3,3) {$\dots$};
\node at (5.3,2) {$\dots$};
\draw (5.6,3) -- (10,3);
\draw (5.6,2) -- (10,2);
\draw[fill=white] (5.8,1.5) rectangle (6.7,3.5) node[pos=0.5] {$\Phi_{d-1}$};
\draw[fill=white] (7,1.5) rectangle (8,2.5) node[pos=0.5] {$\mathcal{Q}(T)$};
\draw[fill=white] (8.3,1.5) rectangle (9,3.5) node[pos=0.5] {$\Phi_d$};
\draw[fill=white] (9.3,0.5) rectangle (10,3.5) node[pos=0.5] {$\Psi_T$};

\node[left] at (0,1) {$\rho_L$};
\node[right] at (10.2,1) {$\rho'_L$};

\node at (5,0) {Goal: $\rho_L'\approx U_{QRACM}(T)\rho_L U_{QRACM}(T)^\dagger$};

\end{tikzpicture}
\caption{Schematic of a QRACM disillation and teleportation. The channels $\Phi_1,\dots, \Phi_d$ attempt to distill the output of accesses to the physical QRACM table $T$, denoted $\mathcal{Q}(T)$. A final teleportation channel $\Psi_T$ attempts to use the resulting distilled state to ``teleport'' the QRACM gate onto a logical state input $\rho_L$.}\label{fig:qracm-distill}
\end{figure*}

In the appendix we prove \Cref{thm:no-qracm-distillation}, summarized approximately as:
\begin{theorem}\label{thm:no-qracm-distillation-informal}
Suppose a QRACM distillation and teleportation process uses $d$ calls to the physical QRACM gate. Then the minimum fidelity of this distillation process is bounded above by $\frac{3}{4} + d\sqrt{\frac{2}{N}}.$
\end{theorem}
That is, if we want a high fidelity QRACM operation to fit into a fault-tolerant computation, we must use resources that are asymptotically similar to those required to do fault-tolerant circuit QRACM.

The idea of the proof is fairly simple: considering the process in \Cref{fig:qracm-distill}, the input states to the QRACM are fixed. This means there is always some address $j$ whose amplitude in the input state is on the order of $O(1/N)$. If we make a new table $T'$ such that $T'[j]\neq T[j]$, then this is barely distinguishable to the state distillation process, since so little of the input to the QRACM will detect this difference. Yet, once we distill this state, it must suffice for all logical inputs. In particular, when our logical input is precisely the modified address $j$, the output between tables $T$ and $T'$ should be orthogonal. This is a contradiction: the distilled states are too close to create this orthogonal output.

A key component of this proof is that the state distillation is the same process for all tables, which is somewhat unrealistic. The specific pre- and post-processing channels could be adapted to the specific table at hand. However, counting arguments similar to \Cref{lem:app:circuit-counts} might still apply; using only $o(N)$ resources in the distillation could mean that a single process handles an exponential number of different tables, and the proof of \Cref{thm:no-qracm-distillation} will still suffice. Unlike \Cref{lem:app:circuit-counts}, measurements are necessarily involved in distillation and teleportation, so formalizing this line of reasoning would need a more sophisticated argument.

\Cref{thm:no-qracm-distillation} is a harsh blow to certain approaches to QRAM, but not all. It states that minimum fidelity is bounded, so applications with ``broad'' inputs to the QRAM, like database HHL and look-ups to classical pre-computation, are not yet forbidden. However, database-QAA suffers greatly. As the QAA hones in on the correct output, its input to the QRACM will get closer to this forbidden logical input. As a concrete example, if $T$ and $T'$ differ in precisely the address that QAA needs to find, then a distilled-and-teleported QRACM will obviously fail.

As a final note, since a QRAQM process implies a QRACM process by simply initializing the quantum memory to a classical input, \Cref{thm:no-qracm-distillation} forbids similar QRAQM distillation and teleportation processes.

\textbf{Heralded QRAM.} 
By ``heralded'' we mean any QRAM technology involving some measurement which indicates whether the QRAM access succeeded or not, typically also implying that, if it succeeded, the fidelity of the resulting access is very high. 

Heralded processes are perfect candidates for gate teleportation because we can try as many times as we want to produce a single good magic state, then teleport it into the code. The impossibility of QRAM state teleportation is a serious blow to heralded QRAM. We cannot produce a resource state independently of the ``live'' data to which we want to apply the QRAM, so we need to instead input the real data into the QRAM, somehow ensuring that when the QRAM fails, this does not decohere the input. 

\textbf{Logical Inputs.} Another potential route to save a process like this is to use the logical input as part of the QRACM access. This is more analogous to a transveral QRACM access, albeit with a very broad definition of ``transversal'' meaning only that we use a physical gate as part of the process to enact a logical gate.

This is a risky process, since the QRACM might noisily decohere the logical input. We want the full logical QRACM gate to have a much higher fidelity than the physical QRACM gate, so there must be little-to-no correlation between the logical state and the input to the QRACM. Intuitively, this suggests an extension of \Cref{thm:no-qracm-distillation} might exist, since if there is no correlation, then a similar table-swapping trick might work. We leave the details of a proof or disproof of this idea to future work.

\section{Classical RAM}\label{sec:classical-ram}
The extensive list of restrictions in \Cref{sec:gate-qram} raises an important question of whether these same issues apply to classical memory. Here we apply the same arguments as \Cref{sec:gate-qram} to classical memory. Either classical memory is fundamentally easier and escapes the restrictions of quantum memory, or the restrictions are already accounted in real-world analysis and use cases of large-scale classical memory.

\subsection{Area-Time Metric}\label{sec:area-time}
As \Cref{sec:summary} discusses, the central thesis of the memory peripheral model, and most critiques of QRAM, is that the qubits in the QRAM (and possibly the memory controllers themselves) could perform other computations more efficient than just the QRAM access. Succinctly, $N$ bits of active QRAM is more like $\Theta(N)$ parallel processors (\cite{TALK:Steiger16} states this explicitly). 

This same argument applies to classical memory, and can mainly be justified by arguing that the cost to construct one bit of memory is proportional to the cost to construct one processor (\cite{WEB:Bernstein2001,SHARCS:Bernstein09} cost algorithms in currency for this reason). While the ratio of costs might be quite high, it is constant relative to the input size of any problem, suggesting that the ratio of memory to processors should stay roughly constant. Some real-world examples of this phenomenon:

\begin{itemize}
\item
Google's ``IMAGEN''~\cite{ARXIV:CWSL+22}, a huge text-to-image diffusion model, used 10 TB of public data plus an internal dataset of a similar size. The training used 512 TPU-v4 chips. Adding up all the independent processing units (i.e., arithmetic logic units (ALUs)), each chip has $2^{21}$ parallel ``processors'', for a total of $2^{30}$ processors~\cite{WEB:Google22}. This means there is only about 18 kB of memory per ALU. 
\item
GPT-3 used 570 GB of training data and used 175 billion parameters~\cite{NEURIPS:BMR+2020}. They use a supercomputer with more than 285,000 CPUs and 10,000 NVIDA V100 GPUs. Each GPU has $2^{12}$ CUDA cores, totalling at least $2^{26}$ parallel processing units, giving 11 kB per processor. 
\item
In the current largest supercomputer in the world, ``Frontier''~\cite{IEEESPEC:Choi22}, each node has 5 TB of memory and 4 AMD Instinct MI250X GPUs. Each GPU has 14,080 stream processors~\cite{WEB:AMD22}. This memory/processor ratio is much higher, but still not huge, at about 89 MB/processor.
\end{itemize}

While the individual processors in a GPU or TPU are not what we usually think of as a ``processor'', they are capable of sequential computation. Machine learning benefits because it is a very parallelizable task, but many problems with algorithms using QRAM also have alternative parallel classical algorithms. 

All of this favours the \emph{area-time} cost model for classical computation, where the cost is the product of the total amount of hardware by the runtime of an algorithm.

A further motivation is parallelism, as in \cite{JPDC:BilPre1995}. Many classical algorithms parallelize perfectly, like the machine learning training above, the linear algebra in \Cref{sec:quantum-linear-algebra}, classical lattice attacks, and unstructured search. This gives little reason to favour a computer with a large amount of memory with few processors, since for only a constant factor increase, many processors can share that memory and achieve the same result faster. As ~\cite{AC:BerLan2013} say about a chip with most of its components idle: ``any such chip is obviously highly suboptimal: for essentially the same investment in chip area one can build a much more active chip that stores the same data and that at the same time performs many other useful computations''. More directly, someone paying for server time literally pays for the area-time product of their computation.

To summarize, the notion that the qubits and controllers in QRAM can be repurposed to more fruitful tasks is an argument that carries over directly to classical computing, and it is actually applied to large, high-performance computing for vast improvements in runtime with little increase in total cost.

\subsection{Active Components}\label{subsection:active_components}
We argued in \Cref{sec:gate-qram} that a QRAM device will probably need interactions with an external environment for the components to function properly. The same is true for classical memory, but here we spell out the biggest difference between quantum and classical memory: classical memory only needs $O(\log N)$ active components to access memory. 

Here we say that a component is ``active'' if it is capable of performing whatever action it needs to properly execute a memory access, and assume that if it is inactive any memory access requiring that component will fail. As an example, in the SWAP-based circuit in \Cref{fig:bucket-brigade-circuit}, each node is active only while the SWAP gates are applied to it. In the bucket-brigade proposal of \cite{PRL:GioLLoMac08} with atoms as routers, the atoms must be stimulated to be active. By our definition of passive vs. active QRAM, most of the components in an active QRAM will require external intervention to be active.

We sketch here how a classical memory could work with only $O(\log n)$ active components: there is a binary tree of routers, all of which are powered off. Each router is capable of receiving a signal which it amplifies and uses to decide to start drawing power. To route a memory request, the router in layer $i$ of the tree, once activated, uses the $i$th address bit to decide whether to forward the address to its left or right child node. The process then repeats. 

In this way, only $O(\log N)$ routers are actually on. It is able to do this because there is no restriction about leaking access patterns to the environment. 

Our concern for quantum memory is that if activating a router requires any interaction with the external environment, then that interaction cannot depend on the quantum state input to the memory, lest the environment become entangled with the state and decohere the input. Thus, for such an architecture, the environment (the control hardware) must decide, \emph{independently of the input state}, which components to activate. 

More specifically, if the input is a uniform superposition of addresses, all the routers in the QRAM must be activated: if one router is not activated, then any address request sent through this router will fail.  Based on the size of the QRAM, there must be $\Omega(N)$ routers, so we need to activate $\Omega(N)$ routers for a successful QRAM access with a uniform superposition of addresses (unless the QRAM is passive). 

Even in algorithms where the input superposition over addresses only has support on a subset of the QRAM memory cells, such as in the final queries in Grover database search~\blfootnote{although technically there would still be non-zero support on most states}, the control hardware must act as though the input could be a uniform superposition. 

For instance, if the control hardware decides not to supply power to routers for certain addresses, any superposition query with support on those addresses will fail. To correctly handle all queries using only some of routers, the control hardware would need to identify which addresses are supported in the coherent quantum state and which are not. However, since the control hardware operates classically, acquiring this information would collapse the quantum state, thereby invalidating the algorithm.
Thus, to reiterate, the control hardware must treat this as though the input \emph{could be} a uniform superposition.

Incorporating algorithm-specific analysis, the control hardware might know that there is a particular subset of $M$ memory locations such that all possible input addresses will be in this subset (this happens in certain quantum lattice sieves~\cite{SAC:Laarhoven2017,AC:ChaLoy2021}). Then we can treat this as a QRAM access to only $M$ bits (with a cost of $\Omega(M)$) rather than $N$, since we could construct a QRAM device that only accesses those bits. However, this relies on knowing exactly \emph{which} locations will have zero support in any input state, and in most algorithms this is not known (and thus this would not be a general QRAM). For example, in a Grover database search, if we know that the desired element is not in certain addresses, we've already solved part of the search problem!

\textbf{Example: Flash memory.} Real-world classical memory devices do not precisely follow this low-power model. Following the description of~\cite{IEEE:CGHLM17}, flash memory is divided into \emph{chips}, which are divided into \emph{dies}, which are divided into \emph{planes}. Within each plane are thousands of \emph{blocks}, each of which contains hundreds of \emph{pages}, each of which contains thousands of bits. The device does not address individual bits, but full pages. Each of these components are integrated circuits, and mostly we do not care about their exact structure or function, since our concern is the fact that each one (chip, die, plane, and block) must have a mechanism to select specific subcomponents based on an address as input. 

For example, near the bottom of this hierarchy, to access a specific block within a plane, two multiplexers (similar in function to a QRAM routing tree) decode the address of the block into a column and a row. A signal is sent along a wire at that column and at that row such that it activates the block at a memory location at the intersection of the selected column and row. From there, another multiplexer selects the appropriate page, and the device acts on that page (to read, write, etc.).

Abstracting this slightly, each block performs the equivalent of an AND gate: it returns the memory address only if there is a signal on the column bit line \emph{and} on the row word line. Thus, as with QRAM, the number of ``gates'' is proportional to the number of pages of memory; however, in contrast to QRAM, only the ``gates'' on the active lines consume energy. In other words, the energy consumption to address one block should be proportional to the size of the length of each column and the length of each row (approximately the square root of the total number of blocks), and experiment confirms this~\cite{THESIS:Mohan10}. 

This experimental result implies also that the analogous process \emph{within} the block, to address the page, consumes energy only in the blocks on the active row and column lines, at most. This is substantially fewer than the total number of blocks, each of which has its own page decoder, meaning that the page decoding process matches the abstract approach of only activating the addressed sub-unit of memory. Here the analogy stops: addressing a single page within a block sends different signals to addressed and non-addressed pages, so the energy per active block is proportional to the size of the block.

Above the block level, to achieve minimum energy consumption we would expect each die to only activate the necessary planes, and each chip to only activate the necessary dies. Instead, the exact opposite happens: flash SSDs will group blocks and pages on different chips and planes into \emph{superblocks} and \emph{superpages}, which are accessed simultaneously. The purpose is so ``the SSD can fully exploit the internal parallelism offered by multiple planes/chips, which in turn maximizes write throughput''~\cite{IEEE:CGHLM17}. Here again we see the ``repurpose memory into processors'' maxim: SSD designers expect that any application needing access to the SSD will have enough processors to reap more benefit from fast parallel access than from a more energy-efficient design.

\textbf{Active memory.} Any volatile classical memory, such as DRAM, is analogous to active hardware-QRAM, in that the total energy cost for $N$ bits is at least proportional to $N$. In fact, it's even worse: volatile memory requires $\Omega(N)$ energy \emph{per clock cycle}.

Despite this, DRAM sees widespread use. For consumer-level computers, RAM is often not the highest energy component, consuming less than any other component in an idle smartphone~\cite{USENIX:CarGer} and similarly for laptops. Though the asymptotic cost might increase with memory size, at practical sizes this is irrelevant. This cost might matter for large-scale applications; however, as in \Cref{sec:area-time}, in practice other components are also scaled up for these applications, so the energy draw of the RAM is not the main cost. Again, the area-time model captures this: total power draw increases with both the amount of volatile memory and the number of processors, so efficient machines will scale both proportionally.

\subsection{Error Propagation and Correction}
In \Cref{sec:gate-qram} we claimed that QRAM technologies that require error rates to scale as $O(1/N)$ per component will probably not be useful. One could make the same argument that if a single gate at one of the row and column lines of DRAM or flash memory spontaneously emits current, then it will create an error in the RAM. Thus, we should expect classical RAM to need component error rates to also scale as $O(1/N)$. Why is this acceptable classically?

The first reason is that extremely high fidelity classical components are fundamentally much easier to produce. Classical components have essentially one dimension of error: bit flips. In contrast, quantum components suffer from analogous bit-flip errors, but also phase-flip errors. Two-dimensonal passively corrected memory, which is provably impossible for quantum memory except in esoteric settings~\cite{RMP:BLPSW2016}, is easy and ubiquitous classically (viz. hard disk drives). 

While this analogy holds for non-volatile memory, volatile classical memory has a further advantage from the continuous energy flow to help correct errors. The technique of a ``cold boot attack'' (such as \cite{ACMCOMM:HSH+2009}) highlights that volatile DRAM is somewhat non-volatile, but it just has an unacceptably high error rate unless power is continuously applied. Recall that we cannot apply power continuously to a passive hardware-QRAM, or else it becomes an active hardware-QRAM. Again, if a QRAM device \emph{does} pay this cost, such as circuit QRAM, it could work perfectly well, but many applications (e.g. database HHL and database QAA) lose their advantage.

Second, classical memory can, and does, perform error correction. Circuit QRAM is allowed to require $O(1/N)$ error because it can be part of an error corrected quantum computer that achieves that level of noise with polylogarithmic overhead. Beyond the consumer level, classical memory \emph{also} uses error correcting codes because the overall failure rate is too high. The error rate for RAM, in failures per Mbit per million hours of operation, is somewhere around 25-75 failures~\cite{ACM:SchPinWeb11} or 0.2 to 0.8~\cite{ISMS:ZDMB+19}, though with high variance not only between devices but also over time for a given device. 

This error rate shows two points: first, it is low enough to be negligible for most consumer applications, and second, it is still high enough to be relevant at large scales. This exactly matches our claims about QRAM. The classical difference is that the achievable error rate for physical components is much, much smaller. 

Finally, errors present less of a problem for classical computing because of the no-cloning theorem. Modern large-scale, high-performance computing often follows the ``checkpoint/restart model''~\cite{BOOK:Gioiosa2017}, where the computer periodically saves the entire state of memory so that it can resume in the case of a fault. A quantum computer cannot do this because of the no-cloning theorem. Instead, the error rate of each operation must be driven so low that the \emph{entire} computation can proceed without fault.

\subsection{Latency}\label{sec:latency}
All memory devices, classical and quantum, must contend with signal latency and connectivity. If connections between qubits are limited to fixed length, then because we live in Euclidean space (a) each qubit can have only a constant number of connections, and (b) sending data between qubits on opposite ends of the memory requires at least $W^{1/d+o(1)}$ connections in $d$-dimensional space, if there are $W$ total qubits. Long-range connections still suffer from signal latency, where -- asymptotically -- the time for signals to propagate across a memory device is still $W^{1/d+o(1)}$. 

We emphasize that even light-speed signals have noticeable latency: for a 3 GHz processor, light travels only 10 cm in each clock cycle. We often don't see this, because (again) huge classical data structures typically co-exist with enormous numbers of parallel processors (as \cite{JPDC:BilPre1995} predict based on this reasoning). Classical memory latency has led to tremendous effort in efficient caching and routing. In contrast, quantum memory access ``takes the mother of all cache misses''~\cite{TWEET:Schanck22}, since any timing information effectively measures which path the access took, and hence which address.

Thus, the case where latency would be an issue -- where a single processor attempts to process an enormous amount of data with completely unpredictable queries -- is extremely rare for classical computers. Yet applications like database HHL typically assume a single quantum processor, and even when assuming parallel access like \cite{ROYSOCA:BBGH+13}, the quantum memory cannot efficiently cache. Overall, quantum memory latency seems more difficult to manage than classical.

\section{Bucket-Brigade QRAM}\label{sec:bucket-brigade}
The bucket-brigade model of QRAM~\cite{PRL:GioLLoMac08} is the most ubiquitous. We describe it here and discuss various proposals to realize the bucket-brigade and a recent result on its error rates.

\subsection{Construction}
The original proposal differed slightly from the circuit version we showed in \Cref{sec:circuit-bb}. It referred to ``atoms'' and ``photons'', and we will continue to use this terminology, though any physical system with similar properties will suffice.

As in the circuit case, we use a binary tree of nodes to route incoming signals to the desired memory cell. However, each node is now an atom, and the lines between nodes represent paths for a photon. The atoms have three states: $\ket{W}_A$ (``wait''), $\ket{L}_A$ (``left''), and $\ket{R}_A$ (``right''), and are initialized in $\ket{W}_A$. The photons have two states, $\ket{L}_P$ and $\ket{R}_P$. 

For QRAM access, we generate a photon from each address qubit (from most significant to least), where $\ket{0}$ generates $\ket{L}_P$ and $\ket{1}$ generates $\ket{R}_P$. When the photon reaches an atom, it interacts depending on the state of the atom:
\begin{itemize}
\item
If the atom is in the wait state $\ket{W}_A$, the photon is absorbed and transforms the atom to $\ket{L}_A$ or $\ket{R}_A$, depending on whether the photon was in $\ket{L}_P$ or $\ket{R}_P$.
\item
If the atom was in the state $\ket{L}_A$, the photon is deflected forwards along the \textbf{left} path, with no change in atomic state.
\item
If the atom was in the state $\ket{R}_A$, the photon is deflected forwards along the \textbf{right} path, with no change in atomic state. 
\end{itemize}
Sequentially sending one photon for each address qubit will create a path in the tree of nodes to the addressed memory location.

The intended feature of this architecture is that \emph{if} we can produce a physical system with these properties, then the entire system needs no external intervention besides inputting photons. We're meant to imagine something like the Galton board or Digi-Comp II~\cite{PATENT:Godfrey68}, wooden devices where balls enter on one side and propagate through with momentum or gravity. These are all \emph{ballistic} computations. 

Recall that if we must intervene at every step, then the system is no longer ballistic and loses many of its advantages. For example, in their original proposal, \cite{PRA:GioLloMac08arch} suggest stimulated Raman emission for the atom-photon interactions; however, \emph{stimulating} that emission with an external field requires work from the control program and/or environment. This implies a cost growing proportional to $N$, which is within a constant factor of circuit QRAM! 

This is not the only way to construct a bucket-brigade QRAM. For example, the atoms could have just the two states $\ket{L}_A$ and $\ket{R}_A$ and be all initialized to one direction, as in~\cite{PRX:HLGJ21}. Their description requires a swap between the incident photon and the node atoms, which requires some active intervention, since (as in the circuit case) we will need to intervene to change the behaviour from swapping into the node atom or propagating through. Hence, two-level photons with two-level atoms is inherently an \emph{active} hardware-QRAM.

These are not exhaustive examples, but give some flavour for how bucket-brigade QRAM could be realized.

\subsection{Reversibility}
A passive bucket-brigade QRAM will need to evolve under a single Hamiltonian, and thus each transition must be reversible. This creates some challenges, since the intended function of a node atom is that once it transitions from the wait state $\ket{W}_A$ to either the left or right states, it should stay in those other states. However, this is not reversible. More plausibly, the left and right states will spontaneously re-emit the photon that excited the wait state. In fact this behaviour appeared from the mathematics of a more complete Hamiltonian analysis in~\cite{THESIS:Cadellans15}. 

This could be a feature, however. We also want to restore the QRAM nodes to their idle state after each access, so if the nodes are carefully engineered so that the final layer's relaxation time is smallest, they will spontaneously re-emit their excitations in reverse order, neatly sending all the control photons back through the routing tree to the address register. Ensuring this happens is an often overlooked problem of passive hardware-QRAM.

More generally, even a simplistic analysis of reversibility creates problems. Considering each node, for each of the three output lines there are 5 states to consider: no photon or a \{left,right\} photon propagating \{up,down\}. For the atom, there are three states (possibly more if we imagine a ``timer'' of how long until an excited atom re-emits its photon). This means each node could be in one of 375 different states, and the QRAM Hamiltonian must act reversibly on them. The necessary functions for routing constrain this only slightly. 

For example, if a $\ket{0}_P$ photon propagates upward from one of the child nodes and reaches an atom in the wait state, what should happen? If it continues propagating upward we lose the favourable scaling proven in~\cite{PRX:HLGJ21}. Instead, if it is reflected backwards, or to the other side, then an erroneous spontaneous emission at some point in the tree will stay constrained to oscillate back and forth among at most three nodes. 

A final note on this is that if we must actively reset the state of all the nodes in a bucket-brigade, this has cost $\Omega(N)$, pushing us back to an active hardware-QRAM. If we do not reset the nodes actively, they must reset themselves. If any errors occur in this process, they will persist to the next QRAM access. 

\subsection{Proposed Technologies}

\paragraph{Trapped Ions}
A particular architecture for bucket-brigade QRAM, proposed in~\cite{PRA:GioLloMac08arch}, is to use trapped ions as the nodes, which would deflect and absorb photons as desired. However, this raises two problems.

The atoms must be carefully positioned to ensure the deflected photons travel between the two atoms. Thus, they must be fixed in place somehow. Earnshaw's theorem forbids a static electric field from trapping a charged particle, so ion traps require continously changing fields, such as a quadropole ion trap. Continuously changing the field requires a constant input of energy.

Second, we need extremely high probabilities of absorption and deflection for the QRAM properly. With, e.g., $2^{40}$ memory elements, we would need error rates of something like $\epsilon/40^2$ for an overall error of $\epsilon$ in the QRAM access, since there are $\binom{40}{2}$ atom-photon interactions for even a single address. This seems implausible for trapped ions. \cite{PRA:GioLloMac08arch} solve this with stimulated emissions, but as mentioned above this is no longer passive.

If this interaction problem is solved, the QRAM routing tree looks like a full-fledged trapped ion quantum computer. Each node has a trapped ion with spectacularly low error rates and a dedicated high-precision laser or microwave emitter. The ions can communicate with each other by the same photons used for routing memory addresses. 

\cite{PRA:HXZJW12} propose microtoroidal resonators instead of trapped ions, but they have the same issue of requiring external photons to drive the interactions at each node. 

\paragraph{Transmons}\label{sec:transfom-bb}
Breaking from literal atoms and photons, \cite{THESIS:Cadellans15} proposed a bucket-brigade where the ``atoms'' consist of 3 superconducting transmons. This provided enough parameters (frequencies, capacitance, etc.) to precisely tune for a high certainty that photons will be absorbed and routed as desired. Their proposed device is truly passive and the desired photon transmission probabilities are over 95\%. 

They mention a critical flaw in their technique which prevents it from scaling. The first photon, which modifies the state of the transmon from ``wait'' to ``left'' or ``right'', must be at a different frequency from the second photon, which will be routed left or right. This means once the second photon is routed, it cannot modify the state of another device. In other words, the routing tree can have at most one layer. 

Solutions might exist. They propose increasing the frequency of the photon, though this requires external energy input to the device and must be done without learning which path the photon took while not scaling energy proportional to the number of transmons, a great challenge. Instead, one might also use slightly different parameters in each layer of the tree, so the frequency which excites nodes in layer $i$ is the same frequency that layer $i-1$ will route. The precise approach in \cite{THESIS:Cadellans15} does not allow this; the range of plausible parameters forbids such layering.

Finally, while the transmission probabilities are high for near-term devices, they are insufficient to scale to a large memory device or a high-fidelity application.

Despite these issues, this is a promising approach, as it explicitly accounts for the fundamental issues in passive hardware-QRAM.

\paragraph{Photonic transistors} 
\cite{NATCOMM:WBL+2022} constructed a photonic transistor with transmons in optical cavities. This could take the place of the node ``atoms'', but the photon is not absorbed by the transmon and instead resonates with it in the cavity. In a bucket-brigade, the photon from one bit of the address would need to resonate with all the nodes in one level of the tree, which a single photon could not do coherently.

\paragraph{Heralded Routers}
\cite{CLEO:CDELE21} propose a ``heralded'' version of bucket-brigade QRAM. The heralded aspect is that they use a photon-atom interaction wherein the photonic state is teleported into the atom, so the device must measure the photon. If the measurement detects no photon, then the photon was lost due to, e.g., propogation error, and the QRAM access can be restarted. Crucially, this avoids the issue of heralding described in \Cref{sec:gate-teleportation}, since measuring photon loss reveals only that a photon was lost, not \emph{which} photon was lost. 

A crucial design choice is that all routers in a single layer of the routing scheme must share their measurement of the photon. Measuring at each detector would immediately reveal the routing path and decohere the address state, so instead we must replace the detectors each each router with waveguides that all converge at a single detector. Similarly, they require a method to simultaneously measure the spins of $\Omega(N)$ atoms; they propose an ancillary photon for this. Altogether this will lead to one measurement per layer of the routing tree, so the controller only needs to perform $O(\log^2 N)$ operations.

As written, their design requires a Hadamard gate on the atom in each router to perform the oblivious teleportation described above, which means $\Omega(N)$ active operations to perform a memory access. However, they also need the atoms initialized in the $\ket{+}$ state to obliviously teleport the photonic input states onto the routers' atoms. Thus, they could skip both operations by somehow simultaneously measuring the spins of all the atoms in the $\{\ket{+},\ket{-}\}$ basis. Since only one atom will not be in the $\ket{+}$ state, measuring the parity of all the spins will perform the necessary teleportation if no errors have occurred. 

If this measurement detects an even parity, then no correction must be done and all atoms are initialized for a subsequent read; however, odd parity means one atom, at an unknown position which is correlated with the address register, is in the $\ket{-}$ state and must be corrected. This is problematic: we now need a second process or device to correctly route a signal to correct that atom, without ever classically learning where the signal is going. Perhaps this could be accomplished with another photon sent through the system (which is already set to route photons to the atom in the $\ket{-}$ state).

Thus, there are some missing components for this to fully realize a passive hardware-QRACM. Moreoever, it suffers from the same ion trap QRACM issues of needing high-probability atom-photon interactions and passively trapped ions.

\cite{CLEO:CDELE21} also propose teleporting the address into the nodes of the tree before the routing process, which requires $\Theta(N)$ gates applied to the device and thus fails to realize passive hardware-QRAM.

\paragraph{Photonic CNOT.}
Another heralded QRAM is \cite{NSR:YMHC+15}, who use a post-selected photonic CNOT for their quantum router such that each routing operation has, at best, a $0.0727$ chance of failure. This means we can apply at most 9 routers sequentially before we have only a 50\% chance of success. Bucket-brigade QRAM with $n$ layers requires $\binom{n}{2}$ routing operations, meaning 9 routings is only about 5 layers (i.e., a 32 bit memory). However, because their scheme requires the control and signal photons to enter the router simultaneously, we may need several routers in each layer to simultaneously route more address qubits. Ultimately, this approach doesn't scale. 

\subsection{Error Rates}\label{sec:bb-errors}
An advantage of the bucket-brigade approach -- even as a circuit-based QRAM -- is favourable error scaling, shown by \cite{PRX:HLGJ21}. To give an intuition of their argument, in any circuit-based QRAM, there are $\Omega(N)$ gates and so we might expect each gate to need error rates less than $O(1/N)$. This is true for unary QRAM, which one can see with a simple analysis. However, with the bucket brigade, errors can get stuck. That is, if an error occurs on some node deep in the tree, this does not necessarily affect any other nodes.

The only way for errors to affect the final result is if they propagate to the root node somehow. At each node, the routing circuit only swaps from one of its two child nodes, so the error will only propagate upward if it is on the correct side. As there is only one path which will get routed all the way to the top of the tree -- the path representing the actual address in the query -- an error propagates to the root only if it occurs on this path. But only $O(\log N)$ qubits are on this path. 

Errors could propagate in other ways; for example, if the control qubit has an error \emph{and} the routing qubits have errors, then this will erroneously create a path that propagates further up. But for the error to keep propagating upward, a very precise sequence of errors must occur to create a path for it, and this is unlikely for even modestly low probability errors. Overall, if qubit undergoes an error channel with probability $p$, the error of the final readout is only $O(p\cdot \log^2(N))$~\cite{PRX:HLGJ21}.

The bucket-brigade requires all nodes to be in the ``wait'' state at the start of an access. However, errors in a previous access that change this state could persist to later accesses, depending on the device's Hamiltonian. If the probability of a persistent error is $p$ for each node in each access, then after $Q$ accesses the proportion of error-ridden nodes grows proportional to $Q^2p$. 

Generally, the per-component error rate $p$ must still be exponentially small for applications like database QAA, but for other applications, like quantum linear algebra (see \cref{subsec:linear_algebra_with_noisy_qram}), higher error rates could be manageable. 

Moreover, we propose that a bucket-brigade QRAM system could be \textit{partially passive}. Following the observation that errors in the root nodes of are more destructive, we can run full error-correction on the first $k$ layers of the bucket-brigade tree. If $k$ is selected to be a constant, or at least $o(\log N)$, then this circumvents the issues with active QRAM that we have extensively highlighted. 
The remaining $\log(N) - k$ layers would be a collection of smaller, passive QRAMs joined to the error-corrected root nodes. This encounters the same problem of interfacing noisy components with quantum data encoded into a logical state that we discuss in \Cref{sec:errors}.
This proposal requires $2^k$ error-corrected nodes, and we expect that the overall probability of error would scale something like $O(p(\log(N)-k)^2)$. This wouldn't change the asymptotic error rate of bucket-brigade, but with the same physical error rate, this method could allow larger memories to become practical than a fully passive bucket-brigade architecture would allow.

\textbf{Example.}
To show why the small bucket-brigade error rates are somewhat unique, we describe here an alternative readout circuit that \emph{does} require $O(1/N)$ gate error, as shown in \Cref{fig:bad-bb-readout}. To explain this circuit, we add an extra layer $L$ of qubits at the leaves of the routing tree (shown as boxes in \Cref{fig:bad-bb-readout}). We then add a single extra qubit prepared in the $\ket{1}$ state, and propagate it through the tree by applying the routing circuits at each node. This routes the $\ket{1}$ to the location indexed by the address regiter. From there, for every $i$ such that $T[i]=1$, we apply a CNOT from the leaf at the $i$th address to a qubit in $L$. This means this final layer of qubits is all zeros if $T[i]=0$, and has precisely one qubit in the $\ket{1}$ state otherwise. Thus, we can read out the result by computing the parity of this final layer. 

Parity is an appealing operation in the surface code, since it can be done by conjugating a multi-target CNOT (itself depth-1 in a surface code) by Hadamard gates. 

The drawback of this method is that if an error occurs anywhere on this final layer of ancilla qubits, it will propagate immediately to the output. Thus, the qubits and gates in $L$ would need to have error rates of $O(1/N)$.

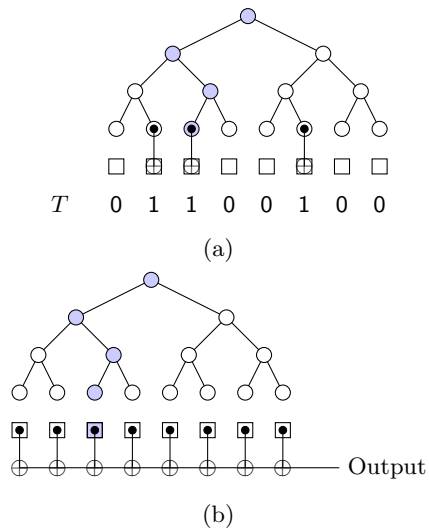
\begin{figure}
\begin{subfigure}{0.45\textwidth}
\begin{tikzpicture}

\draw (5,5) -- (4,4.5);
\draw (5,5) -- (6,4.5);
\draw (4,4.5) -- (3.5,4);
\draw (4,4.5) -- (4.5,4);
\draw (6,4.5) -- (5.5,4);
\draw (6,4.5) -- (6.5,4);
\draw (3.5,4) -- (3.25,3.5);
\draw (3.5,4) -- (3.75,3.5);
\draw (4.5,4) -- (4.25,3.5);
\draw (4.5,4) -- (4.75,3.5);
\draw (5.5,4) -- (5.25,3.5);
\draw (5.5,4) -- (5.75,3.5);
\draw (6.5,4) -- (6.25,3.5);
\draw (6.5,4) -- (6.75,3.5);

\draw[fill=blue!20!white] (5,5) circle (0.1);
\draw[fill=blue!20!white] (4,4.5) circle (0.1);
\draw[fill=white] (6,4.5) circle (0.1);
\draw[fill=white] (3.5,4) circle (0.1);
\draw[fill=blue!20!white] (4.5,4) circle (0.1);
\draw[fill=white] (5.5,4) circle (0.1);
\draw[fill=white] (6.5,4) circle (0.1);
\draw[fill=white] (3.25,3.5) circle (0.1);
\draw[fill=white] (3.75,3.5) circle (0.1);
\draw[fill=blue!20!white] (4.25,3.5) circle (0.1);
\draw[fill=white] (4.75,3.5) circle (0.1);
\draw[fill=white] (5.25,3.5) circle (0.1);
\draw[fill=white] (5.75,3.5) circle (0.1);
\draw[fill=white] (6.25,3.5) circle (0.1);
\draw[fill=white] (6.75,3.5) circle (0.1);

\node at (2.5,2.5) {$T$};
\node at (3.25,2.5) {$\mathsf{0}$};
\node at (3.75,2.5) {$\mathsf{1}$};
\node at (4.25,2.5) {$\mathsf{1}$};
\node at (4.75,2.5) {$\mathsf{0}$};
\node at (5.25,2.5) {$\mathsf{0}$};
\node at (5.75,2.5) {$\mathsf{1}$};
\node at (6.25,2.5) {$\mathsf{0}$};
\node at (6.75,2.5) {$\mathsf{0}$};

\foreach \i in {0,...,7}{
	\draw (3.15+0.5*\i,2.9) rectangle (3.35+0.5*\i,3.1);
}

\draw[fill=black] (3.75,3.5) circle (0.05);
\node at (3.75, 3) {$\oplus$};
\draw (3.75,3) -- (3.75,3.5);

\draw[fill=black] (4.25,3.5) circle (0.05);
\node at (4.25, 3) {$\oplus$};
\draw (4.25,3) -- (4.25,3.5);

\draw[fill=black] (5.75,3.5) circle (0.05);
\node at (5.75, 3) {$\oplus$};
\draw (5.75,3) -- (5.75,3.5);
\end{tikzpicture}\caption{}
\end{subfigure}
\begin{subfigure}{0.45\textwidth}
\begin{tikzpicture}

\draw (5,5) -- (4,4.5);
\draw (5,5) -- (6,4.5);
\draw (4,4.5) -- (3.5,4);
\draw (4,4.5) -- (4.5,4);
\draw (6,4.5) -- (5.5,4);
\draw (6,4.5) -- (6.5,4);
\draw (3.5,4) -- (3.25,3.5);
\draw (3.5,4) -- (3.75,3.5);
\draw (4.5,4) -- (4.25,3.5);
\draw (4.5,4) -- (4.75,3.5);
\draw (5.5,4) -- (5.25,3.5);
\draw (5.5,4) -- (5.75,3.5);
\draw (6.5,4) -- (6.25,3.5);
\draw (6.5,4) -- (6.75,3.5);

\draw[fill=blue!20!white] (5,5) circle (0.1);
\draw[fill=blue!20!white] (4,4.5) circle (0.1);
\draw[fill=white] (6,4.5) circle (0.1);
\draw[fill=white] (3.5,4) circle (0.1);
\draw[fill=blue!20!white] (4.5,4) circle (0.1);
\draw[fill=white] (5.5,4) circle (0.1);
\draw[fill=white] (6.5,4) circle (0.1);
\draw[fill=white] (3.25,3.5) circle (0.1);
\draw[fill=white] (3.75,3.5) circle (0.1);
\draw[fill=blue!20!white] (4.25,3.5) circle (0.1);
\draw[fill=white] (4.75,3.5) circle (0.1);
\draw[fill=white] (5.25,3.5) circle (0.1);
\draw[fill=white] (5.75,3.5) circle (0.1);
\draw[fill=white] (6.25,3.5) circle (0.1);
\draw[fill=white] (6.75,3.5) circle (0.1);

\draw[fill=blue!20!white] (4.15,2.9) rectangle (4.35,3.1);
\foreach \i in {0,...,7}{
	\draw (3.15+0.5*\i,2.9) rectangle (3.35+0.5*\i,3.1);
	\draw[fill=black] (3.25+0.5*\i,3) circle (0.05);
	\draw (3.25+0.5*\i,3) -- (3.25+0.5*\i,2.5);
	\node at (3.25+0.5*\i,2.5) {$\oplus$};
}

\draw (3.25,2.5) -- (7.5,2.5);
\node[right] at (7.5,2.5) {Output};

\end{tikzpicture}\caption{}
\end{subfigure}
\caption{Alternative, suboptimal bucket-brigade readout forcing the readout gates to have error rates of $O(1/N)$. The blue nodes are those with $\ket{1}$ during the memory access for a specific address (in this case, $i=2$). Notice that the parity in the second step only needs to use leaves for $i$ with $T[i]=1$, but we include all of them for clarity.}\label{fig:bad-bb-readout}
\end{figure}
\section{Other Proposals}\label{sec:other-qram}
Here we summarize some other proposals which do not follow the bucket-brigade approach, and detail the ways in which they fail to provide passive hardware-QRAM, as they are currently described.

\begin{figure*}
\centering
\includegraphics{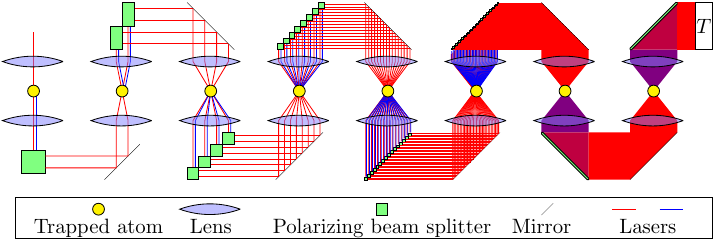}
\caption{Schematic of quantum optical fanout QRAM, almost exactly as shown in \cite{PRA:GioLloMac08arch}.}\label{fig:qopt-fanout}
\end{figure*}

\subsection{Time-bins}
Many applications and methods to store quantum states with photon echoes exist (e.g.~\cite{LPR:TACC+10}), but \cite{JMO:MoiMoi16} propose an additional addressing component to turn this storage into QRAM. To summarize their method, there is a memory cavity with $M\gg N$ atoms and a control cavity. To store data, the device would send quantum states encoded as photons into the memory cavity, where they would excite the memory atoms. The memory atoms can replay this input via a two-pulse echo technique; sending a second pulse will prompt the atoms to re-emit the photons that were sent in. Importantly, if photons are sent in one-at-a-time, they will be re-emitted in a time-reversed order. 

Thus, to read the memory, the device sends a pulse which prompts the memory cavity to re-emit its stored state. With just the memory cavity, this would re-emit all input states at classically determined times, giving us CRAQM but not QRAQM. For quantum addressing they use a second \emph{control} cavity. The control cavity contains an atom such that, when this atom is in its ground state, the memory cavity cannot emit any photons. However, when the control cavity's atom is excited, it allows the transfer of photons out of the memory cavity.

An address state $\sum_i \alpha_i \ket{i}$ must then be translated so that each address $\ket{i}$ is entangled with a photonic state $\ket{\psi_i(t)}$ which reaches the control cavity at precisely the right time to allow the memory cavity's photons to pass through. In this way, other addresses, which will not be entangled with a photon pulse at that time, will not allow any emission from the memory cavity.

The fundamental nature of the time-bin addressing means that this requires $\Omega(N)$ time for each memory access. The purported feature of this method is the minimal hardware requirements (though it still requires $\Omega(N)$ atoms in the memory cavity), but unary circuit QRAM already provides a near-optimal circuit on the $\Omega(N)$ side of the time-hardware tradeoff of QRAM circuits. 

In \cite{PRX:OKD+2022}, they use a second type of pulse to decohere a memory atom so that it will only re-emit once an identical pulse is sent to it. Using different frequencies of pulses to store different memory bits, this could reduce the total time requirement, at the expense of needing higher precision in pulse frequencies.

There is still an issue of how the time-bin photonic address state is created. As an example of how this could be created, consider the unary circuit QRAM. It classically iterates through addresses $j$, and for each address $j$, flips the state of a control qubit if the address in superposition equals $j$. Suppose that flipping the control qubit puts it into an excited state which can be stimulated to emit a photon. If we do this sequentially, then at time $t_j$, only the address state $\ket{j}$ will cause a photon to be emitted, thus the presence of photons at that time is in superposition and entangled with the state of the address. In this way, applying this stimulation instead of a CNOT turns the unary circuit QRAM into a circuit that translates an address superposition into a superposition of time-bin photons. 

Unfortunately, this has solved the addressing problem by using another QRAM circuit to create the address! While superpositions of photons at different points in time \emph{may} be a native quantum state in some architecture, for other architectures, creating such a state is essentially a quantum addressing problem, which is the difficult problem of QRAM anyway. The decohering pulses of \cite{PRX:OKD+2022} face a similar issue, where they must be addressed and sent in a coherent superposition to create a QRAM device.

As a final note, many aspects of the scheme in \cite{JMO:MoiMoi16} require $\Omega(N)$ active interventions. For example, with the photon echo technique, stimualating the memory cavity to emit its states requires a phase-flipping pulse on all $M$ atoms for each time bin, meaning $\Omega(N)$ pulses that each require $\Omega(N)$ energy. As in other proposals with photon-atom interactions, the control cavity requires a stimulating pulse for each interaction (thus, another $\Omega(N)$ stimulating pulses). 

\begin{figure}
\resizebox{\columnwidth}{!}{
\includegraphics[scale=1.0]{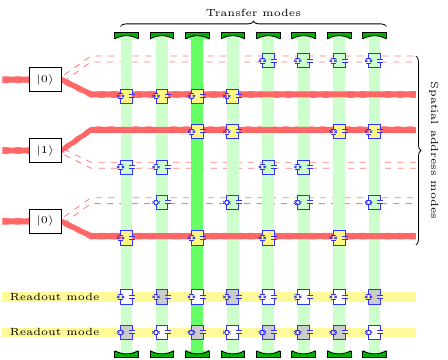}
}
\caption{Phase gate fanout QRACM, from \cite{PRA:GioLloMac08arch}. The address qubits route the address modes (red), which excites the phase-control qubits (blue transmons). The excited phase-control qubits are highlighted yellow. This will cause exactly one cavity's transfer mode to resonate (dark green), exciting the 2 memory qubits (bottom, blue) for that address. The memory qubits change the phase of the readout modes (yellow) which can be detected. The arrangement of the two types of memory qubits (shown as white and grey) represents the data: in this case, the table is $\mathsf{01010001}$. }\label{fig:phase-gate}
\end{figure}

\subsection{Quantum Optical Fanout}
\cite{PRA:GioLloMac08arch} describe an architecture using only $\log N$ trapped atoms, and instead encode the QRAM access in $N$ spatial modes of light. The basic principle is this: each bit of the address becomes one trapped atom, and then each beam interacts with the atom, and if the atom is in $\ket{1}$, the photon is polarized differently than if the atom is $\ket{0}$. Then the photon beam is sent through a half-wave plate which spatially separates different polarizations, thus doubling the number of photonic modes and ensuring that an input beam is directed according to the state of the atom.

Figure~\ref{fig:qopt-fanout} duplicates Figure 4 from~\cite{PRA:GioLloMac08arch}, with a crucial difference: we show the schematic for an 8-bit address, rather than only 3. This immediately highlights the problem: the number of spatial modes grows exponentially with the number of address bits, but they must be focused back onto a single point.

If the photons are not aimed precisely, then after interacting with the atom they will drift into another spatial mode. This causes an error. The precision required for this aiming is $O(1/N)$ for $N$ bits of memory. This seems to be the only physical mechanism for errors to propagate from one sequence of spatial modes to another, suggesting that the overall error scaling has similar robustness to bucket-brigade QRAM, but the aiming argument above means that each component requires a precision scaling as $O(1/N)$ or it will have an error. Intuitively, it doesn't matter whether errors propagate across ``paths'' of access, if all paths have an error.

Trying to guide the photons with a waveguide does not fix the problem, regardless of the miniaturization we achieve. Taking the perspective of the atom, the sources and receptors of the photons and must occupy non-overlapping sections of all radial paths out of the atom. If either source or receptor has a smaller cross-sectional area, then we can move it closer, but it occupies the same ``arc-area'' from the atom's perspective. This means the precision of the beam must be the same -- $O(1/N)$ -- no matter the absolute area of the photonic components.

Thus, in the exact formulation given, errors in this QRAM must scale infeasibly. We can wonder whether a different technology, inspired by this approach, could succeed. A core function of this approach is that the $i$th address bit interacts with $2^i$ spatial modes in one ``operation''. This is the crux: finding a mechanism to enact such a large number of interactions passively, with error scaling much less than $2^i$, seems difficult if not impossible.

\subsection{Phase Gate Fanout}\label{sec:phase-gate}
\cite{PRA:GioLloMac08arch} describe another quantum fanout system, shown in \Cref{fig:phase-gate}. This system translates each address qubit into a photon in one of two spatial modes (based on $\ket{0}$ or $\ket{1}$ in the address qubit, shown as red horizontal beams in \Cref{fig:phase-gate}). Each of these $2\lg N$ spatial address modes couples to $N$ phase-control qubits (the upper blue transmons in \Cref{fig:phase-gate}), so that the phase-control qubits change state when there is a photon in this spatial address mode.

Each phase-control qubit also couples to a another optical mode, which we'll call a transfer mode (the vertical green beams in \Cref{fig:phase-gate}), such that there are $N$ transfer modes, each coupled to $\lg N$ phase-control qubits. The coupling is such that if the phase-control qubit is in its $\ket{1}$ state, it adds a phase to the transfer mode. This two-mode coupling is similar to the recently experimentally-realized photon transistor in~\cite{NATCOMM:WBL+2022}.

The net effect of all this is that if the address is $\ket{j}$, then the $i$th transfer mode gains a phase of $e^{ij\varphi_0}$ for some constant $\varphi_0$. 

For readout, we have another collection of $N$ memory qubits (the bottom two rows of transmons in \Cref{fig:phase-gate}), each coupled to a different transfer mode, so that the $j$th memory qubit interacts with the transfer mode only when the transfer mode has a phase of $e^{ij\varphi_0}$. Thus, the memory qubit changes state depending on whether the transfer mode is in resonance. For the final readout, we use another optical readout mode (the yellow beams in \Cref{fig:phase-gate}), which interacts with all memory qubits, and picks up a phase for each one that is active. This means we will actually need $2N$ memory qubits and 2 readout modes, so that \emph{which} readout mode gains a phase will encode the binary result of the memory lookup.

There are a number of practical issues with this scheme.

The foremost issue is the narrowness of the resonance of the memory qubits. \cite{PRA:GioLloMac08arch} note that there will need to be $N$ different resonant phases, and analyze the width of a Fabry-Perot cavity's resonance and conclude that if the transmissitivity of the beamsplitters in the cavity is $O(1/N)$, the width of resonant phases is also $O(1/N)$. However, this does not address the need to fabricate each cavity, where the resonant phase must be constructed with $O(1/N)$ precision as well. For a Fabry-Perot cavity, this would mean $O(1/N)$ precision in the physical distance between the mirrors.

Second, each address mode must interact with $N$ qubits. Modelled with a Jaynes-Cummings Hamiltonian, every excited qubit means a reduced photon number in the address mode. If not all qubits are excited, we will end up with some superposition where only a small portion of the total amplitude corresponds to states where the \emph{right} qubits are excited to properly perform the memory access. Thus, we need all qubits to be excited, meaning the photon number of each address mode must be $\Omega(N)$. 

In addition to forcing us to use $\Omega(N\log N)$ energy per QRAM access, we must also provide this energy to the address modes in such a way that, for each address qubit, we do not learn \emph{which} of the two address modes took this energy. If we learned this, it effectively measures the address qubit. The quantum router of \cite{THESIS:Cadellans15} might suffice.

This QRAM may not even be passive, in fact. To excite the memory qubits, we do not want to activate a new pulse -- which requires $\Omega(N)$ interventions and $\Omega(N)$ energy -- but rather we want the memory qubits to reside in an optical cavity that maintains some electric field. Consider a back-of-the-envelope analysis of the losses of the optical cavities. If each cavity  has quality factor $Q$ and resonant frequency $f$, the fraction of energy lost per unit time is proportional to $Nf/Q$, since there are $N$ cavities. Modern ``ultra-high'' quality factors are on the order of $10^9$~\cite{NATCOMM:LJGK16}, about the same scale as the frequency (in hertz) of microwaves. Thus, to maintain the QRAM for more than a few seconds requires either enormous optical cavities (to increase the resonant frequency) or spectacular improvements in optical cavity quality factors. In either case, if these improvements do not keep pace with the increase in the number of memory bits, the \emph{total} energy loss quickly becomes enormous. 

Finally, each readout mode interacts with $N$ qubits. If our error model permits some probability $\epsilon$ of adding the wrong phase from each memory qubit (e.g., the probability from measuring after interacting with that memory qubit), then we need $\epsilon$ to be at most $O(1/N)$, or the memory readout will fail. 

Unlike \cite{PRA:GioLloMac08arch}, the arrangment of phase-control qubits in \Cref{fig:phase-gate} means that for each address, only one cavity has all $n$ phase-control qubits excited. Potentially one could engineer each phase-control qubit to contribute $1/n$ of the necessary phase to bring the transfer mode into resonance. If so, the precision of the resonator might only need to be $O(1/\log N)$. However, given the readout method, any leakage from the resonators will contribute directly to readout error (if leakage excites the incorrect memory qubit, the readout detects this). Hence, the transfer modes may still need $O(1/N)$ precision.

\subsection{Derangement Codes}
\cite{THESIS:Hann21} proposes error correction based on code switching, which we summarize and slightly tweak here. 

We start by imagining a no-op state $\ket{\text{nop}}$, such that not only is $\ket{\text{nop}}$ invariant under $U_{QRAM}$, the internal quantum state of the QRAM nodes are also invariant when we input $\ket{\text{nop}}$. As an example for bucket-brigade QRAM, we could use an ancilla flag qubit which controls the production of the input photons from the address qubits. If this flag is set to $0$, no input photons are generated and nothing enters the QRAM device. 

Then we consider the circuit shown in Figure~\ref{fig:qracm-swap-error-mit}. This has $m$ devices to access the same table $T$. It generates a superposition of $m$ different states, indexing which QRAM device to shuffle our real query into (the query is $\ket{\phi}=\sum_i \alpha_i\ket{i}\ket{0}$), while sending $\ket{\text{nop}}$ states to all the other devices. We then apply all the QRACM devices, un-shuffle the state, and measure the shuffle controls in the Hadamard basis.

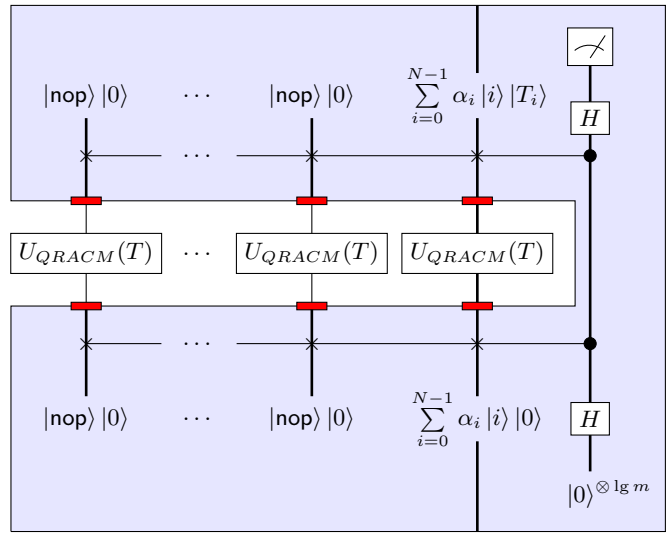
\begin{figure*}
\centering
\begin{tikzpicture}[scale=1.2]
\draw[fill = white!90!blue] (-1,-0.5) -- (-1,2.5) -- (6.5,2.5) -- (6.5,3.9) -- (-1,3.9)
	-- (-1,6.5) -- (7.7,6.5) -- (7.7,-0.5) -- (-1,-0.5);

\draw[line width=1] (6.7,0.3) -- (6.7,6);

\draw[line width=1] (5.2,-0.5) -- (5.2,0.7);
\draw[line width=1] (5.2,1.3) -- (5.2,5.0);
\draw[line width=1] (5.2,5.6) -- (5.2,6.5);

\node at (0,1) {$\ket{\mathsf{nop}}\ket{0}$};
\node at (0,2) {$\times$};
\node at (1.5,1) {$\dots$};
\node at (1.5,2) {$\dots$};
\node at (3,1) {$\ket{\mathsf{nop}}\ket{0}$};
\node at (3,2) {$\times$};
\node at (5.2,1) {$\sum\limits_{i=0}^{N-1}\alpha_i\ket{i}\ket{0}$};
\node at (5.2,2) {$\times$};
\node at (6.7,0) {$\phantom{^{\otimes m}}\ket{0}^{\otimes \lg m}$};
\draw[fill=black] (6.7,2) circle (0.08);
\draw (6.7,2) -- (2,2);
\draw (1,2) -- (0,2);
\draw[line width = 1] (0,1.3) -- (0,2.5);
\draw[line width = 1] (3,1.3) -- (3,2.5);
\draw (3,2.5) -- (3,3.9);
\draw (0,2.5) -- (0,3.9);
\draw[line width = 1] (0,3.9) -- (0,5);
\draw[line width = 1] (3,3.9) -- (3,5);

\node[draw=black,rectangle,fill=white] at (6.7,1) {$H$};
\node[draw=black,rectangle,fill=white] at (0,3.2) {$U_{QRACM}(T)$};
\node[draw=black,rectangle,fill=white] at (3,3.2) {$U_{QRACM}(T)$};
\node at (1.5,3.2) {$\dots$};
\node[draw=black,rectangle,fill=white] at (5.2,3.2) {$U_{QRACM}(T)$};
\draw[fill=red,draw=black] (-0.2,2.45) rectangle (0.2,2.55);
\draw[fill=red,draw=black] (2.8,2.45) rectangle (3.2,2.55);
\draw[fill=red,draw=black] (-0.2,3.85) rectangle (0.2,3.95);
\draw[fill=red,draw=black] (2.8,3.85) rectangle (3.2,3.95);
\draw[fill=red,draw=black] (5.0,3.85) rectangle (5.4,3.95);
\draw[fill=red,draw=black] (5.0,2.45) rectangle (5.4,2.55);

\node at (0,4.5) {$\times$};
\node at (3,4.5) {$\times$};
\node at (5.2,4.5) {$\times$};
\draw[fill=black] (6.7,4.5) circle (0.08);
\draw (0,4.5) -- (1,4.5);
\node at (1.5,4.5) {$\dots$};
\draw (2,4.5) -- (6.7,4.5);

\node[draw=black,rectangle,fill=white] at (6.7,5) {$H$};
\draw[draw=black,fill=white] (6.4,5.7) rectangle (7.0,6.2);
\draw (6.7,5.8) -- (6.9,6.1);
\draw  (6.88,5.9) arc (50:130:0.3);
\node at (0,5.3) {$\ket{\mathsf{nop}}\ket{0}$};
\node at (3,5.3) {$\ket{\mathsf{nop}}\ket{0}$};
\node at (1.5,5.3) {$\dots$};
\node at (5.2,5.3) {$\sum\limits_{i=0}^{N-1}\alpha_i\ket{i}\ket{T_i}$};

\end{tikzpicture}
\caption{QRACM with heralded error mitigation, adapted from \cite{THESIS:Hann21}. Shaded regions and thicker wires indicate error corrected qubits and operations, and red rectangles on the border indicate decoding and re-encoding. The controlled swaps indicate that, based on the $\lg m$ control qubits, the last register is swapped to one of the intermediate registers. The states shown at the top are the states if no errors occur, in which case the measurement result is all $0$.}\label{fig:qracm-swap-error-mit}
\end{figure*}

Measuring all zeros in the shuffle controls means the states that went out came back to the registers they started in, enforcing a ``sameness'' among the queries between different QRAM devices. If the errors are independent, then this suppresses errors, since it's unlikely for the same error to occur in all QRAM devices at the same time. 

\cite{THESIS:Hann21} proposes compressing this to use only one QRAM device by sequentially controlling swaps, so that only one of $m$ applications of the QRAM applies to the real query state. Since this error correction method only works for uncorrelated errors, it fails for bucket-brigade QRAM. Specifically, if an error occurs at some point in the routing nodes, that error will persist between queries (unless the routing nodes are reset, which creates an $\Omega(N)$ cost as discussed previously). Thus, on average half of all errors will be identical between queries, rendering the error suppression ineffective. 

Expanding their idea to a wide, parallel repetition increases hardware overhead by a factor of $m$, but we imagine $m\ll N$, otherwise the error suppression scheme -- which requires $\Omega(m\log N)$ gates for the shuffles alone -- will start to approach the cost of a circuit QRAM access.

The overall effect of this is a ``heralded'' QRAM access, where if we measure all zeros in the control qubits, we know that the infidelity has been suppressed to $1/m$ of the physical error rate (see~\cite[Chapter 4]{THESIS:Hann21} for proof), \emph{but} the probability $p$ of measuring all-zero is approximately equal to the original physical infidelity. This means that for an algorithm to have a reasonable probability of success, it must be limited to $O(\frac{1}{1-p})$ QRAM queries. Recalling the applications from \Cref{sec:applications}, this forbids database QAA.

Finally, the code switching process raises some issues. Imagine \Cref{fig:qracm-swap-error-mit} in the surface code. We can attempt to limit the time when a qubit is outside of the surface code to the smallest possible, but even a brief period without error correction could spell disaster for a resource-intensive algorithm. 

\subsection{Two-Level Quantum Walkers}
A recent proposal for ballistic quantum computation imagines ``walkers'' (e.g., fermions) propagating along a waveguide \cite{asaka2021quantum,asaka2022two}. Each walker is a two-level quantum system, but they also use a dual-rail encoding. A special ``roundabout gate'' can switch between these encodings by rotating $\ket{0}$ states of the walker clockwise, and $\ket{1}$ states counterclockwise. Combined with single qubit rotations on each walker (such as magnetic fields if the walkers are fermions), they construct a universal, ballistic gate set. In \cite{asaka2022two} they use these gates to construct QRAM. They route each address walker through a binary tree, with CNOT gates between each edge in the tree to perform the routing. 

If the underlying ballistic computation were built, then this would work as QRAM; however, such a technology faces all the issues we discuss in \Cref{sec:gate-qram}. In fact, their proposal can construct a ballistic version of \emph{any} circuit, which would be a powerful and revolutionary quantum technology. 

To function properly, the walkers must transition correctly at the roundabout gates, and \cite{asaka2021quantum} computes the probability of this transition in terms of the momentum of the wave packets. If the momentum differs from the target momentum by $\epsilon$, the transition probability scales as $1-\Omega(\epsilon^2)$. Keeping the momentum of a wavepacket constrained requires a broader wavepacket, and an exponentially precise momentum in their system requires an exponentially wide wavepacket. Thus, their proposed system needs improvement to accommodate exponential-sized circuits. 

However, it may be that the error rates of their QRAM access scale like the bucket-brigade, and the momentum precision only needs to be logarithmic in the memory size. In this case the quantum walkers could act as a specialized QRAM module, but more work must be done to clarify the error rates and to avoid the issues we have discussed in the interfaces between the quantum walkers and the rest of a fault-tolerant computation.

\section{Discussion}
Large scale, high quality QRAM which is cheap to use is unlikely. It would require passively-corrected quantum memory and highly complex, high-fidelity ballistic computation. We argued in \Cref{sec:gate-qram} and \Cref{sec:errors} that these are much stronger assumptions about quantum computing hardware than general large-scale, fault-tolerant quantum computing. Even if those assumptions come true, QRAM may still be difficult: building a ballistic computation is complicated (e.g. \Cref{lem:hamiltonian-cost}) and it may not interact well with the rest of a fault-tolerant quantum computation (\Cref{thm:no-qracm-distillation-informal}). Indeed, we could find no proposed hardware-QRAM technology that was not implicitly active or required unreasonably low error rates (\Cref{sec:bucket-brigade},\Cref{sec:other-qram}). 

All of that said, we cannot \emph{prove} that cheap QRAM will not be possible. However, many powerful computing devices (say, a polynomial time SAT-solver) are not provably impossible, but we do not assume they will be built. Our claim is that cheap QRAM should be considered infeasible until shown otherwise, instead of the reverse. While classical RAM appears to set a precedent for such devices, we argued in \Cref{sec:classical-ram} that this is an inappropriate comparison: QRAM is more difficult relative to other quantum computations, than classical RAM is difficult relative to classical computation.

Despite these critiques, QRAM will likely be a crucial tool for quantum algorithms. In a circuit model, QRAM has a linear gate cost (regardless of depth), but optimized circuits allow it to help with quantum chemistry and some cryptographic attacks (\Cref{sec:circuit-qram}). We encourage readers to consider circuit QRAM for their algorithms, and use one of the optimized circuits we referenced.

We leave here a summary of some open questions we raised in this work:
\begin{enumerate}
    \item 
    With the favourable error scaling of bucket-brigade QRACM, can it be more effective than other methods in a surface code, especially for cryptographic attacks?
    \item If we correct only some of the nodes in a bucket-brigade QRAM, can this suppress enough noise for practical purposes?
    \item
    Can \Cref{thm:circuit-qram-lower-bound} and \Cref{thm:no-qracm-distillation-informal} extend to account for measurement feedback, and can the latter extend to data-specific distillation processes?
    \item
    How should we model the errors and the costs of construction and use of large-scale Hamiltonians for ballistic computations?
    \item
    Are there codes for which QRAM is transversal, or can otherwise be applied directly to logical states?
    \item
    Can we fix the scaling issues of the transmon-based bucket-brigade memory of \cite{THESIS:Cadellans15}?
    \item How much noise can fault-tolerant quantum linear algebra tolerate in the QRAM access, and could bucket-brigade QRAM reach that level?
\end{enumerate}

\vspace{1em}

{ 

\centering\textbf{\small ACKNOWLEDGEMENTS}
\vspace{1em}

}

We thank John Schanck for many fruitful discussions and providing many of these arguments; Simon C. Benjamin, Bal{\'i}nt Koczor, Sam McArdle, Shouvanik Chakrabarti, Dylan Herman, Yue Sun, Patrick Rebentrost, Matteo Votto, and an anonymous reviewer for helpful comments; and Romy Minko, Ryan Mann, Oliver Brown, and Christian Majenz for valiant attempts to generalize \Cref{thm:no-qracm-distillation}. S.J. was funded by the University of Oxford Clarendon fund, and acknowledges the support of the Natural Sciences and Engineering Research Council (NSERC), funding reference number RGPIN-2024-03996. A.G.R. was funded by a JPMorgan Chase Global Technology Applied Research PhD Fellowship and the SEEQA project (EP/Y004655/1).

\vspace{1em}

{

\centering\textbf{\small DISCLAIMER}

\vspace{1em}
}

This research was funded in part by JPMorgan Chase \& Co. Any views or opinions expressed herein are solely those of the authors listed, and may differ from the views and opinions expressed by JPMorgan Chase \& Co. or its affiliates. This material is not a product of the Research Department of J.P. Morgan Securities LLC. This material should not be construed as an individual recommendation for any particular client and is not intended as a recommendation of particular securities, financial instruments or strategies for a particular client. This material does not constitute a solicitation or offer in any jurisdiction.

\bibliographystyle{plainnat}
\bibliography{qram}

\appendix

\section{Linear Algebra Proofs}

\begin{lemma}[Optimality of Quantum Eigenvalue Transform]\label{lemma:quantum_polynomial_eigenvalue_transformation_lower_bound}
Any quantum algorithm implementing a degree-$k$ polynomial of a matrix $H$ (with the polynomial defined on the eigenvalues) requires $\Omega(k)$ queries to $O_H$, in general. 
\end{lemma}
\begin{proof}
We prove this lower-bound by a reduction to unstructured search, in the QSVT framework (as per \cite{gilyen2019quantum}), from which the bound immediately follows. We note that the result of \cite{gilyen2019quantum} almost immediately implies this bound, but we present the following derivation in a format to enhance the clarity of our discussion. First we formally define a block-encoding. As per \cite{gilyen2019quantum}, an $n+a$-qubit unitary matrix $U_A$ is called an $(\alpha, a, \epsilon)$-block encoding of the $n$-qubit matrix $A$ if $\lnorm{A  - \alpha (\bra{0}^{\otimes a}\otimes I_n) U_A (\ket{0}^{\otimes a}\otimes I_n)}_2 \le \epsilon$.

Given an unstructured search function $f(x) = 0$ for all $x\neq m$, and $f(m) = 1$ (with m being unique), specified by an $O_f$, oracle defined as $O_f\ket{x} = (-1)^{f(x) \oplus 1}\ket{x}$ (where we negate the sign for ease of analysis), we can construct a $(1, 9, 0)$-block-encoding  $U_A$ of $A := \frac{1}{\sqrt{N}}\op{m}{+^n}$, where $\ket{+^n} = H^{\otimes n}\ket{0}$. We give a very inefficient block encoding (in terms of the number of ancillas) for ease of exposition.
First, we construct a $(1, 4, 0)$-block-encoding of $H^{\otimes n}\op{0}{0}H^{\otimes n}$. This can first be done by constructing a $(1, 2, 0)$-block-encoding of $\op{0}{0} = \frac{1}{2}(I_n + (2\op{0}{0}  - I_n))$ (which can be done by applying the standard sum of block encoding lemma \cite{gilyen2019quantum}) of an $n$ qubit identity matrix with the standard $2\op{0}{0}  - I_n$ unitary from Grover search. We can then use the product of block encodings lemma with $H^{\otimes n}$, $\op{0}{0}$, and $H^{\otimes n}$ again to get a $(1, 4, 0)$-block encoding of $\op{+^n}{+^n}$. We can then use the product of block encoding with $I_1\otimes O_f$ and $\op{+^n}{+^n}$ to get a $(1, 5, 0)$-block encoding of $O_f\op{+^n}{+^n}$. We can then use the sum of block encodings with $\op{+^n}{+^n}$ and $O_f\op{+^n}{+^n}$ to get a $(1, 9, 0)$-block-encoding of $\frac{1}{2}((O_f + I)\op{+^n}{+^n}) = \frac{1}{\sqrt{N}}\op{m}{+^n}$, as desired.

We can apply the polynomial approximation of the sign function as per \cite{gilyen2019quantum} to our block encoding unitary $U_A$, giving a constant probability of success of measuring the marked state. They give the degree of this (odd degree) polynomial as $k = O(\log(1/\epsilon)/\delta)$, where $\epsilon$ is an upper-bound on the maximum deviation of the polynomial approximation to the sign function in the $x \in [-2, 2] \backslash (-\delta, \delta)$. If we set $\delta = \frac{1}{\sqrt{N}}$, then the sign function maps our non-zero singular value to $1$, while leaving the zero-valued singular values unchanged. As a result, we can measure the ancilla register in the $\ket{0}$ state with $\propto 1 - \epsilon$ probability, learning the result of the unstructured search problem. If it were possible to implement this degree $\tilde{O}(1/\delta)$ degree-polynomial with fewer than $\Omega(1/\delta)$ queries to $U_A$, we could solve unstructured search in general with fewer than $\Omega(\sqrt{N})$ queries to the unstructured search oracle, violating well-known lower-bounds for unstructured search. As a consequence, any general algorithm implementing a SVT of a degree-$k$ polynomial of some matrix $H$ must use at least $\Omega(k)$ queries to that matrix. 

A polynomial implementing the sign function applied to the eigenvalues of a block encoding of $\tilde{A}:= \begin{pmatrix} 0 & A\\A^{\dagger} & 0 \end{pmatrix}$ would  also solve unstructured search in the same way (noting that the non-zero eigenvalues of $\tilde{A}$ correspond to $\pm \frac{1}{\sqrt{N}}$ and have associated eigenvectors $\frac{1}{\sqrt{2}}(\ket{0}_1\ket{m} \pm \ket{1}_1\ket{+^n})$. Thus, applying $\text{Sign}(\tilde{A})$ leaves the $0$ eigenvalues unchanged, and maps the $\pm\frac{1}{\sqrt{N}}$ eigenvalues to $\pm 1$ (within $\epsilon$ distance). Consequently, $\text{Sign}(\tilde{A}) \approx \begin{pmatrix}0 & \op{m}{+^n}\\ \op{+^n}{m} & 0 \end{pmatrix}$. Applying $\text{Sign}(\tilde{A})$ to an initial state $\ket{1}\ket{+^n}$ then gives the solution to the unstructured search problem, and so the same bound holds for \cref{problem:polynomial_eigenvalue_transform}.
\end{proof}

\begin{lemma}\label{lemma:classical_matrix_vector_multiplication}
Given $P$ classical processors sharing a common memory of access time $O(1)$, a $d$-sparse matrix $A \in \mathbb{C}^{N\times N}$ and a vector $\bm{v} \in \mathbb{C}^N$, we can compute $A\bm{v}$ with $\tilde{O}(Nd/P)$ time complexity.
\end{lemma}
\begin{proof}
This is a standard result; see \cite{BOOK:GGKK2003}.
Assume we represent $A$ as $N$ lists of elements $(j,A_{ij})$, where $A_{ij}$ are the non-zero components in row $i$, and $\bm{v}$ as an $N$-element array. With fewer than $N$ processors, each processor will iterate through a row of $A$. For an entry $(j,A_{ij})$, it will look up the $j$th element of $\bm{v}$, multiply it by $A_{ij}$ and add that to a running total. This produces one element of $A\bm{v}$. Each element can be computed independently, so this parallelizes perfectly. With more than $N$ processors, rows will have multiple processors. The processors for row $i$ can divide the entries of $A$ in that row and compute sub-totals. To compute the full sum, they add their sub-totals together in a tree structure; this tree has depth $O(\log(P/N))$. 
\end{proof}

Following the techniques of \cite{WEB:Bernstein2001}, we use sorting to build matrix multiplication.

\begin{lemma}\label{lemma:matrix_multiplication_via_hypercube_sort}
    Given $P$ parallel processors forming a sorting network that can sort in time $\mathsf{S}$, multiplying an $N$-dimensional vector $\bm{v}$ by a $d$-sparse $N\times N$ matrix $A$ takes time $O(\mathsf{S}+Nd/P+\lg(d))$.
\end{lemma}
\begin{proof}
First the processors create $d$ copies of $\bm{v}$. Each processor will store a block of $A$ or $\bm{v}$ in local memory, so by treating $A$ as a block matrix, we suppose each processor has one entry $A_{ij}$ of $A$ or $\bm{v}_j$ of $\bm{v}$. We assume there is an efficiently-computable function $f_j$ for each column $j$, which injectively maps indices $k$ from $1$ to $d$ to the non-zero elements of that column. Each processor creates a tuple $(k,j,0,A_{f_j(k)j})$ for components of $A$, and $(k,j,\bm{v}_j)$ for components of $\bm{v}$ (each copy of $\bm{v}$ has a distinct value of $k$). They will then sort this data by the second component, then the first. Because each column of $A$ contains at most $d$ non-zero entries, this ensures each processor with a tuple $(k,j,0,A_{f_j(k)j})$ is adjacent to a tuple $(k,j,\bm{v}_j)$, so these two processors can compute $A_{f_j(k)j}\bm{v}_j$ and store the result in the tuple for $A$. The processors can then discard the tuples for $\bm{v}$, and change the first component from $k$ to $f_j(k)$. Then they re-sort so that the rows of $A$ are physically close (e.g., in a square of length $\sqrt{d}$ in 2 dimensions). They then add the entries of $A_{ij}\bm{v}_j$ for each row. Cascading these together will take at least $\lg(d)$ sequential additions, but it will be asymptotically less than the sorts of the full matrix and vector. This produces the $i^{th}$ element of $A\bm{v}$. Another sort can send these to whatever memory location represents the output.
\end{proof}

\begin{lemma}[Dense Matrix-Vector Multiplication with Local Memory]\label{lemma:dense_matrix_vector_multiplication_with_local_memory}
A set of $O(N^2 \log(N))$ processors arranged in a $3D$ grid and with nearly local connectivity, leading to a total wire-length of $\tilde{O}(N^2)$, can implement a matrix-vector multiplication with an $N\times N$ matrix $A$ and an $N$-dimensional vector $\bm{v}$ with $O(\log (N))$ complexity. 
\end{lemma}
\begin{proof}
 Consider a grid of $N\times N$ processors, each with local memory, and \textit{no} connections to their nearest neighbors. 
 Assume that each element in $A$ is assigned to a corresponding processor in the grid, i.e. the processor in the $i^{th}$ row and $j^{th}$ column stores element $A_{ij}$.
 If our initial vector $\bm{v}$ is stored in a set of $N$ data cells, using a stack of $\log(N)$ local processors, of dimension $2 \times N, 4\times N, ..., N\times N$ we can recursively spread the elements in the vector to construct an $N\times N$ matrix $V$ such that $V_{ij} = \bm{v}_j$. The first layer in this stack has 2 wires per layer, each of length $N / 2$, for a total wire-length of $O(N^2)$. The final layer in this stack, the one that connects element $V_{ij}$ to element $A_{ij}$ in the main $N\times N$ grid, has $O(N^2)$ wires, each of length $O(1)$ for a total wire length of $O(N^2)$. The middle layers in the stack similarly have a total wire length of $O(N^2)$ (summing progressively more shorter wires). Thus, the main $N\times N$ processor grid can store $A_{ij}\bm{v}_j$ with a total of $O(N^2 \log N)$ ancillary local processors with a total wire length of $O(N^2 \log N)$. We can then repeat this process of spreading out the elements of $\bm{v}$ in reverse, adding another stack of $\log N$ processors of size $N\times N, N \times N/2, ..., N\times 2, N\times 1$. We then add elements in adjacent cells, sending them to the cell in the layer above,  building up the sum of products of the appropriate matrix and vector element. This produces the output vector $A\bm{v}$, where the $i^{th}$ row stores $\sum_{j=1}^N A_{ij}\bm{v}_j$, with a total wire-length of $O(N^2 \log N)$, and with a total of $O(N^2 \log N)$ processors. 
\end{proof}

\section{Proof of \Cref{lem:num-operations}}
We start by defining the set of circuits $\mathcal{C}(W,D,G,g,k)$ to be all circuits on $W$ qubits, of circuit depth $D$, using $G$ gates from a set of size $g$ with fanin at most $k$. Here we assume all gates have depth 1 and we count circuits based on the arrangement of gates, not by function. Define a ``circuit'' as a function that takes as input a pair of (qubit,time step) and returns as output the particular gate applied to that qubit at that time step.

If two circuits have a different effect on quantum states, they must use different gates, and thus will be distinct circuits by this definition. This definition overcounts circuits, since equivalent circuits are counted twice, but we want to upper bound the number of circuits so this is fine.

\begin{proposition}\label{prop:high-fanin-reduce-many-single-fanin}
Under this notation, $\vert \mathcal{C}(W,D,G,g,k)\vert \leq \vert \mathcal{C}(W,D,\min\{kG,DW\},Wg,1)\vert$.
\end{proposition}
\begin{proof}
Let $U$ be an $\ell$-qubit gate in the gate set of $\mathcal{C}(W,D,g,k)$. We can define $\ell$ distinct single-qubit gates $U_1,\dots, U_\ell$ from $U$, expanding the total number of gates from $g$ to $gk$, and and then further define $W/\ell$ multiples of each of these single qubit gates, i.e., $U_{1}^{(1)},\dots, U_{\ell}^{(1)},\dots, U_{1}^{(m)},\dots, U_{\ell}^{(m)}$ where $m=\floor{W/\ell}$. This gives at most $W$ single-qubit gates from $U$, so the total number of single-qubit gates defined in this way is at most $Wg$.

We now define an injective function from $\mathcal{C}(W,D,G,g,k)$ to $\mathcal{C}(W,D,\min\{kG,DW\},Wg,1)$. For each circuit, at each timestep we take all $\ell$-qubit gates $U^{(1)},\dots U^{(m)}$ (we know there are fewer tham $m=\floor{W/\ell}$ gates of fanin $\ell$ in one timestep, because there aren't enough qubits for more!) and decompose each into $U_1^{(i)},\dots, U_\ell^{(i)}$, such that the qubit mapped to the $j$th input of the $i$th gate is given gate $U_j^{(i)}$. This increases the total number of gates by at most a factor of $k$, and cannot increase the number of gates above $DW$.

For example, a CNOT would split into two gates, a ``target'' gate and a ``control'' gate, and if these are reversed, the circuit is different (by our counting).

To show injectivity, suppose two circuits $C_1$ and $C_2$ in $\mathcal{C}(W,D,G,g,k)$ map to the same circuit in $\mathcal{C}(W,D,kG,gW,1)$. Let $U$ be any gate in $C_1$; suppose it applies to $\ell$ qubits $q_1,\dots, q_\ell$ in time step $t$. This means there is some $i$ such that the image circuit has $U_1^{(i)},\dots, U_\ell^{(i)}$ applied to $q_1,\dots, q_\ell$ in time step $t$. By construction of our function, this means there must be some gate in $C_2$ of the same type as $G$ which is also applied to those same qubits, in the same order, in time step $t$. Repeating for all gates shows that $C_1=C_2$.
\end{proof}

\begin{proposition}\label{prop:simple-circuit-count}
In the same notation, $\vert \mathcal{C}(W,D,G,g,1)\vert \leq \binom{DW}{G} g^G$.
\end{proposition}
\begin{proof}
We can straightforwardly count: There are $DW$ possible ``slots'' for each gate, depending on which qubit and which time step. We choose $G$ slots, and for each one, we select one of the $g$ gates.
\end{proof}

\begin{lemma}\label{lem:app:circuit-counts}
 $\vert \mathcal{C}(W,D,G,g,k)\vert \leq \binom{DW}{\min\{kG,DW\}} (Wg)^G$.
\end{lemma}
\begin{proof}
We combine Proposition~\ref{prop:high-fanin-reduce-many-single-fanin} and Proposition~\ref{prop:simple-circuit-count} to obtain the result.
\end{proof}

\section{Proof of No Distillation}
There are two components of the proof: First, we argue that with only $d$ inputs in the QRACM, we can always find some indices which are ``underrepresented'', i.e., very little of the amplitude of the input state is concentrated on those indices.

\begin{lemma}\label{lem:indistinguishable-tables}Let $\Hil_1,\dots, \Hil_d$ be any collection of finite-dimensional Hilbert spaces, and let $\Hil_{QRACM}$ be the space that $U_{QRACM}$ acts on for tables of size $N$.
Let $\rho_1,\dots, \rho_d$ be any collection of states, each in $B(\Hil_i\otimes \Hil_{QRACM})$, i.e., part of the state is in the input space for a QRACM gate. Then there exists some set of $\ell$ indices such that for any two tables $T$ and $T'$ differing only in those $\ell$ indices, if we define
\begin{equation}\label{eq:distillation-delta}
\delta_i := \left\Vert \mathcal{U}_{QRACM}(T)(\rho_i) - \mathcal{U}_{QRACM}(T')(\rho_i)\right\Vert_1
\end{equation}
(i.e., the trace distance, with an implicit identity channel on $\Hil_i$) then
\begin{equation}
\sum_{i=1}^d \delta_i \leq d\sqrt{\frac{2\ell}{N}}.
\end{equation}
\end{lemma}
This is more general than necessary: we only need $\ell=1$ to prove our main theorem.
\begin{proof}
Because we care about the application of a perfect QRACM gate, we will purify each of these states, so we consider some purification of $\rho_i$:
\begin{equation}
\ket{\psi_i} = \sum_{j=0}^{N-1} \sum_{b\in \{0,1\}} \alpha_{ijb}\ket{\psi_{ijb}}\ket{j}\ket{b}
\end{equation}

where we have expressed the input space to the QRACM gate in the computational basis. Here $\ket{\psi_{ijb}}$ is in $\Hil_i\otimes \Hil_p$, where $\Hil_p$ is the space necessary to purify $\rho_i$. 
We see that $U_{QRACM}(T) = \sum_{j,b} \ketbra{j}{j} \otimes \ketbra{b\oplus T[j]}{b}$, so that 
\begin{align}
&U_{QRACM}(T')^\dagger U_{QRACM}(T)\nonumber\\
&= \sum_{j:T'[j]=T[j]}\ketbra{j}{j}\otimes I_2 + \sum_{j:T'[j]\neq T[j]} \ketbra{j}{j}\otimes X
\end{align}
where $X$ is the Pauli $X$. This implies that 
\begin{align}
\bra{\psi_i}&U_{QRACM}(T')^\dagger U_{QRACM}(T) \ket{\psi_i} \\
= &\sum_{j:T[j]=T'[j],b\in \{0,1\}}\vert\alpha_{ijb}\vert^2 + \\
&\sum_{j:T[j]\neq T'[j]}\overline{\alpha_{ij0}}\alpha_{ij1}\braket{\psi_{ij0}}{\psi_{ij1}} + \overline{\alpha_{ij1}}\alpha_{ij0}\braket{\psi_{ij1}}{\psi_{ij0}}\\
\geq & 1 - \left(\sum_{j:T[j]\neq T'[j]}\vert \alpha_{ij0}\vert^2 + \vert \alpha_{ij1}\vert^2 \right.\nonumber\\
&\left.\phantom{\sum_{j:T[j]\neq T'[j]}}+\overline{\alpha_{ij0}}\alpha_{ij1} + \overline{\alpha_{ij1}}\alpha_{ij0}\right)\\
\geq & 1 - 2\sum_{j:T[j]\neq T'[j]}\vert \alpha_{ij0}\vert^2+ \vert \alpha_{ij1}\vert^2
\end{align}
Since $U_{QRACM}(T')$ does not act on the purify space, it commutes with tracing out the purifying space, and partial trace can only decrease trace distance, we have that the trace distance between the states after applying QRACM to the two different tables is at most
\begin{equation}
\delta_i\leq\sqrt{2\sum_{j:T[j]\neq T'[j]}\vert \alpha_{ij0}\vert^2+ \vert \alpha_{ij1}\vert^2}
\end{equation}
where $\delta_i$ is defined in \Cref{eq:distillation-delta}.

We then let 
\begin{equation}
m_j := \sum_{i=1}^d \vert \alpha_{ij0}\vert^2 + \vert \alpha_{ij1}\vert^2
\end{equation}
We know that $\sum_{j=0}^N m_j = d$. Thus, for any $\ell$, there exist a set $\mathcal{J}$ of  $\ell$ indices such that 
\begin{equation}
\sum_{j\in\mathcal{J}}m_j \leq \frac{\ell d}{N}.
\end{equation}
If we set $T$ and $T'$ to differ only on $\mathcal{J}$, we see that $\sum_{i=1}^d \delta_i^2 \leq 2\sum_{j\in\mathcal{J}}m_j\leq \frac{2\ell d}{ N}$, from which we see that 
\begin{equation}
\sum_{i=1}^d \delta_i \leq d\sqrt{\frac{2\ell}{N}}.
\end{equation}
\end{proof}

We now argue by saying that since so little of the input state is concentrated on certain indices, we can replace any table $T$ with a different table that differs on those indices, and this will be indistinguishable to the QRACM distillation-and-teleportation process. Since the process needs to handle all logical inputs, it must fail on that index.

\begin{theorem}\label{thm:no-qracm-distillation}
Suppose there is a QRACM state distillation-and-teleportation process, accessing tables of size $N=2^n$ and making at most $d$ calls to the physical QRACM. Then the minimal fidelity for this process is at most 
\begin{equation}
\frac{3}{4} + d\sqrt{\frac{2}{N}}.
\end{equation}
\end{theorem}
\begin{proof}
Our distillation process is defined by some initial state $\rho_1$, and $d$ channels $\Phi_i$, where we interleave access to the QRAM device with these channels. Then for each table $T$, we can then define states $\rho_i(T)$ such that
\begin{equation}
    \rho_{i+1}(T) = \Phi_{i+1}\circ (I\otimes \mathcal{Q}(T))\rho_i(T)
\end{equation}
(using $\mathcal{Q}(T)$ as a more compact notation for $\mathcal{U}_{QRACM}(T)$, the unitary channel which performs the QRACM access), with $\rho_1(T)$ defined to be $\rho_1$ for all tables $T$.

Let $T$ be any table. Let $\mathcal{J}$ be the set of $\ell$ indices implied by \Cref{lem:indistinguishable-tables} for these states, and let $T'$ be a table that differs from $T$ in those $\ell$ indices. We will show by induction that 
\begin{equation}
    \Vert \rho_i(T') - \rho_i(T)\Vert_1 \leq \sum_{j=1}^i \delta_j
\end{equation}
with $\delta_i$ as in \Cref{lem:indistinguishable-tables}. This holds for $i=1$ since all tables use the same starting state.

For induction, recall from the definition of the distillation process that 
\begin{equation}
    \rho_{i+1}(T) = \Phi_{i+1}\circ (I\otimes \mathcal{Q}(T)) \rho_i(T)
\end{equation}
(similarly for $T'$) where $\mathcal{Q}(T)=\mathcal{U}_{QRACM}(T)$ (a more compact notation), and $\Phi_i$ is the channel of any other computations performed during distillation between subsequent QRACM accesses. Since the channel  $\Phi_{i+1}$ can only reduce the trace distance between then by the data processing inequality, we obtain:
\begin{align}
\Vert\rho_{i+1}&(T') - \rho_{i+1}(T)\Vert_1 \\
\leq &\Vert (I\otimes \mathcal{Q}(T'))\rho_i(T') - (I\otimes \mathcal{Q})(T)\rho_i(T)\Vert_1 \\
\leq & \Vert ((I\otimes \mathcal{Q}(T) - (I\otimes \mathcal{Q}(T'))\rho_i(T)\Vert_1\\
&+ \Vert (I\otimes \mathcal{Q}(T'))(\rho_i(T') - \rho_i(T))\Vert_1 \\
\leq & \delta_{i+1}+ \sum_{j=1}^i \delta_j
\end{align}
using \Cref{lem:indistinguishable-tables}. 

We finally note that $\rho_d(T) = \rho_{distill}(T)$, the state input into the final teleportation process. This means the distance between $\rho_{distill}(T)$ and $\rho_{distill}(T')$ is at most $d\sqrt{\frac{2\ell}{N}}$ by \Cref{lem:indistinguishable-tables}. Again by the data processing inequality, the final teleportation process cannot increase the distance. Denoting the full process as $\tilde{\mathcal{U}}_{QRACM}(T)$, we see that for any logical input $\rho_L$, we obtain
\begin{equation}
\Vert (\tilde{\mathcal{U}}_{QRACM}(T) - \tilde{\mathcal{U}}_{QRACM}(T'))\rho_L\Vert_1  \leq d\sqrt{\frac{2\ell}{N}}
\end{equation}

However, let $\rho_L=\ketbra{j}{j}\otimes\ketbra{0}{0}$ for some $j$ with $T[j]\neq T'[j]$. Then we have for the \emph{ideal} QRACM access,  
\begin{equation}
\Vert (\mathcal{U}_{QRACM}(T) - \mathcal{U}_{QRACM}(T'))(\rho_L)\Vert_1 = 1
\end{equation}
But by the triangle inequality,
\begin{align}
1=&\Vert (\mathcal{U}_{QRACM}(T) - \mathcal{U}_{QRACM}(T'))(\rho_L)\Vert_1\\
 \leq & \Vert (\mathcal{U}_{QRACM}(T) - \tilde{\mathcal{U}}_{QRACM}(T))\rho_L\Vert_1\\
 & + \Vert (\mathcal{U}_{QRACM}(T') - \tilde{\mathcal{U}}_{QRACM}(T'))\rho_L\Vert_1 \nonumber\\
 &+ \Vert (\tilde{\mathcal{U}}_{QRACM}(T) - \tilde{\mathcal{U}}_{QRACM}(T'))\rho_L\Vert_1\nonumber\\
1 - d\sqrt{\frac{2\ell}{N}}\leq & \Vert (\mathcal{U}_{QRACM}(T) - \tilde{\mathcal{U}}_{QRACM}(T))\rho_L\Vert_1\\
& + \Vert (\mathcal{U}_{QRACM}(T') - \tilde{\mathcal{U}}_{QRACM}(T'))\rho_L\Vert_1\nonumber
\end{align}
meaning (WLOG) that the distance between the realized QRACM channel and the ideal QRACM channel, for state $\rho_L$ on access to table $T$ or $T'$, is at least $\frac{1}{2}-d\sqrt{\frac{\ell}{2N}}$. Using known relations between trace distance and fidelity and setting $\ell=1$ gives the result. 
\end{proof}

\end{document}